\documentclass[11pt, letterpaper, onecolumn]{IEEEtran}

\usepackage[mathscr]{eucal}
\usepackage[cmex10]{amsmath}
\usepackage{epsfig,epsf,psfrag}
\usepackage{amssymb,amsmath,amsthm,amsfonts,latexsym}
\usepackage{amsmath,graphicx,bm,xcolor,url}
\usepackage[caption=false]{subfig} 
\usepackage{fixltx2e}
\usepackage{array}
\usepackage{verbatim}
\usepackage{bm}
\usepackage{verbatim}
\usepackage{textcomp}
\usepackage{mathrsfs}
\usepackage{epstopdf}

\catcode`~=11 \def\UrlSpecials{\do\~{\kern -.15em\lower .7ex\hbox{~}\kern .04em}} \catcode`~=13 

\allowdisplaybreaks[1]
 
\newcommand{\nn}{\nonumber}

\newcommand{\calA}{\mathcal{A}}
\newcommand{\calB}{\mathcal{B}}

\newcommand{\calE}{\mathcal{E}}

\newcommand{\calM}{\mathcal{M}}
\newcommand{\calN}{\mathcal{N}}

\newcommand{\calP}{\mathcal{P}}

\newcommand{\calT}{\mathcal{T}}

\newcommand{\calX}{\mathcal{X}}

\newcommand{\calZ}{\mathcal{Z}}

\newcommand{\ba}{\mathbf{a}}

\newcommand{\be}{\mathbf{e}}

\newcommand{\rmd}{\mathrm{d}}

\newcommand{\rme}{\mathrm{e}}

\newcommand{\bbE}{\mathbb{E}}

\newcommand{\bbN}{\mathbb{N}}

\newcommand{\bbR}{\mathbb{R}}

\DeclareMathAlphabet{\mathbsf}{OT1}{cmss}{bx}{n}
\DeclareMathAlphabet{\mathssf}{OT1}{cmss}{m}{sl}

\DeclareSymbolFont{bsfletters}{OT1}{cmss}{bx}{n}  
\DeclareSymbolFont{ssfletters}{OT1}{cmss}{m}{n}
\DeclareMathSymbol{\bsfGamma}{0}{bsfletters}{'000}
\DeclareMathSymbol{\ssfGamma}{0}{ssfletters}{'000}
\DeclareMathSymbol{\bsfDelta}{0}{bsfletters}{'001}
\DeclareMathSymbol{\ssfDelta}{0}{ssfletters}{'001}
\DeclareMathSymbol{\bsfTheta}{0}{bsfletters}{'002}
\DeclareMathSymbol{\ssfTheta}{0}{ssfletters}{'002}
\DeclareMathSymbol{\bsfLambda}{0}{bsfletters}{'003}
\DeclareMathSymbol{\ssfLambda}{0}{ssfletters}{'003}
\DeclareMathSymbol{\bsfXi}{0}{bsfletters}{'004}
\DeclareMathSymbol{\ssfXi}{0}{ssfletters}{'004}
\DeclareMathSymbol{\bsfPi}{0}{bsfletters}{'005}
\DeclareMathSymbol{\ssfPi}{0}{ssfletters}{'005}
\DeclareMathSymbol{\bsfSigma}{0}{bsfletters}{'006}
\DeclareMathSymbol{\ssfSigma}{0}{ssfletters}{'006}
\DeclareMathSymbol{\bsfUpsilon}{0}{bsfletters}{'007}
\DeclareMathSymbol{\ssfUpsilon}{0}{ssfletters}{'007}
\DeclareMathSymbol{\bsfPhi}{0}{bsfletters}{'010}
\DeclareMathSymbol{\ssfPhi}{0}{ssfletters}{'010}
\DeclareMathSymbol{\bsfPsi}{0}{bsfletters}{'011}
\DeclareMathSymbol{\ssfPsi}{0}{ssfletters}{'011}
\DeclareMathSymbol{\bsfOmega}{0}{bsfletters}{'012}
\DeclareMathSymbol{\ssfOmega}{0}{ssfletters}{'012}

\newcommand{\tilD}{\tilde{D}}

\newcommand{\tilf}{\tilde{f}}

\newcommand{\hatR}{\hat{R}}

\newcommand{\eps}{\varepsilon}

\def\fndot{\, \cdot \,}

\newcommand{\dotgeq}{\stackrel{.}{\geq}}

\DeclareMathOperator*{\argmin}{arg\,min}

\DeclareMathOperator{\var}{\mathsf{Var}}

\newtheorem{theorem}{Theorem} 
\newtheorem{lemma}{Lemma}

\newtheorem{corollary}{Corollary}
\newtheorem{definition}{Definition}

\newcommand{\qednew}{\nobreak \ifvmode \relax \else
      \ifdim\lastskip<1.5em \hskip-\lastskip
      \hskip1.5em plus0em minus0.5em \fi \nobreak
      \vrule height0.75em width0.5em depth0.25em\fi}

\usepackage{cite,overpic}

\newcommand{\mix}{\mathrm{mix}}

\begin{document}
\flushbottom
\allowdisplaybreaks[1]

\title{Equivocations, Exponents and Second-Order Coding Rates under Various  R\'enyi Information Measures} 
\author{Masahito Hayashi$^\dagger$,~\IEEEmembership{Senior Member,~IEEE}  $\,$ and  $\,$  Vincent Y.~F.\ Tan$^\ddagger$,~\IEEEmembership{Senior Member,~IEEE} 
\thanks{$^\dagger$ M.~Hayashi is with the  Graduate School of Mathematics, Nagoya University, and the Centre for Quantum Technologies (CQT),  National University of Singapore   (Email: masahito@math.nagoya-u.ac.jp).   }  \thanks{$\ddagger$ V.~.Y.~F. Tan is with the Department of Electrical and Computer Engineering and the Department of Mathematics, National University of Singapore (Email:  vtan@nus.edu.sg).   } \thanks{This paper was presented in part at the 2015 International Symposium on Information Theory in Hong Kong. }}


\maketitle

\begin{abstract}
We evaluate the asymptotics of equivocations,  their exponents as well as their second-order coding rates under various R\'enyi information measures. Specifically, we consider the effect of applying a hash function on a source and we quantify the level of non-uniformity and dependence of the compressed source from another correlated source when the number of copies of the sources is large.   Unlike previous works that use   Shannon information measures   to quantify  randomness, information or uniformity, we define our security measures in terms of a more general class of information measures---the R\'enyi  information measures and their Gallager-type counterparts. A special case of these R\'enyi information measure is  the class of  Shannon information measures. We prove tight asymptotic results for the security measures and their exponential rates of decay. We also prove bounds on the second-order asymptotics and show that these bounds match when the magnitudes of the second-order coding rates are large. We do so by    establishing  new classes non-asymptotic bounds on the equivocation and evaluating these bounds using various probabilistic limit theorems asymptotically. 
\end{abstract}

\begin{IEEEkeywords}
 Information-theoretic security,  Equivocation,  Conditional R\'enyi entropies, R\'enyi divergence, Sibson's mutual information, Arimoto's mutual information, Error exponents, Secrecy Exponents, Second-order coding rates
\end{IEEEkeywords}

\section{Introduction}  \label{sec:intro}

Consider the situation where we are given $n$ independent and identically distributed (i.i.d.) copies of a joint source $(A^n,E^n)$. One of the central tasks in information-theoretic security is to understand the effect of applying a   hash function~\cite{carter79} (binning operator) $f$ on $A^n$. This hash function is used to ensure that the compressed source $f(A^n)$ is almost uniform on its  alphabet and also almost independent of another discrete memoryless source $E^n$. Mathematically, we want to understand the deviation of $f(A^n)\in\{1,\ldots, \lceil \rme^{nR}\rceil\}$ from the uniform distribution on the same support $P_{\mix, f(\calA^n)}$ and the level of remaining dependence between $f(A^n)$ and a correlated source $E^n$. These two criteria can be described by equivocation measures.   Traditionally in information-theoretic security~\cite{liang_book,  Bloch_book}, equivocation is measured in terms of the Shannon-type quantities such as the Shannon entropy, relative entropy (Kullback-Leibler divergence), and mutual information. In particular, it is common to design $f$ such that  the following is small for any rate $R$:
\begin{equation}
 D( P_{f(A^n) , E^n  }\| P_{\mix, f(\calA^n)}  \times   P_{E^n})= nR -   H(f(A^n) | E^n  ). \label{eqn:equiv_1}
\end{equation}
Clearly if the above quantity is small in some   sense, the message $f(A^n)$ is close to uniform and almost independent of $E^n$, two desirable traits of a hash function for  security applications.

\subsection{Motivations}\label{sec:motivation}
A novel feature of this paper is that we depart from using Shannon information measures to quantify randomness and independence. It is known that the Shannon entropy $H$ or the relative entropy $D$ are special cases of a   larger family of information measures known as R\'enyi information measures, denoted as  $H_{1+s}$ and $D_{1+s}$ for $s\in \bbR$. Thus, as expounded by Iwamoto and Shikata~\cite{iwamoto}, we can quantify equivocation using these measures, gaining deeper insights into the fundamental limits of information leakage under the effect of hash functions. There may also be a possibility of  the optimal key generation rate changing when we use alternative information measures. In addition, in the study of cryptography and quantum key distribution (QKD), the R\'enyi entropy of order $2$~\cite{Beigi14} (or collision entropy) $H_2(A|P_A):=-\log\sum_{a\in\calA}P_A(a)^2$ and the min-entropy $H_{\min}(A|P_A):=-\log \max_{a\in\calA}P_A(a)$ play  important roles in quantifying randomness. A case in point is the {\em leftover hash lemma}~\cite{Impagliazzo, BBCM, hastad}. Another motivation stems from the recent study of {\em overcoming weak expectations} by Dodis and Yu~\cite{dodis13} where   cryptographic primitives are based on weak secrets, in which the only information about the secret is some fraction of min-entropy. The authors in~\cite{dodis13} provided bounds on the weak expectation $\bbE f(Y)$ of some function $f$ of a random variable $Y$ in terms of the min-entropy and the R\'enyi entropy of order $2$. In a follow-on paper by Yao and Li~\cite{yao14}, this study was generalized to   R\'enyi entropies of   general orders. Finally, in the study of {\em secure authentication codes} (or {\em A-codes} in short), which is one of the most fundamental cryptographic protocols in information-theoretic cryptography, Shikata~\cite{shikata15} quantified lengths of secret keys in terms of R\'enyi entropies of   general orders.  Motivated by these studies, the authors opine that it is of interest to study the performance of hashing under these generalized families of entropies (generalized uncertainty measures) and divergences (generalized distance measures).

\subsection{Main Contributions}
We consider three asymptotic settings---the asymptotics of  R\'enyi-type security measures, its exponential decay and a certain second-order behavior. 

\begin{enumerate}

\item  First, we characterize the asymptotic behavior of  the security measure
\begin{equation}
D_{1+s}( P_{f (A^n) , E^n   } \| P_{\mix, f (\calA^n)}\times P_{E^n} )\label{eqn:equiv_2} 
\end{equation}
for a fixed rate $R=\frac{1}{n}\log \|f\|$ where $\|f \|:=|f (\calA^n)|$ is the cardinality of the range of  a hash function $f$. The function $f$ will be taken to be a {\em random} hash function as we will explain and motivate later. Further, as we shall see in Section~\ref{sec:security_meas}, the quantity in \eqref{eqn:equiv_2} is closely related to the equivocation~\cite{Wyn75}. In Section~\ref{sec:equiv} (particularly in Corollary~\ref{cor:key} therein), we show  that if we measure security using $D_{1+s}$ with $s>0$, the fundamental limits of key generation rates change relative to  those for  traditional  Shannon-type measures   $D_1$.  The security measure in \eqref{eqn:equiv_2} quantifies the deviation of the hashed or compressed random variable $f (A^n)$ from the uniform distribution and also its remaining dependence from a correlated random variable $E^n$.


\item We are also interested in the speed of the exponential decay of \eqref{eqn:equiv_2} given a fixed rate $R$. That is, we are interested in   the asymptotic behavior of 
\begin{equation}
\frac{1}{n}\log D_{1+s}( P_{f (A^n) , E^n   } \| P_{\mix, f (\calA^n)}\times P_{E^n}  )\label{eqn:equiv_3} .
\end{equation}
This is likened to {\em error exponent} or {\em reliability function} analysis in classical information theory~\cite{gallagerIT,Csi97}.    We study this in Section~\ref{sec:exponent}.

\item Finally,  in Section~\ref{sec:2nd}, we also study  the second-order asymptotics~\cite{Strassen, Hayashi08} of the decay of $D_{1+s}$ with the blocklength, i.e., the   asymptotic behavior of 
\begin{align}
\frac{1}{\sqrt{n}} &D_{1+s}( P_{f (A^n) , E^n   } \| P_{\mix, f (\calA^n)}\times P_{E^n} ), \qquad\mbox{and} \label{eqn:sec_order1} \\
\frac{1}{\sqrt{n}} \log &D_{1+s}( P_{f (A^n) , E^n   } \| P_{\mix, f (\calA^n)}\times P_{E^n} )\label{eqn:sec_order2} .
\end{align}
where the number of compressed symbols (size of the hash function) $ \|f \|$ equals $\rme^{nR+\sqrt{n}L}$ for some first-order rate $R$ (usually the conditional R\'enyi entropy) and second-order rate $L\in\bbR$.  For some cases (R\'enyi parameter less than one) where we cannot exactly determine the  tight  second-order asymptotics (i.e., the upper and lower bounds do not match), we study the asymptotic behavior of~\eqref{eqn:equiv_2} when the second-order rate  $L$ tends to $+\infty$ or $-\infty$. In this case, the upper and lower bounds match up to and including a term quadratic in $L$.  
\end{enumerate}

As we mentioned  earlier, we will regard  $f$ as a  random hash function in the sequel. That is, it is randomly selected depending  on a random variable $X_n \in\calX_n$  that is available to all parties and is also independent of all other random variables. This random variable has distribution $P_{X_n}$. See Fig.~\ref{fig:AE}. To further elaborate, instead of the the R\'enyi divergences in \eqref{eqn:equiv_2}--\eqref{eqn:sec_order2},  for the purposes of asserting the existence of a particular function $f$ with some desired properties (cf.\ the random selection argument), we consider the quantity 
\begin{equation}
D_{1+s}^{(n)}:=D_{1+s}( P_{f_{X_n} (A^n) , E^n ,X_n  } \| P_{\mix, f_{X_n}(\calA^n)}\times P_{E^n}\times P_{X_n} ).\label{eqn:equivo_rand}
\end{equation}
Here, we note that $f_{X_n}$ is a  random hash function (to be defined precisely in Definition~\ref{def:has}) and $\| f_{X_n}(\calA^n)\|$ is a constant random variable, i.e., it does not depend on the realization of $X_n$. Even though $D_{1+s}^{(n)}$ in \eqref{eqn:equivo_rand} is not an expectation of any quantity of interest, $\exp\big( (1+s) D_{1+s}^{(n)} \big)$ is the expectation of 
\begin{equation}
\tilD_{1+s}^{(n)}(x_n):=\exp\big((1+s)D_{1+s}( P_{f_{x_n} (A^n) , E^n   } \| P_{\mix, f_{x_n}(\calA^n)}\times P_{E^n} ) \big),\label{eqn:equivo_rand2}
\end{equation}
 where the probability of observing $x_n$ is  $P_{X_n}(x_n)$. Thus by a random selection argument, if the former is less than $\eps>0$, there exist an $x_n^* \in\calX_n$, indexing a  deterministic protocol  $f_{ x_n^*}$,  such that $\tilD_{1+s}^{(n)}(x_n^*)$ is also less than $\eps$.  
When $s=0$,   the expectation of quantities in~\eqref{eqn:equiv_2}--\eqref{eqn:sec_order2}  under the common randomness $X_n$ generating a universal$_2$ hash function $f_{X_n}(\cdot)$ is equivalent to the quantity in~\eqref{eqn:equivo_rand} but for $s\ne 0$, they are, in general, different. In the sequel, we adopt the latter criterion in \eqref{eqn:equivo_rand} to simplify the presentation of the results.  

 
We believe the results contained herein may serve as logical starting points to derive {\em tight} exponential error bounds and second-order coding rates for the wiretap channel~\cite{Wyn75} (as was done in~\cite{Hayashi11,Hayashi13}) and other information-theoretic security problems such as the secret key agreement~\cite{AC93} (as was done in~\cite{chou12,chou15}) problem. The leakage rates for these problems may be measured using traditional Shannon information  measures or   R\'enyi  information  measures (or their Gallager-type counterparts). Here,  we are only concerned with the secrecy requirement rather than both the secrecy and reliability requirements of the wiretap problem. The reliability requirement can be handled using, by now, standard error exponent analyses~\cite{Csi97,gallagerIT}.

\begin{figure}
\centering
\begin{picture}(115,60)
\put(0,55){\mbox{$A^n$}}
\put(50,55){\mbox{$E^n$}}
\put(91,55){\mbox{$X_n\sim P_{X_n}$}}
\put(15,57){\line(1,0){32}}
\put(4,52){\vector(0,-1){35}}
\put(54,52){\vector(0,-1){35}}
\put(104,52){\vector(0,-1){35}}
\put(7,35){\mbox{$f_{X_n}(\cdot)$}}
\put(-21,6){\mbox{$f_{X_n}(A^n)$}}
\put(50,6){\mbox{$E^n$}}
\put(100,6){\mbox{$X_n$}}
\multiput(22,10)(8,0){3}{\line(1,0){4}}
\multiput(4,-4)(8,0){13}{\line(1,0){4}}
\put(4,3){\line(0,-1){7}}
\put(104,3){\line(0,-1){7}}
\end{picture}
\caption{Illustration of applying a hash function $f_{X_n}$   on the source $A^n$. Common randomness $X_n$, independent of a correlated source $E^n$, is available to all parties and it determines the hash function $f_{X_n}$.  We would like $f_{X_n}(A^n)$ to be uniform on its support $\{1,\ldots,\|f_{X_n}\|\}$ and almost independent of $E^n$ in the sense of ensuring that quantity in~\eqref{eqn:equiv_2} is small. We examine  \eqref{eqn:equiv_2} under different asymptotic regimes such as the equivocation, the exponential behavior~\eqref{eqn:equiv_3}, and the second-order asymptotics~\eqref{eqn:sec_order1}--\eqref{eqn:sec_order2}. }
\label{fig:AE}
\end{figure}
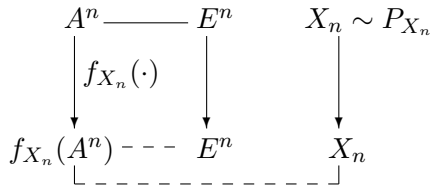

\subsection{Related Works}
In~\cite{Hayashi06,Hayashi11}, Hayashi generalized and strengthened the seminal privacy amplification analyses of Bennett {\em et al.} \cite{BBCM}, Renner~\cite{Rennerthesis} and Renner and Wolf \cite{Renner05} to obtain exponential error bounds for the leakage rate of the discrete memoryless wiretap channel  and the secrey key agreement problems~\cite{liang_book,  Bloch_book, Wyn75,AC93}. The leakage rate   was measured by the mutual information $I(A\wedge E|P_{AE})$ and the variational (or trace) distance $\| P_{AE}- P_A\times P_E\|_1$. The exact exponents for the   variational distance are, by now, well known \cite{Hayashi13, hay14}. 

 However, the results concerning the exponential decay of the leakage  rate  quantified via the mutual information contained in~\cite{Hayashi11}, and further generalized to other setings in \cite{chou12,chou15}, are only {\em achievability} results (i.e., lower bounds on the exponents). The {\em converse} has been open for some time. The present contribution, though not focusing on the wiretap channel or any specific information-theoretic security problem, derives {\em tight} exponential  bounds for  a generalization of the Shannon information measures, namely the family of R\'enyi information measures. In the process, we obtain a tight result for the exponential leakage rate for the mutual information,  thus resolving the converse part that was open in \cite{Hayashi11,chou12,chou15}. As a by-product, for some range of the R\'enyi parameter, we also obtain tight exponents for security measures defined using  the R\'enyi divergence under various hash functions.  
 
    Hayashi and Tsurumaru~\cite{HayashiT2013} proposed an efficient construction of hash functions for the purpose  of privacy amplification with less random seeds, thus potentially realizing the system in Fig.~\ref{fig:AE} with less random resources.  Other works along the  lines of deriving exponential error bounds for information-theoretic security problems include those by Hou and Kramer~\cite{hou13,hou13a}, Pierrot and Bloch~\cite{Pierrot13}, Bloch and Laneman~\cite{Bloch},   Han {\em et al.}~\cite{Han2013} and Parizi and Telatar~\cite{parizi15}. After the present work was submitted, Parizi,  Telatar and Merhav~\cite{Parizi16} proved ensemble tight exponential error bounds for the wiretap channel by appealing to type counting methods and channel resolvability arguments.

\subsection{Paper Organization}
The rest of the paper is organized as follows: In Section~\ref{sec:prelims}, we state the relevant preliminaries and define relevant information measures and security criteria for understanding the rest of the paper. In Section \ref{sec:equiv}, we state our results for the asymptotics of the equivocation. In Section \ref{sec:exponent}, we state our results for the exponential behavior of the R\'enyi-type security criteria. In Section \ref{sec:2nd}, we state our results for the second-order asymptotics  of the equivocation. We also consider the case where the magnitudes of the second-order rates are large. These are proved using novel one-shot bounds which are stated in Section~\ref{sec:one-shot}. The proofs of the asymptotic results are provided in Section~\ref{sec:prfs_asymp}. We conclude the paper in Section~\ref{sec:concl} by summarizing our key contributions and stating avenues for further investigations. The proofs of the one-shot bounds are rather technical and are thus relegated to the Appendices.

\section{Preliminaries and Information Measures}\label{sec:prelims}
\subsection{Basic Shannon and R\'enyi Information Quantities} \label{sec:info_measures}
We now introduce some information measures that generalize Shannon's information measures. Fix a normalized distribution $P_A \in\calP(\calA)$ and a non-negative measure (a non-negative vector but not necessarily summing to one) $Q_A\in\bar{\calP}(\calA)$ supported on a finite set $\calA$. Then the {\em relative entropy} and the {\em R\'enyi divergence of order $1+s$} are respectively defined as 
\begin{align}
D(P_A  \| Q_A)  &:= \sum_{a\in\calA}P_A(a) \log\frac{P_A(a)}{Q_A(a)} \\*
D_{1+s}(P_A  \| Q_A)  &:=  \frac{1}{s}\log\sum_{a\in\calA} P_A(a)^{1+s}Q_A(a)^{-s},
\end{align}
where throughout, $\log $ is to the natural base $\rme$.  
It is known that $\lim_{s\to 0} D_{1+s}(P_A\| Q_A) = D(P_A\| Q_A)$ so a special case of the  R\'enyi divergence is the usual relative entropy. It is known that  the map $s\mapsto s D_{1+s}(P_A\| Q_A)$ is concave in $s \in\bbR$ and hence $D_{1+s}(P_A\| Q_A)$ is monotonically increasing for $s\in\bbR$. Furthermore, the following {\em data processing or information processing inequalities} for R\'enyi divergences hold for $s\in [-1,1]$,
\begin{align}
D(P_A W  \| Q_A W) &\le D(P_A  \| Q_A)\\
D_{1+s}(P_A W  \| Q_A W) &\le D_{1+s}(P_A  \| Q_A) . \label{eqn:dpi_rd}
\end{align}
Here $W:\calA\to\calB$ is any stochastic matrix (channel) and $P_AW(b):=\sum_a W(b|a)P_A(a)$ is the output distribution induced by $W$ and $P_A$.

We   use $P_{\mix,\calA}$ to denote the uniform distribution on  $\calA$.  
 We also introduce conditional entropies on the joint alphabet  $\calA\times\calE$. 
If $P_{AE}$ is a  distribution  on $\calA\times\calE$, the {\em conditional entropy} and the {\em conditional R\'enyi entropy of order $1 + s$    relative to another normalized distribution $Q_E$ on $\calE$} as
\begin{align}
H(A|E|P_{AE}\| Q_E)  &:= - D(P_{AE} \| I_{A}\times Q_E) , \label{eqn:cond_entr_given} \\
H_{1+s}(A|E|P_{AE}\| Q_E)  &:= - D_{1+s}(P_{AE} \| I_{A}\times Q_E) . 
\end{align}
Here $I_A(a)=1$ for each $a\in\calA$ and it is known that $\lim_{s\to 0 } H_{1+s}(A|E|P_{AE}\| Q_E)=H(A|E|P_{AE}\| Q_E)$. If $Q_E=P_E$, we simplify the notation and denote the {\em conditional entropy}  and the {\em conditional R\'enyi entropy of order $1 + s$} as
\begin{align}
H(A|E|P_{AE}) &:= H(A|E|P_{AE}\|P_E) =-\sum_e P_E(e)\sum_a P_{A|E}(a|e) \log P_{A|E}(a|e)\\
H_{1+s}(A|E|P_{AE}) &:= H_{1+s}(A|E|P_{AE}\|P_E) = -\frac{1}{s}\log\sum_{e} P_E(e)\sum_a P_{A|E}(a|e)^{1+s}. \label{eqn:renyi_ent_Q} 
\end{align}
The function $s\mapsto s H_{1+s}(A|E|P_{AE})$ is concave, and  $H_{1+s}(A|E|P_{AE}\|Q_E)$ is 
monotonically decreasing  on $(0,\infty)$  and $(-\infty,   0)$.

We are also interested in the  so-called {\em Gallager form} of the conditional R\'enyi entropy  for a  joint distribution $P_{AE} \in\calP(\calA\times\calE)$:
\begin{align}
H_{1+s}^{\uparrow} (A|E|P_{AE})  :=  -\frac{1 +  s}{s}\log\sum_e \bigg( \sum_a P_{AE}(a,e)^{1+s}\bigg)^{\frac{1}{1+s}} .\label{eqn:gallager_form} %
\end{align}
By defining the familiar  {\em Gallager function} \cite{gallagerIT,Gal76} (parametrized slightly differently)
\begin{equation}
\phi(s|A|E|P_{AE}):=\log\sum_e \bigg( \sum_a P_{AE}(a,e)^{ \frac{ 1}{1-s}}\bigg)^{1-s}
\end{equation}
we can express \eqref{eqn:gallager_form} as 
\begin{equation}
H_{1+s}^{\uparrow} (A|E|P_{AE})=-\frac{1+s}{s}\phi\bigg( \frac{s}{1+s}\Big|A |E |P_{AE}\bigg),
\end{equation}
thus (loosely) justifying the nomenclature ``Gallager form'' of the conditional R\'enyi entropy in \eqref{eqn:gallager_form}. The quantities $H_{1+s}$ and $H_{1+s}^{\uparrow}$ can be   shown to be related as follows:
\begin{align}
\max_{Q_E \in \calP(\calE)} H_{1+s}(A|E|P_{AE}\| Q_E) = H_{1+s}^{\uparrow}(A|E|P_{AE}) \label{eqn:ent_min}
\end{align}
for $s\in [-1,\infty)\setminus\{0\}$. The maximum on the left-hand-side is attained for the tilted distribution
\begin{equation}
Q_E(e) = \frac{ (\sum_a P_{AE}(a,e)^{1+s})^{\frac{1}{1+s}} }{\sum_e(\sum_a P_{AE}(a,e)^{1+s})^{\frac{1}{1+s}}}. \label{eqn:Q_tilt}
\end{equation}
The map $s\to s H_{1+s}^{\uparrow}(A|E|P_{AE})$  is concave and the map $s\mapsto H_{1+s}^{\uparrow}(A|E|P_{AE})$ is monotonically decreasing for $s\in (-1,\infty)$. It can be shown by L'H\^{o}pital's rule that 
\begin{equation}
\lim_{s\to 0}  H_{1+s}^\uparrow(A|E|P_{AE}) = H(A|E|P_{AE}). \label{eqn:limit_s0}
\end{equation}
Thus, we regard $ H_{1}^\uparrow(A|E|P_{AE} )$ as $H(A|E|P_{AE})$, i.e., for R\'enyi parameter $\alpha=1+s=1$, the conditional R\'enyi entropy and its Gallager form coincide.    We also find it useful to consider a {\em two-parameter family} of the conditional  R\'enyi entropy:
 \begin{equation}
 H_{1+s| 1+t}(A|E|P_{AE}) := -\frac{1+t}{s}\log\sum_{e} P_E(e)\bigg(\sum_a P_{A|E}(a|e)^{1+s} \bigg)^{\frac{1}{1+t}}. \label{eqn:two_param}
 \end{equation}
 Clearly, 
\begin{equation}
 H_{1+s| 1+s }(A|E|P_{AE})=H_{1+s}^{\uparrow} (A|E|P_{AE})\label{eqn:two_to_one}
 \end{equation} 
 so two-parameter conditional  R\'enyi entropy is  a generalization of the Gallager form  of the conditional R\'enyi entropy in \eqref{eqn:gallager_form}. 
 
 For a fixed joint source $P_{AE}$ define
\begin{align}
\hatR_s &:=\frac{\rmd }{\rmd t} \, tH_{1+t}(A|E|P_{AE})\Big|_{t=s } , \quad\mbox{and}\label{eqn:crit_rate1} \\
\hatR_s^\uparrow  &:=\frac{\rmd }{\rmd t}\, tH_{1 + t }^{\uparrow}(A|E|P_{AE})\Big|_{t=s}. \label{eqn:crit_rate2}
\end{align}
We note that $\hatR_s$ and $\hatR_s^\uparrow$ are  monotonically non-increasing in $s$ because the functions $t\mapsto tH_{1+t}(A|E|P_{AE})$ and $t\mapsto tH_{1+t}^\uparrow (A|E|P_{AE})$  are concave.  The fact that $t\mapsto -tH_{1+t}^\uparrow(A|E|P_{AE})$ is convex is because the maximum of convex functions is convex; cf.~\eqref{eqn:ent_min}.  Furthermore, both $\hatR_s$ and $\hatR_s^\uparrow$ are non-negative by direct evaluation of the derivatives and noting that $\log P_{A|E}(a|e)\le 0$.  We assume, henceforth, that the source $P_{AE}$ satisfies the conditions that $t\mapsto tH_{1+t}(A|E|P_{AE})$ and $t\mapsto tH_{1+t}^\uparrow (A|E|P_{AE})$  are both {\em strictly concave} so $\hatR_s$ and $\hatR_s^\uparrow$ are both {\em monotonically decreasing} in $s$.
 
 The R\'enyi entropies can be shown to satisfy a form of {\em data processing inequality}. In particular if $f:\calA\to\calM$ is any function on the set $\calA$, we have 
\begin{align}
H(f(A)|E|P_{AE}) &\le  H(  A |E|P_{AE}) ,\label{eqn:dpi1}\\
H_{1+s}(f(A)|E|P_{AE}) &\le  H_{1+s}(  A |E|P_{AE}), \label{eqn:dpi2}\\
H_{1+s}^\uparrow(f(A)|E|P_{AE}) &\le  H_{1+s}^\uparrow(  A |E|P_{AE})\label{eqn:dpi3} .
\end{align}
Inequalities \eqref{eqn:dpi2} and \eqref{eqn:dpi3} hold true for all $s> -1$.  
These inequalities  say  that processing the random variable $A$ cannot increase its randomness measured under any of the above conditional R\'enyi entropies.
\subsection{R\'enyi Security Criteria}\label{sec:security_meas}
Now, we introduce various criteria that measure  independence and uniformity {\em jointly}. The {\em mutual information} is 
\begin{equation}
I(A\wedge E|P_{AE}) := D(P_{AE}\| P_A\times P_E).
\end{equation}
This, together with its normalized version,  has been  traditionally used as measure of dependence in classical information-theoretic security \cite{liang_book,  Bloch_book}, going back to the seminal work of Wyner~\cite{Wyn75} for the wiretap channel. It was also used by Ahlswede and Csisz\'ar for the secret key agreement problem~\cite{AC93}. 
However,   it does not guarantee approximate uniformity of the source $P_A$ on $\calA$. Thus, we introduce the {\em modified mutual information}
\begin{align}
C(A|E|P_{AE}) &:=D(P_{AE}\| P_{\mix,\calA}\times P_E) \\ 
&=\log|\calA| - H(A|E |P_{AE}).
\end{align}
This quantity was also considered by Csisz\'ar and Narayan~\cite[Eq.~(6)]{CN04} in their work on secrecy capacities. An axiomatic justification of $C(A|E|P_{AE})$ was provided recently by Hayashi~\cite[Thm.~8]{Hayashi2013}.  The modified mutual information $C(A|E|P_{AE})$  clearly satisfies
\begin{equation}
C(A|E|P_{AE})= I(A\wedge E|P_{AE})+ D(P_A\| P_{\mix,\calA}). \label{eqn:standard_secur}
\end{equation}
Hence, if $C(A|E|P_{AE})$ is small, $A$ is approximately independent of $E$  {\em and} $A$ is approximately uniform on its alphabet, desirable properties in information-theoretic security. 
We may further generalize the modified mutual information by considering {\em R\'enyi information measures}, introduced in Section~\ref{sec:info_measures}, as follows:
\begin{align}
C_{1+s}(A|E|P_{AE})  &:= D_{1+s} (P_{AE}\| P_{\mix,\calA}\times P_E) \\* 
&=\log|\calA| - H_{1+s}(A|E |P_{AE}). \label{eqn:relate_C_H}
\end{align}
This   can be relaxed to give yet another security measure---the {\em Gallager-form of the modified mutual information}: 
\begin{align}
C_{1+s}^{\uparrow}(A|E|P_{AE})&:=     \min_{Q_E\in\calP(\calE)}     D_{1+s} (P_{AE}\| P_{\mix,\calA}\times Q_E) \\
&=\log|\calA| - H_{1+s}^{\uparrow}(A|E |P_{AE}).\label{eqn:relate_C_H_g}
\end{align}
We characterize these quantities asymptotically  when $(A,E)\equiv (f(A^n),E^n)$  for some (classes of) hash functions $f(\cdot)$. The quantities $H_{1+s}$ and $H_{1+s}^\uparrow$ can be regarded as   equivocations~\cite{Wyn75} so  $C_{1+s}$ and $C_{1+s}^\uparrow$  are  the negative of the equivocations up to a   shift.  We work with $C_{1+s}$ and $C_{1+s}^\uparrow$ in the rest of the paper as they are more convenient and they admit the interpretation as {\em security criteria}.  

\subsection{Decomposition of the R\'enyi Security Criteria into   Mutual Information and  Divergence Terms}\label{sec:prop_security_meas}
We   note that for any $s\ge -1$, $C_{1+s}(A|E|P_{AE})=0$ if and only if $P_{AE}=  P_{\mix,\calA}\times P_E$ or equivalently, $A$ is uniform on $\calA$ and statistically independent of $E$.  This is because $D_{1+s}(P\|Q)$ is a divergence so $D_{1+s}(P\| Q) = 0$ if and only if $P=Q$ \cite{vanErven14}. The same is true for the case $C_{1+s}^\uparrow(A|E|P_{AE})=0$. From this observation, we see that $C_{1+s}(A|E|P_{AE})$ and $C_{1+s}^\uparrow(A|E|P_{AE})$ also measure how close the source or ``key'' $A$ is to uniform and how secure $A$ is from an adversary $E$. Thus the quantities we consider are generalizations of the standard security measure $C(A|E|P_{AE})$ in \eqref{eqn:standard_secur} and measure uniformity and security in a different way. 

More quantitatively, one may wonder whether the security criteria $C_{1+s}(A|E|P_{AE})$ and $C_{1+s}^\uparrow(A|E|P_{AE})$ admit a decomposition into  ``mutual information'' and ``divergence'' terms and similar to  \eqref{eqn:standard_secur}. We first  consider $C_{1+s}(A|E|P_{AE})$. Define $g_s(a) := \sum_e P_{AE}(a,e)^{1+s} P_E(e)^{-s}$. We then see from the definition of the R\'enyi divergence of order $(1+s)$ that
\begin{align}
&\rme^{ s D_{1+s}  (P_{AE}\| P_{\mix,\calA}\times P_E)}  \nn\\*
&= \Big(\frac{1}{|\calA|}\Big)^{-s}\sum_{a}  \bigg( \sum_e P_{AE}(a,e)^{1+s}P_E(e)^{-s} \bigg) \label{eqn:sibson0}\\
&=  \Big(\frac{1}{|\calA|}\Big)^{-s}\sum_{a} \Big( g_s(a)^{\frac{1}{1+s}} \Big)^{1+s}\\
&=  \Big(\frac{1}{|\calA|}\Big)^{-s}\Big(\sum_{a'} g_s(a')^{\frac{1}{1+s}} \Big)^{1+s}  \sum_a  \bigg( \frac{g_s(a)^{\frac{1}{1+s} } }{\sum_{a'} g_s(a')^{\frac{1}{1+s}}} \bigg)^{1+s}  .
\end{align}
As a result, one has 
\begin{align}
 D_{1+s}  (P_{AE}\| P_{\mix,\calA}\times P_E)=\frac{1+s}{s} \log \sum_{a'} g_s(a')^{\frac{1}{1+s}} + \frac{1}{s}\log  \sum_a  \bigg( \frac{g_s(a)^{\frac{1}{1+s} } }{\sum_{a'} g_s(a')^{\frac{1}{1+s}}} \bigg)^{1+s} +\log|\calA|. \label{eqn:decomposeD}
\end{align}
Invoking the definition of $g_s(a)$, we see that the first term can be rewritten as 
\begin{align}
\frac{1+s}{s} \log \sum_{a'} g_s(a')^{\frac{1}{1+s}}  &= \frac{1+s}{s}\log\sum_{a'} P_A(a') \bigg( \sum_e P_{E|A}(e|a')^{1+s} P_E(e)^{-s} \bigg)^\frac{1}{1+s} \\
&=: I^{(\mathrm{Sibson})}_{1+s}(E\wedge A|P_{AE}). \label{eqn:sibson}
\end{align}
This is exactly {\em Sibson's} definition of the order-$(1+s)$  R\'enyi mutual information~\cite{Sibson}.   See  Verd\'u's work in \cite[Sec.~III]{Verdu_RenyiMI} for the properties of $I^{(\mathrm{Sibson})}_{1+s}(E\wedge A|P_{AE})$  and a   generalization to arbitrary  alphabets. See Hayashi's work \cite[Sec.~II.C]{Hayashi15} for a generalization of  $I^{(\mathrm{Sibson})}_{1+s}(E\wedge A|P_{AE})$  to   quantum systems. The work of Tomamichel and Hayashi in~\cite[Sec.\ IV.B]{TomHay15} provides an operational interpretation of this quantity in the context of composite hypothesis testing. The sum of the second and third terms in \eqref{eqn:decomposeD} form a R\'enyi divergence of order $(1+s)$. In particular, the second term is the negative R\'enyi entropy of order $(1+s)$ of the probability mass function $Q_A^{(s)}(a) := g_s(a)^{\frac{1}{1+s}}/\sum_{a'} g_s(a')^{\frac{1}{1+s}}$. Hence,
\begin{equation}
D_{1+s}  (P_{AE}\| P_{\mix,\calA}\times P_E)=I^{(\mathrm{Sibson})}_{1+s}(E\wedge A|P_{AE})+D_{1+s} \big( Q_A^{(s)} \| P_{\mix,\calA} \big). \label{eqn:decomposeD2}
\end{equation}
Because $Q_{A}^{(0)} =P_A$, and $\lim_{s\to 0}I^{(\mathrm{Sibson})}_{1+s}(E\wedge A|P_{AE})= I(A\wedge E|P_{AE})$, the decomposition in \eqref{eqn:decomposeD2} is a generalization of~\eqref{eqn:standard_secur}. Equation  \eqref{eqn:decomposeD2}   is also reminiscent of an information geometric Pythagorean theorem~\cite{Ama00} (but for R\'enyi divergence here). The distribution $Q_A^{(s)} \times P_E$ can be regarded as the $D_{1+s}$-information projection of $P_{AE}$ onto the set $\{ Q_A \times P_E : Q_A\in\calP(\calA)\}$.  This was  observed in the quantum information context by Sharma and~Warsi~\cite[Lemma~3 in Suppl.~Mat.]{SharmaWarsi13}. They called the relation the {\em quantum Sibson identity}.

Next we consider  the Gallager-form of the modified mutual information $C_{1+s}^\uparrow(A|E|P_{AE})$.   From \eqref{eqn:relate_C_H_g},  it can be seen by adding and subtracting $H_{1+s}(A|P_A)$ that 
\begin{align}
C_{1+s}^{\uparrow}(A|E|P_{AE}) = H_{1+s}(A|P_A)- H_{1+s}^\uparrow(A|E|P_{AE}) +D_{1+s}(P_A\|P_{\mix,\calA}) .\label{eqn:decompose_galla}
\end{align}
We recognize that the sum of the first two terms constitutes {\em Arimoto's}~\cite{arimoto75} definition of the order-$(1+s)$  R\'enyi mutual information 
\begin{equation}
I_{1+s}^{(\mathrm{Arimoto})}  (A\wedge E|P_{AE}) =  H_{1+s}(A|P_A)- H_{1+s}^\uparrow(A|E|P_{AE}).
\end{equation}
Since $\lim_{s\to 0} I_{1+s}^{(\mathrm{Arimoto})}(A\wedge E|P_{AE})  =I(A\wedge E|P_{AE})$, the security criterion $C_{1+s}^{\uparrow}(A|E|P_{AE}) $ also admits a decomposition similar to \eqref{eqn:standard_secur}. See \cite[Sec.~II.A]{Verdu_RenyiMI} for detailed discussions of the properties of $I_{1+s}^{(\mathrm{Arimoto})}(A\wedge E|P_{AE}) $.

\section{Asymptotics of the Equivocation} \label{sec:equiv}
In this section we present our   results concerning the asymptotic behavior of the equivocation. First we define precisely the notion of hash function.  This is a generalization of the definition by Carter and Wegman~\cite{carter79}.

\begin{definition} \label{def:has}
 A {\em  random\footnote{For brevity, we will sometimes omit the qualifier ``random''. It is understood, henceforth, that all so-mentioned hash functions are random hash functions.}   hash function} $f_X$ is a stochastic map from $\calA$ to $\calM:= \{1, \ldots ,M\}$, where $X$ denotes a random variable describing its stochastic behavior. An ensemble of random hash 
functions $f_X$ is called an {\em $\epsilon$-almost  universal$_2$ hash function} if it satisfies the following
condition: For any {\em distinct} $a_1, a_2\in\calA$, 
\begin{equation}
\Pr\big(f_X(a_1)=f_X(a_2) \big) \le \frac{\epsilon}{M}.\label{eqn:hash}
\end{equation}
   When $\epsilon=1$, we simply say that the ensemble of functions is a  {\em  universal$_2$ hash function}.
\end{definition}
As an example, if we randomly and uniformly assign each element of $a\in\calA$ into one of $M$ bins indexed by $m\in\calM$ (i.e., the familiar random binning process introduced by Cover~\cite{cover75}), then $\Pr(f_X(a_1)=f_X(a_1))=\frac{1}{M}$ so this is a universal$_2$ hash function, and furthermore, \eqref{eqn:hash} is achieved with equality. 

Let $|t|^+=\max\{0,t\}$.  The following is our first main result.
\begin{figure}[t]
\centering
\begin{picture}(250,150)
\put(0,10){\vector(1,0){250}}
\put(10,0){\vector(0,1){150}}
\put(140,10){\line(1,1){10}}
\put(160,30){\line(1,1){10}}
\put(160,30){\line(1,1){10}}
\put(180,5){\line(0,1){10}}
\put(140,5){\line(0,1){10}}
\put(100,5){\line(0,1){10}}
\put(60,5){\line(0,1){10}}
\multiput(180,10)(0,8){5}{\line(0,1){4}}
\thicklines
\put(10,10){\line(1,0){90}}
\put(60,10){\line(1,1){120}}
\put(180,50){\line(1,1){60}}
\qbezier(100,10)(140,10)(180,50)
\put(50,-5){$H_{1+s}$}
\put(100,-5){$H_{1}$}
\put(130,-5){$H_{1-s}$}
\put(0,0){$0$}
\put(175,-6){$\hatR_{-s}$}
\put(245,-3){$\mbox{Rates}$}
\put(14,145){$\mbox{Security}$}
\put(78,90){$\displaystyle\lim_{n\to\infty}\frac{1}{n} C_{1+s}$}
\put(168,100){$\displaystyle\lim_{n\to\infty}\frac{1}{n} C_{1-s}$}
\end{picture}
\caption{Schematic showing the relation between the various entropies and the transition rate $\hatR_{-s}$ (defined in \eqref{eqn:crit_rate1} and \eqref{eqn:minus_equivalence}). The figure with the Gallager forms of the  conditional R\'enyi entropy $H_{1\pm s}^\uparrow$  and $\hatR_{-s}^\uparrow$  (defined in \eqref{eqn:crit_rate2} and \eqref{eqn:minus_equivalence2}) is completely analogous.  See Fig.~\ref{fig:equivs_gal}. }
\label{fig:schematic}
\end{figure}
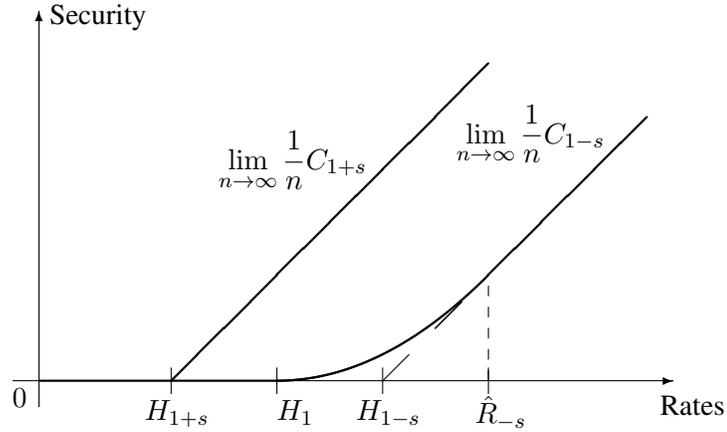

\begin{figure}[t]
\centering
\includegraphics[width = .6\columnwidth]{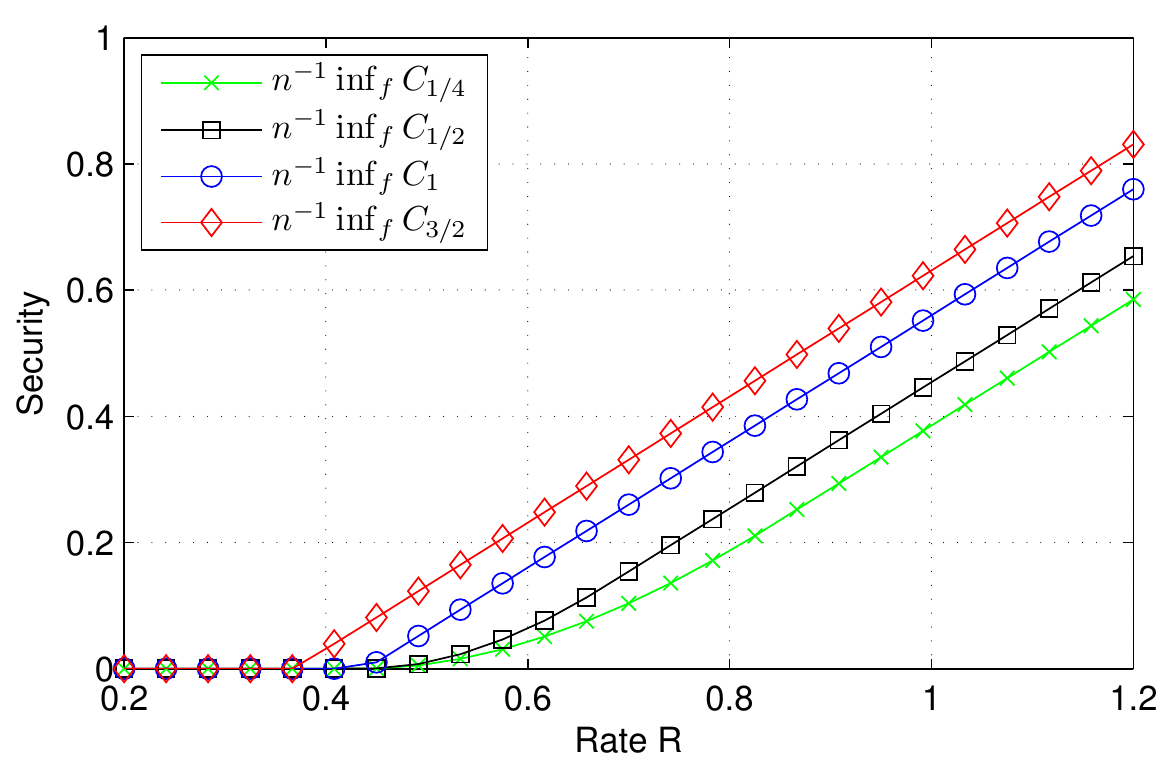} 
\caption{Illustration of the security measures  $C_{1+s}$ and $C_{1-s}$ (for $s\in [0,1]$) in \eqref{eqn:C1s} and~\eqref{eqn:C1s_minus} for the discrete memoryless multiple source $P_{AE}$ where $P_{AE}(0,0)=  0.7$ and $P_{AE}(0,1)=P_{AE}(1,0)=P_{AE}(1,1)=0.1$.  }\label{fig:equivs}
\includegraphics[width = .6\columnwidth]{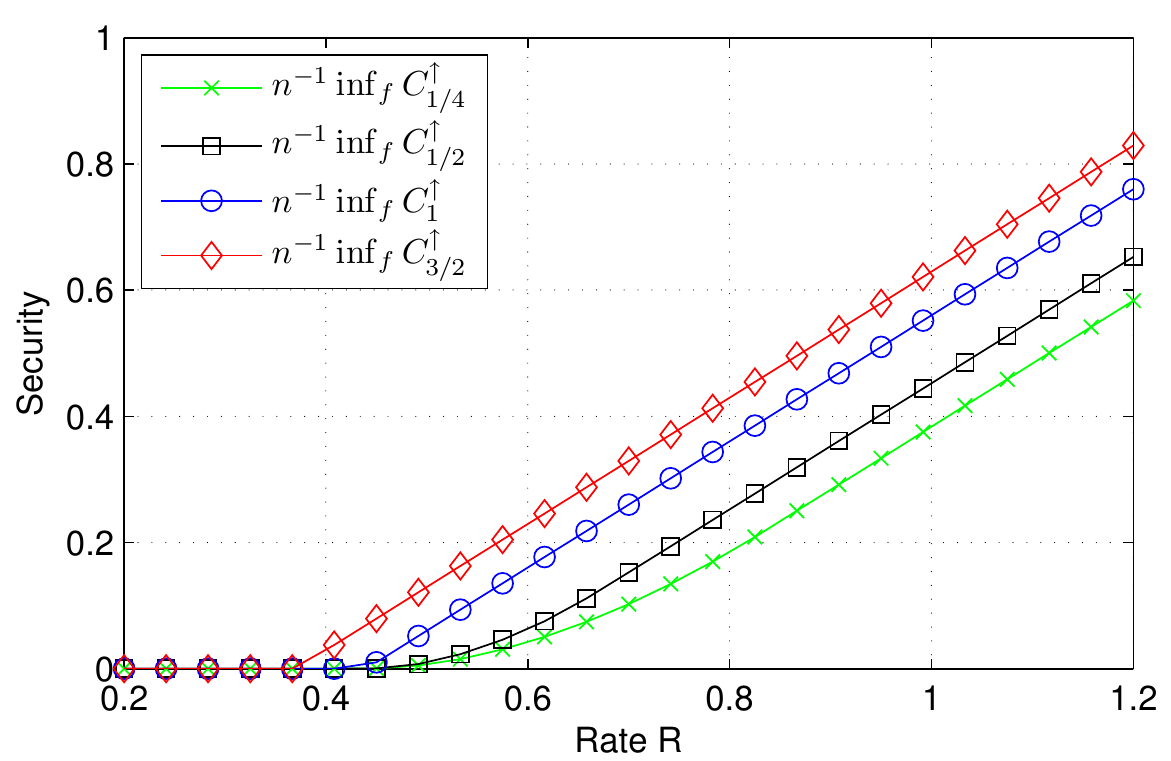}
\caption{Illustration of the security measures  $C_{1+s}^\uparrow$  and $C_{1-s}^\uparrow$ (for $s\in [0,1]$) in \eqref{eqn:C1sG} and \eqref{eqn:C1sG_minus} respectively  for the same source. }
\label{fig:equivs_gal}
\end{figure}

\begin{theorem}[Asymptotics of the Equivocation] \label{thm:equiv}
Let\footnote{As is usual in information theory, we ignore the integer effects on the size of the hash function $M_n=\|f\|$ since this is inconsequential asymptotically. This imprecision is also employed in the sequel for notational convenience. }  $M_n=\rme^{nR}$.  Assume that  $f_{X_n}:\calA^n\to \calM_n=\{1,\ldots, M_n\}$ is a 
random  hash function.\footnote{In particular, all the infima in \eqref{eqn:C1s}--\eqref{eqn:C1sG_minus} (as well as similar statements in the sequel) are taken over all $f_{X_n}$ that are random hash functions. \label{fn:inf}} For any $s\in [0,1]$, we have 
\begin{align}
\lim_{n\to\infty}\frac{1}{n}\inf_{f_{X_n}}    C_{1+s}(f_{X_n}(A^n)|E^n X_n|P_{AE}^n \times P_{X_n})  &= |R- H_{1+s}(A|E|P_{AE})|^+ , \label{eqn:C1s}\\
\lim_{n\to\infty}\frac{1}{n}\inf_{f_{X_n}}   C_{1+s}^\uparrow(f_{X_n}(A^n)|E^n X_n|P_{AE}^n \times P_{X_n})&  = |R- H_{1+s}^{\uparrow}(A|E|P_{AE})|^+ .\label{eqn:C1sG}
\end{align}
Furthermore, for any $s\in ( 0,1]$, we also have 
\begin{align}
\lim_{n\to\infty}\frac{1}{n}\inf_{f_{X_n}}      C_{1-s}(f_{X_n}(A^n)|E^n X_n|P_{AE}^n \times P_{X_n})  &=\left\{ \begin{array}{cc}
 R- H_{1-s}(A|E|P_{AE})  & R\ge \hatR_{-s}\\
\max_{t\in [0,s]} \frac{t}{s}(R- H_{1-t}(A|E|P_{AE}) )  & R\le \hatR_{-s}
\end{array} \right.  ,  \label{eqn:C1s_minus} \\
\lim_{n\to\infty}\frac{1}{n}\inf_{f_{X_n}}    C_{1-s}^\uparrow(f_{X_n}(A^n)|E^n X_n|P_{AE}^n \times P_{X_n})   &=\left\{ \begin{array}{cc}
 R- H_{1-s}^{\uparrow}(A|E|P_{AE}) & R\ge \hatR_{-s}^\uparrow\\
\max_{t\in [0,s]} \frac{t}{s}(R- H_{1-t|1-s } (A|E|P_{AE}) )  & R\le \hatR_{-s}^\uparrow
\end{array} \right.     , \label{eqn:C1sG_minus}
\end{align}
where recall that  $ \hatR_{-s}$ and $ \hatR_{-s}^\uparrow$ are defined in \eqref{eqn:crit_rate1} and \eqref{eqn:crit_rate2} respectively.  (Also see \eqref{eqn:minus_equivalence} and \eqref{eqn:minus_equivalence2} for alternative representations.)  
Furthermore,  the infima in \eqref{eqn:C1s}--\eqref{eqn:C1sG_minus}  are achieved by any  sequence of   $\epsilon$-almost  universal$_2$ hash functions $f_{X_n}$ (where $\epsilon$ is a fixed positive number).
\end{theorem}

This result is proved in Section~\ref{sec:prf_equiv}. The  ideas to prove the direct parts  (upper bounds on the leakage rates) for the $s=0$ cases are contained in previous works such as \cite{Hayashi11,chou12,chou15}. All other parts are   novel. 

We remark that the converse parts  (lower bounds) to \eqref{eqn:C1s}--\eqref{eqn:C1sG} hold for all $s\ge 0$ (and not only being upper bounded by $1$) owing to the data processing inequalities in \eqref{eqn:dpi2}--\eqref{eqn:dpi3}.  Furthermore, instead of the formulae in \eqref{eqn:crit_rate1} and \eqref{eqn:crit_rate2}, the   rates in which the behavior of the security measures change  $\hatR_{-s}$ and $\hatR_{-s}^\uparrow$ can also be expressed as 
\begin{align}
\hatR_{-s}  &= \frac{\rmd }{\rmd t} \, tH_{1-t}(A|E|P_{AE})\Big|_{t=s } , \quad\mbox{and}\label{eqn:minus_equivalence}\\
\hatR_{-s}^\uparrow & = \frac{\rmd }{\rmd t} \, tH_{1-t}^\uparrow(A|E|P_{AE})\Big|_{t=s } . \label{eqn:minus_equivalence2}
\end{align}
These alternative expressions  for $\hatR_{-s}$ and $\hatR_{-s}^\uparrow$ will be useful in the proof of Theorem~\ref{thm:equiv}.


The results in  \eqref{eqn:C1s}--\eqref{eqn:C1sG_minus} imply that an optimum sequence of hash functions $\{f_{X_n}\}_{n\in\bbN}$ is such that asymptotically  the normalized security measure  $C_{1 + s}$ and its Gallager-type counterpart $C_{1 + s}^{\uparrow}$ increase  {\em linearly} with the rate $R$ if the rate is larger than the conditional R\'enyi entropy and its Gallager-type counterpart.   However, note that this only holds for the case where $R$ is greater than the analogue of the critical rates, defined in \eqref{eqn:crit_rate1}--\eqref{eqn:crit_rate2} in the case where the R\'enyi parameter $\alpha=1-s$ is less than one. Observe that there is difference in   behavior when we consider the other direction, i.e., the quantities $C_{1-s}$ and $C_{1-s}^{\uparrow}$ for $s\in [0,1]$. 
     Below the critical rate, the equivocation no longer increases linearly with $R$ but  is nonetheless still convex in $R$. See Fig.~\ref{fig:schematic} for a schematic of the various rates and the behavior of the equivocations. We numerically calculate the asymptotics of the equivocations in Theorem~\ref{thm:equiv} and display the results in  Figs.~\ref{fig:equivs}  and~\ref{fig:equivs_gal}. The behaviors of the normalized security measure $C_{1+s},C_{1-s}$ and their Gallager-type counterparts $C_{1+s}^\uparrow,C_{1-s}^\uparrow$ are similar.  


Finally, we examine the optimal (maximum) key generation rates, i.e., the largest rates $R$ for which there exists a sequence of functions from $\calA^n$ to $\{1,\ldots, \rme^{nR}\}$ such that  $\frac{1}{n} C_{1+ s}$ or $\frac{1}{n} C_{1 + s}^\uparrow$ tend to zero as the blocklength grows.  We observe from the following corollary that this cutoff rate depends strongly on the sign of $s$.  In particular for $s \in (0,1]$, the cutoff rates are  $H_{1+s}(A|E|P_{AE})$ and $H_{1+s}^\uparrow(A|E|P_{AE})$ respectively, while for $s\in [-1,0]$, the cutoff rates are both equal to the Shannon conditional entropy $H(A|E|P_{AE})$ independent of $s$. This difference between the behaviors of the optimal key generation rates depending on the sign of $s$ (also illustrated in Figs.~\ref{fig:equivs}  and~\ref{fig:equivs_gal}) is somewhat surprising (at least to the authors).

\begin{corollary}[Optimal  key generation rates] \label{cor:key}
We have
\begin{align}
\sup\bigg\{ R \in\bbR_+ : \lim_{n\to\infty}\inf_{f:\calA^n \to \{1,\ldots, \rme^{nR} \}} \frac{C_{1+s}(f(A^n) | E^n| P_{AE}^n) }{n} = 0 \bigg\} &= \left\{ \begin{array}{cl}
H_{1+s}(A|E|P_{AE}) & \mbox{  if  }  s\in (0,1] \\
H (A|E|P_{AE}) & \mbox{  if  } s\in [-1,0] \\
\end{array}  \right. ,  \label{eqn:opt_key1}\\*
\sup\bigg\{ R \in\bbR_+ :\lim_{n\to\infty} \inf_{f:\calA^n \to \{1,\ldots, \rme^{nR} \}}\frac{C_{1+s}^\uparrow(f(A^n) | E^n| P_{AE}^n) }{n} = 0 \bigg\} &= \left\{ \begin{array}{cl}
H_{1+s}^\uparrow( A|E|P_{AE}) & \mbox{  if  }  s\in (0,1] \\
H (A|E|P_{AE}) & \mbox{  if  } s\in [-1,0] \\
\end{array}  \right. . \label{eqn:opt_key2}
\end{align}
\end{corollary}
\begin{proof}
We only prove the statement for $C_{1+s}$ in \eqref{eqn:opt_key1} since that for $C_{1+s}^\uparrow$ in  \eqref{eqn:opt_key2} is completely analogous.
The case for $s\in (0,1]$ is obvious from~\eqref{eqn:C1s}  in  Theorem~\ref{thm:equiv} since the limit is $|R- H_{1+s}(A|E|P_{AE})|^+$.
Now, for the case $s\in [-1,0]$, if $R\le H(A|E|P_{AE})$,  we know from the monotonically decreasing nature  of $H_{1+s}(A|E|P_{AE})$ (in $s$) that $R-H_{1-t}(A|E|P_{AE})$ is non-positive for $t\in[0,s]$. Thus, referring to \eqref{eqn:C1s_minus} in Theorem~\ref{thm:equiv},  the optimal $t$ in the optimization $\max_{t\in [0,s]} \frac{t}{s}(R- H_{1-t}(A|E|P_{AE}) )  $ is attained at $t=0$ and consequently, the optimal objective value is $0$. On the other hand, for any $R > H(A|E|P_{AE})$,  the optimal $t\in (0,s]$ and so the optimal objective value  is (strictly) positive. Thus, for $s\in [-1,0]$, the optimal key generation rate is the Shannon conditional entropy $H(A|E|P_{AE})$.  This concludes the proof for \eqref{eqn:opt_key1}.
\end{proof}
 In Section \ref{sec:motivation}, we alluded to the importance of the collision entropy $H_2$ in cryptography and QKD. The implication of \eqref{eqn:opt_key1} in  Corollary \ref{cor:key} is that if we operate at a hashing rate $R> H_2$ and we employ the security criterion $C_2$, then there will inevitably be some residual leakage of the source $A^n$ given a hashed version $f(A^n)$ and side-information $E^n$. 
 
Because of the normalizations of $C_{1+ s}$ and $C_{1+ s}^\uparrow$ by $n$ in \eqref{eqn:opt_key1} and \eqref{eqn:opt_key2}, Corollary \ref{cor:key} is analogous to results in the vast majority of the literature in information-theoretic security~\cite{liang_book,  Bloch_book}  where the {\em weak secrecy} criterion is employed. We address the analogue of the {\em strong secrecy} criterion~\cite{Maurer:2000} in Theorem \ref{thm:exponents}  to follow where we not only demand that the unnormalized quantities $C_{1+ s}$ and $C_{1+ s}^\uparrow$ vanish with $n$, we also demand that they do so exponentially fast and we identify the exponents. 

\section{Exponential Behavior    of the Security Measures}\label{sec:exponent}
In this section,   we evaluate the exponential rates of decay of the security measures $C_{1\pm s} $ and $C_{1\pm s}^{\uparrow}$ for fixed rates $R$ above an analogue of the critical rate. 

\begin{theorem}[Exponents of the  Equivocation] \label{thm:exponents}
Let $M_n=\rme^{nR}$. Assume that  $f_{X_n}:\calA^n\to \calM_n=\{1,\ldots, M_n\}$ is a 
random hash function. For $R\ge \hatR_1$ ($\hatR_s$ being defined in \eqref{eqn:crit_rate1}), 
and any $s\in [0,1]$, we have 
\begin{align}
\lim_{n\to\infty}-\frac{1}{n}\log \inf_{f_{X_n}} C_{1+s}(f_{X_n}(A^n)|E^n X_n|P_{AE}^n \times P_{X_n}) &   = \left|\sup_{t\in [s,1)} tH_{1+t}(A|E|P_{AE})-tR \right|^+ ,\label{eqn:exp1} \\
\lim_{n\to\infty}-\frac{1}{n}\log \inf_{f_{X_n}}  C_{1-s}(f_{X_n}(A^n)|E^n X_n|P_{AE}^n \times P_{X_n}) &  =\max_{t\in [0,1]} tH_{1+t}(A|E|P_{AE})-tR . \label{eqn:gal_exp1}
\end{align}
For the Gallager-type counterparts of the R\'enyi quantities  and  $R\ge\hatR_1^\uparrow$ ($\hatR_s^\uparrow$ being defined in \eqref{eqn:crit_rate2}), 
and any $s\in [0,1]$, we also have 
\begin{align}
\lim_{n\to\infty}-\frac{1}{n}\log \inf_{f_{X_n}}  C_{1+s}^{\uparrow}(f_{X_n}(A^n)|E^n X_n|P_{AE}^n \times P_{X_n})    & = \left|\max_{t\in [s,1]} tH_{1+t}(A|E|P_{AE})-tR \right|^+ ,\label{eqn:exp2}\\
\lim_{n\to\infty}-\frac{1}{n}\log \inf_{f_{X_n}}  C_{1-s}^{\uparrow}(f_{X_n}(A^n)|E^n X_n|P_{AE}^n \times P_{X_n}) &=\max_{t\in [0,1]} tH_{1+t}(A|E|P_{AE})-tR.  \label{eqn:gal_exp2}
\end{align}The   infima in \eqref{eqn:exp1}--\eqref{eqn:gal_exp2}   are achieved by any sequence of $\epsilon$-almost universal$_2$ hash functions $f_{X_n}$.
\end{theorem}

\begin{figure}
\centering
\includegraphics[width = .6\columnwidth]{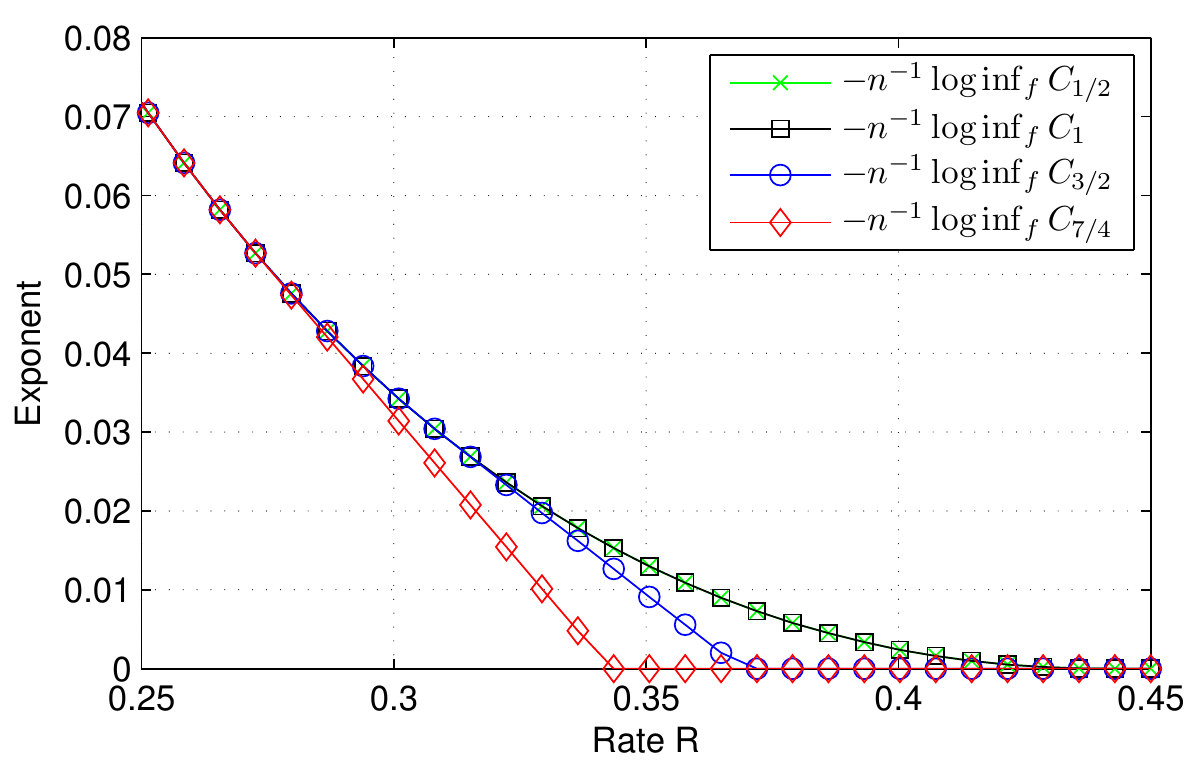}
\caption{Illustration of the exponents of the security measures  $C_{1+s}$  and $C_{1-s}$ (for $s\in [0,1]$) in \eqref{eqn:exp1}  and \eqref{eqn:gal_exp1} respectively  for the discrete memoryless multiple source $P_{AE}$ as in Fig.~\ref{fig:equivs}.  The curves for $C_{1/2}$ and $C_1$  are identical and they are equal to zero for all rates $R\ge H_1 = 0.4400$ bits per source symbol (cf.\ Corollary~\ref{cor:key}).  }
\label{fig:exp}
\end{figure}

This result is proved in Section~\ref{sec:prf_exp}. The techniques for the direct parts are somewhat similar to those in \cite{Hayashi06,Hayashi11, chou15} using  improved  versions of Bennett {\em et al.'s}~\cite{BBCM} bound which was based on the R\'enyi entropy of order 2. However, the non-asymptotic bounds (e.g., Lemma \ref{lem:os_c_2}) and asymptotic evaluations for the converse parts require new ideas.  Different from the direct part, we need to convert the evaluation of $\rme^{-sC_{1-s}}$ and $\rme^{-\frac{s}{1-s}C_{1-s}^\uparrow}$ into information spectrum~\cite{Han10} quantities (involving the conditional entropy random variable) so that is is amenable to asymptotic evaluation. These information spectrum quantities are then evaluated using various large deviation~\cite{Dembo} bounds, such as Cramer's theorem.   
 We make several other observations. 

First,  the exponents of the security indices (namely $-\frac{1}{n}\log C_{1\pm s}$ and $-\frac{1}{n}\log C_{1\pm s}^\uparrow$) are non-negative because $C_{1\pm s} \le\log| f_{X_n} (\calA^n) | =O(n)$ and $C_{1\pm s}^\uparrow \le\log| f_{X_n} (\calA^n) | =O(n)$ (cf.\ their definitions in \eqref{eqn:relate_C_H} and \eqref{eqn:relate_C_H_g}). The expressions in \eqref{eqn:gal_exp1} and \eqref{eqn:gal_exp2} are already nonnegative and so we only need to include the $| \fndot |^+$ operation for \eqref{eqn:exp1} and \eqref{eqn:exp2}.

Second, the derivative of the conditional R\'enyi entropies $\hatR_1$ and $\hatR_1^\uparrow$  are the analogues of the critical rate in error exponent analysis \cite{gallagerIT,Csi97}. For the exponents, we have a complete characterization of the exponential rates of decay of both $ C_{1\pm s}$ and $ C_{1\pm  s}^\uparrow$  for $s\in [0,1]$ and they are given by optimization of quantities that are related to the conditional R\'enyi entropy. We observe that  the Gallager form results in larger exponents in general as the optimizations in  \eqref{eqn:gal_exp1} and \eqref{eqn:gal_exp2} are larger than their non-Gallager counterparts in \eqref{eqn:exp1} and \eqref{eqn:exp2} respectively. 

Finally, the exponents in  \eqref{eqn:exp1}  and \eqref{eqn:gal_exp1} of Theorem \ref{thm:exponents} are illustrated in Fig.~\ref{fig:exp}.    We observe the same behavior for the exponents of the Gallager forms in~\eqref{eqn:exp2} and~\eqref{eqn:gal_exp2} since the  expressions  are the same and so we omit these cases. We note (from the plot and from direct evaluations) that the zero-crossings for the exponents of $C_{1/2}, C_1 , C_{3/2}$ and $C_{7/4}$ occur at $H_1, H_1, H_{3/2}$ and $H_{7/4}$ respectively ($H_1$ being the Shannon entropy). This is in line with Corollary~\ref{cor:key}. Indeed, the exponent being positive implies that the normalized security measures $\frac{1}{n} C_{1\pm s}$  and $\frac{1}{n} C_{1\pm s}^\uparrow$ vanish   as   blocklength grows.  Thus, we conclude that the optimal key generation rates under both the strong and weak secrecy criteria are the same.

\section{Second-Order Asymptotics } \label{sec:2nd}
 In the previous sections, the security measures in terms of equivocations and their logarithms were normalized by the blocklength $n$. In this section, we study different normalizations, e.g., by $\sqrt{n}$. In addition, we examine the effect of changing the size of the hash function $M_n$ from $\rme^{nR}$ (considered in Sections~\ref{sec:equiv} and \ref{sec:exponent})  to  $\rme^{nR +\sqrt{n}L}$, where $L\in\bbR$ is an arbitrary real number.

 \subsection{Basic Definitions}
  To present our results, we first define the following important quantities.
 
\begin{definition}
Given a discrete joint source $P_{AE} \in\calP(\calA\times\calE)$, define the {\em conditional varentropy}~\cite{verdu14} or  {\em conditional source dispersion}~\cite{kost12, TK12c} to be 
\begin{equation}
 V(A|E|P_{AE}) := \sum_{a,e}P_{AE}(a,e) \big(\log P_{A|E}(a|e) + H(A|E|P_{AE}) \big)^2. \label{eqn:dispersion}
 \end{equation} 
 We also define the following variants  of the conditional varentropy
 \begin{align}
 V_1(A|E|P_{AE}) & := \sum_e P_E(e) \big( H(A|E|P_{AE}) -H(A|P_{A|E=e}) \big)^2  \label{eqn:V1_def}\\
  V_2(A|E|P_{AE}) & :=V(A|E|P_{AE}) -V_1(A|E|P_{AE})  \label{eqn:sum_vars} \\
  &=\sum_{a,e} P_{AE}(a,e) \big(\log P_{A|E}(a|e) + H(A|P_{A|E=e})\big)^2.\label{eqn:V2_def}
 \end{align}
\end{definition}
One can  readily  check that $V=V_1+V_2$ from the definitions. This also  follows immediately  from the law of total variance. Let 
\begin{equation}
\Phi(t) :=\frac{1}{\sqrt{2\pi}} \int_{-\infty}^t \rme^{-u^2/2}\, \rmd u
\end{equation}
be the cumulative distribution function of the standard Gaussian random variable. 
With these definitions, we are ready to state our results on the second-order asymptotics for the security measures $C_{1+s}$ and $C_{1+s}^\uparrow$ which are simple functions of the  equivocation $H_{1+s}$ and $H_{1+s}^\uparrow$ respectively.  Note that the second-order analysis of $C_1$  (corresponding to the $s = 0$ case) with  no side information (i.e., $E = \emptyset$) was performed in  Hayashi's work~\cite[Theorem 8]{Hayashi08} in the context of intrinsic randomness based on the relative entropy (Kullback-Leibler divergence) criterion.  The other results in Theorems~\ref{thm:second} and~\ref{thm:second_large} are novel.


To state our result succinctly, we define the quantities which all depend on $s$, $L$ and $P_{AE}$ (but we suppress the dependence on the fixed joint distribution $P_{AE}$ for brevity):
\begin{align}
\Gamma_1(s,L)&:= -\frac{1}{s}\log \left(2^{\frac{s}{1-s}}s^{\frac{s}{1-s}} (1-s) \right)-\frac{1}{s} \log \Phi\bigg( -\frac{L}{\sqrt{V(A|E|P_{AE}) } }\bigg) ,\label{eqn:defA}\\
\Gamma_2(s,L)&:= -\frac{1-s}{s }\log\Phi\bigg( -\frac{L}{\sqrt{V(A|E|P_{AE}) } }\bigg) ,\label{eqn:defB} \\
\Psi_1(s,L)&:= -\frac{1}{s}\log \left(2^{s+ \frac{s}{1-s}}   s^{\frac{s}{1-s}} (1-s) \right) \nn\\*
&\qquad\qquad-\frac{1-s}{s}\log \int_{-\infty}^{\infty} \Phi\bigg(- \frac{L+x}{\sqrt{ V_2(A|E|P_{AE})} }\bigg)^{\frac{1}{1-s}} \frac{ \rme^{-x^2/(2 V_1(A|E|P_{AE}))}}{\sqrt{2\pi V_1(A|E|P_{AE})}}\,\rmd x, \label{eqn:defD}\\
\Psi_2(s,L)&:=-\frac{1-s}{s}\log \Phi\bigg( -\frac{L}{\sqrt{V(A|E|P_{AE}) } } \bigg). \label{eqn:defE}
\end{align} 

\subsection{Bounds on the  Second-Order Asymptotics } 
\begin{theorem}[Second-Order Asymptotics] \label{thm:second}
  Assume that  $f_{X_n}:\calA^n\to \calM_n=\{1,\ldots, M_n\}$ is a 
random hash function. Consider the following three cases:
\begin{itemize}
\item Case (A): $\alpha=1+s$ with $s\in (0,1]$: Suppose that the number of messages $M_n = \rme^{n H_{1+s}(A|E|P_{AE}) + \sqrt{n} L }$ or $M_n = \rme^{n H_{1+s}^\uparrow(A|E|P_{AE}) + \sqrt{n} L }$. When $L\ge 0$, we have 
\begin{align}
\lim_{n\to\infty}\frac{1}{\sqrt{n}}\inf_{f_{X_n}}    C_{1+s}(f_{X_n}(A^n)|E^n X_n|P_{AE}^n \times P_{X_n}) &=L \label{eqn:2order1}\\
\lim_{n\to\infty}\frac{1}{\sqrt{n}}\inf_{f_{X_n}}    C_{1+s}^\uparrow(f_{X_n}(A^n)|E^n X_n|P_{AE}^n \times P_{X_n})& =L. \label{eqn:2order1_g}
\end{align}
Similarly to Theorem~\ref{thm:equiv}, the infima in \eqref{eqn:2order1} and \eqref{eqn:2order1_g} are achieved by any sequence of $\epsilon$-almost universal$_2$ hash functions $f_{X_n}$.

When $L\le 0$, we have 
\begin{align}
\lim_{n\to\infty}-\frac{1}{\sqrt{n}}\log\inf_{f_{X_n}}     C_{1+s}(f_{X_n}(A^n)|E^n X_n|P_{AE}^n \times P_{X_n}) &=-sL \label{eqn:2order1_Lneg}\\
\lim_{n\to\infty}-\frac{1}{\sqrt{n}}\log\inf_{f_{X_n}}  C_{1+s}^\uparrow(f_{X_n}(A^n)|E^n X_n|P_{AE}^n \times P_{X_n})& =-sL. \label{eqn:2order1_g_Lneg}
\end{align}
Similarly to Theorem~\ref{thm:exponents}, the infima in \eqref{eqn:2order1_Lneg} and \eqref{eqn:2order1_g_Lneg} are achieved by any sequence of $\epsilon$-almost  universal$_2$ hash functions $f_{X_n}$.
\item Case (B):   $\alpha=1$ (i.e., $s=0$):  Suppose that $M_n = \rme^{n H  (A|E|P_{AE}) + \sqrt{n} L }$ for some $L\in\bbR$, we have 
\begin{equation}
\lim_{n\to\infty}\frac{1}{\sqrt{n}}\inf_{f_{X_n}}    C_{1}(f_{X_n}(A^n)|E^n X_n|P_{AE}^n \times P_{X_n})  =\int_{-\infty}^{L/\sqrt{ V(A|E|P_{AE})}} \frac{L-\sqrt{V(A|E|P_{AE})}x}{\sqrt{2\pi}} \rme^{-x^2/2}\, \rmd x.\label{eqn:caseB}
\end{equation}
By \eqref{eqn:limit_s0}, the same asymptotic behavior also holds true for the Gallager version of the security measure $C_{1}^\uparrow(f_{X_n}(A^n)|E^n X_n|P_{AE}^n \times P_{X_n})$. Similarly to Theorem~\ref{thm:equiv}, the infima in \eqref{eqn:caseB} is achieved by any sequence of $\epsilon$-almost universal$_2$ hash functions $f_{X_n}$.
 
\item Case (C):  $\alpha=1-s$  with $s\in (0,1]$:   Suppose that $M_n = \rme^{n H (A|E|P_{AE}) + \sqrt{n} L }$ for some $L\in\bbR$, we have 
\begin{align}
\max\left\{ \Gamma_1(s,L),\Gamma_2(s,L) \right\} &\le\liminf_{n\to\infty}  \inf_{f_{X_n}}    C_{1-s}(f_{X_n}(A^n)|E^n X_n|P_{AE}^n \times P_{X_n}) \nn\\*
&\le\limsup_{n\to\infty}  \inf_{f_{X_n}}    C_{1-s}(f_{X_n}(A^n)|E^n X_n|P_{AE}^n \times P_{X_n})  \le \frac{\Gamma_2(s,L)}{1-s}  . \label{eqn:caseC1}
\end{align}
In addition, for the Gallager-type counterparts, with $M_n = \rme^{n H (A|E|P_{AE}) + \sqrt{n} L }$ for some $L\in\bbR$, we also have 
\begin{align}
&\max \{\Psi_1(s,L),\Psi_2(s,L)\}\nn\\*
&\qquad\le\liminf_{n\to\infty}  \inf_{f_{X_n}}    C_{1-s}^\uparrow(f_{X_n}(A^n)|E^n X_n|P_{AE}^n \times P_{X_n}) \nn\\*
&\qquad\le\limsup_{n\to\infty}  \inf_{f_{X_n}}    C_{1-s}^\uparrow(f_{X_n}(A^n)|E^n X_n|P_{AE}^n \times P_{X_n}) \nn\\*
&\qquad\qquad\le \Psi_1(s,L)+\frac{1}{s}\log \left(2^{s+ \frac{1}{1-s}}   s^{\frac{s}{1-s}} (1-s) \right) .  \label{eqn:caseC2}
\end{align}
The upper bounds in \eqref{eqn:caseC1} and \eqref{eqn:caseC2} are achieved by any sequence of $\epsilon$-almost universal$_2$ has functions $f_{X_n}$. 
\end{itemize}
\end{theorem}
This result is proved in Section~\ref{sec:prf_sec}.  We remark that the converse parts  (lower bounds) to \eqref{eqn:2order1}--\eqref{eqn:2order1_g} hold for all $s\ge 0$ (and not only being upper bounded by $1$) owing to the data processing inequalities in \eqref{eqn:dpi2}--\eqref{eqn:dpi3}.

\begin{figure}[t]
\centering
\includegraphics[width = .6\columnwidth]{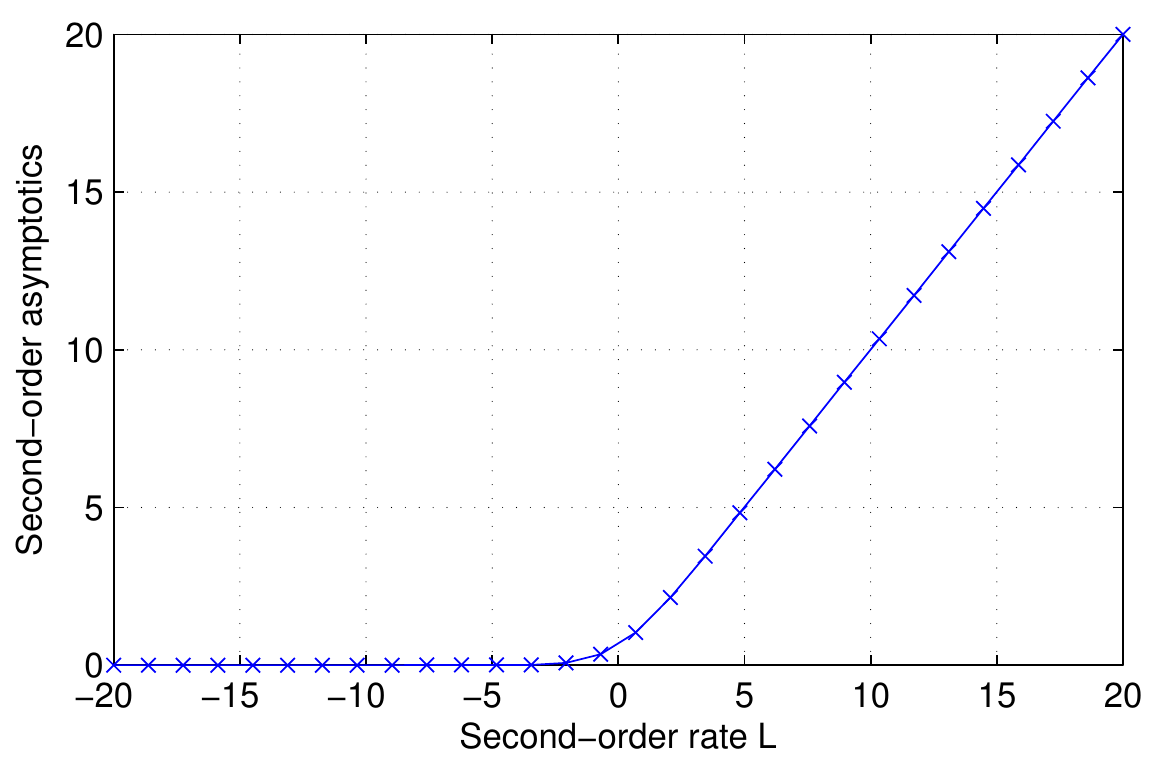}
\caption{Illustration of the second-order asymptotics in Case (B) given by the right-hand-side of~\eqref{eqn:caseB}   for the discrete memoryless multiple source $P_{AE}$ as in Fig.~\ref{fig:equivs}. It is easy to see that the integral there is non-negative. }
\label{fig:sec}
\end{figure} 
 
\subsection{Remarks on Theorem \ref{thm:second}} 
Observe that in Theorem \ref{thm:second} (Case (A) for instance), the number of compressed symbols $M_n$ satisfies
\begin{align}
\log M_n &= n H_{1+s}(A|E|P_{AE}) + \sqrt{n} L, \qquad\mbox{or}, \label{eqn:sec_ord}\\
\log M_n & = n H_{1+s}^\uparrow(A|E|P_{AE}) + \sqrt{n} L.\label{eqn:sec_ord1}
\end{align}
The leading conditional R\'enyi entropy terms  scaling in $n$ are known as the first-order terms, while the terms scaling as $\sqrt{n}$ are known as the second-order terms. The coefficient  $L$ is known as the {\em second-order coding rate}~\cite{Hayashi08,Hayashi09,TanBook} and the second-order asymptotic characterizations depend on $L$.  Note that even though $L$ is termed as the second-order coding {\em rate}, it may be negative. Observe that the conditional varentropies appear in \eqref{eqn:caseB}--\eqref{eqn:caseC2}, which suggests that we evaluate the one-shot bounds using the central limit theorem among other techniques. We have tight results (equalities) for Cases (A) and (B) but unfortunately  not for Case (C) where  the R\'enyi parameter $\alpha=1-s$ for $s\in (0,1]$. However, in the limit of the second-order coding rate $L$ being large (either in the positive or negative direction), we can assert that one of the terms in the maxima in the lower bounds of  \eqref{eqn:caseC1} and \eqref{eqn:caseC2} dominates and matches the upper bound and hence, we have a tight result up to the term in $L^2$ (Theorem~\ref{thm:second_large}). We now comment specifically on each of the cases.  

\begin{enumerate}
\item For Case (A), the second-order asymptotic behaviors of  $C_{1+s}$ and $C_{1+s}^\uparrow$ when they are  normalized by $\frac{1}{\sqrt{n}}$ are linear in $L$.
\item The same is true for Case (B) for large positive $L$ because with $V:= V(A|E|P_{AE})$, 
\begin{align}
\lim_{L\to\infty}\frac{1}{L} \cdot \int_{-\infty}^{L/\sqrt{ V }} \frac{L-\sqrt{V }x}{\sqrt{2\pi}} \rme^{-x^2/2}\, \rmd x=\lim_{L\to\infty}\left\{ \Phi\Big( \frac{L}{\sqrt{V}}\Big)-\frac{\sqrt{V}}{L} \int_{-\infty}^{L/\sqrt{V}}\frac{x\rme^{-x^2/2}}{\sqrt{2\pi}}\,\rmd x\right\}=1.
\end{align}
In contrast, when  $L\to-\infty$ in Case (B),  the limit is zero.   The second-order asymptotics in Case (B) in~\eqref{eqn:caseB} is shown in Fig.~\ref{fig:sec} and is obtained via numerical integration to approximate the integral. The limit in~\eqref{eqn:caseB} is monotonically increasing in $L$. This is intuitive because as $L$ increases, there is potentially more leakage to $E^n$ and less uniformity on the  (larger) support $\{1,\ldots, \rme^{n H (A|E|P_{AE}) + \sqrt{n} L }\}$. 
\item  For Case (C) there is no normalization by $\frac{1}{\sqrt{n}}$ and we only have bounds. However, for    large $|L|$, we will see from Theorem~\ref{thm:second_large} that the second-order asymptotic behavior is quadratic in $L$.  The bounds on the  second-order asymptotics in the two parts (conditional R\'enyi entropy and its Gallager version) of Case (C) in~\eqref{eqn:caseC1}  and~\eqref{eqn:caseC2} are shown in Figs.~\ref{fig:sec_caseC} and~\ref{fig:sec_caseC2} respectively. 
\end{enumerate}
 We conclude that in the second-order asymptotic regime where the number of compressed symbols satisfies \eqref{eqn:sec_ord}--\eqref{eqn:sec_ord1}, there are distinct differences between the three regimes of the R\'enyi parameter $\alpha \in [0,1)$, $\alpha=1$ and $\alpha \in (1,2]$.

\begin{figure}[t]
\centering
\includegraphics[width = .85\columnwidth]{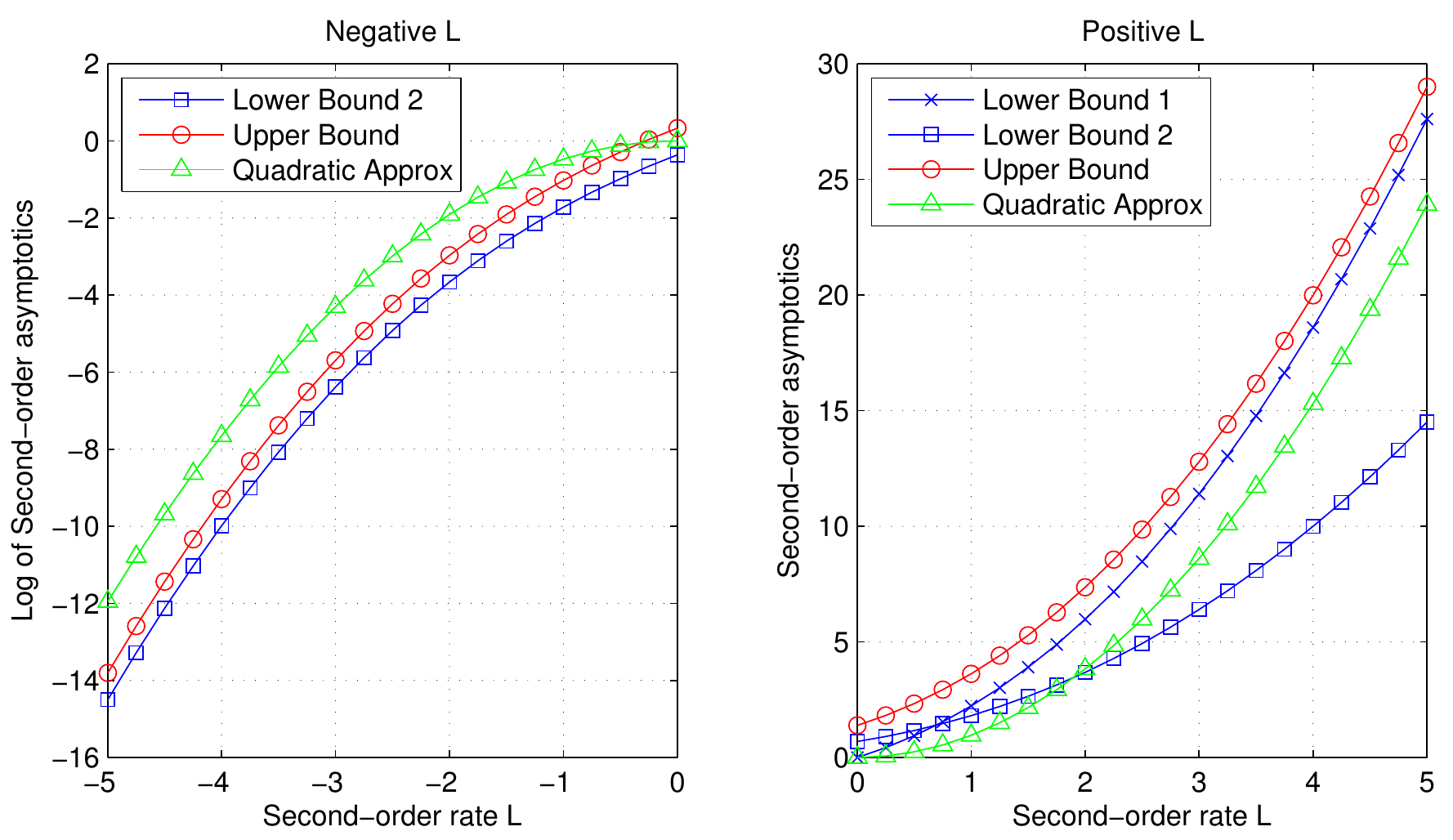}
\caption{Illustration of the bounds on the second-order asymptotics of $C_{1/2}(f_{X_n}(A^n)|E^n X_n|P_{AE}^n \times P_{X_n})$ (i.e., $s=1/2$) in Case (C) given by the left- and right-hand-sides of~\eqref{eqn:caseC1}   for the same source. Note that the figure on the left is plotted in log scale (corresponding to  \eqref{eqn:L-inf}) while the figure on the right is plotted in linear scale.  For $L\le 0$,  lower bound 1 in \eqref{eqn:caseC1} is negative (lower bound 2 dominates) so is not shown in the left plot. Observe the quadratic behaviors; this is corroborated by Theorem~\ref{thm:second_large}. The quadratic approximations in \eqref{eqn:Linf} and \eqref{eqn:L-inf} (without the $O(\log L)$ terms) are also plotted. Observe that there is a constant offset between the quadratic and the bounds as we do not determine the $O(\log L)$ terms in Theorem \ref{thm:second_large} exactly. }
\label{fig:sec_caseC}
\end{figure} 
 

\begin{figure}[t]
\centering
\includegraphics[width = .85\columnwidth]{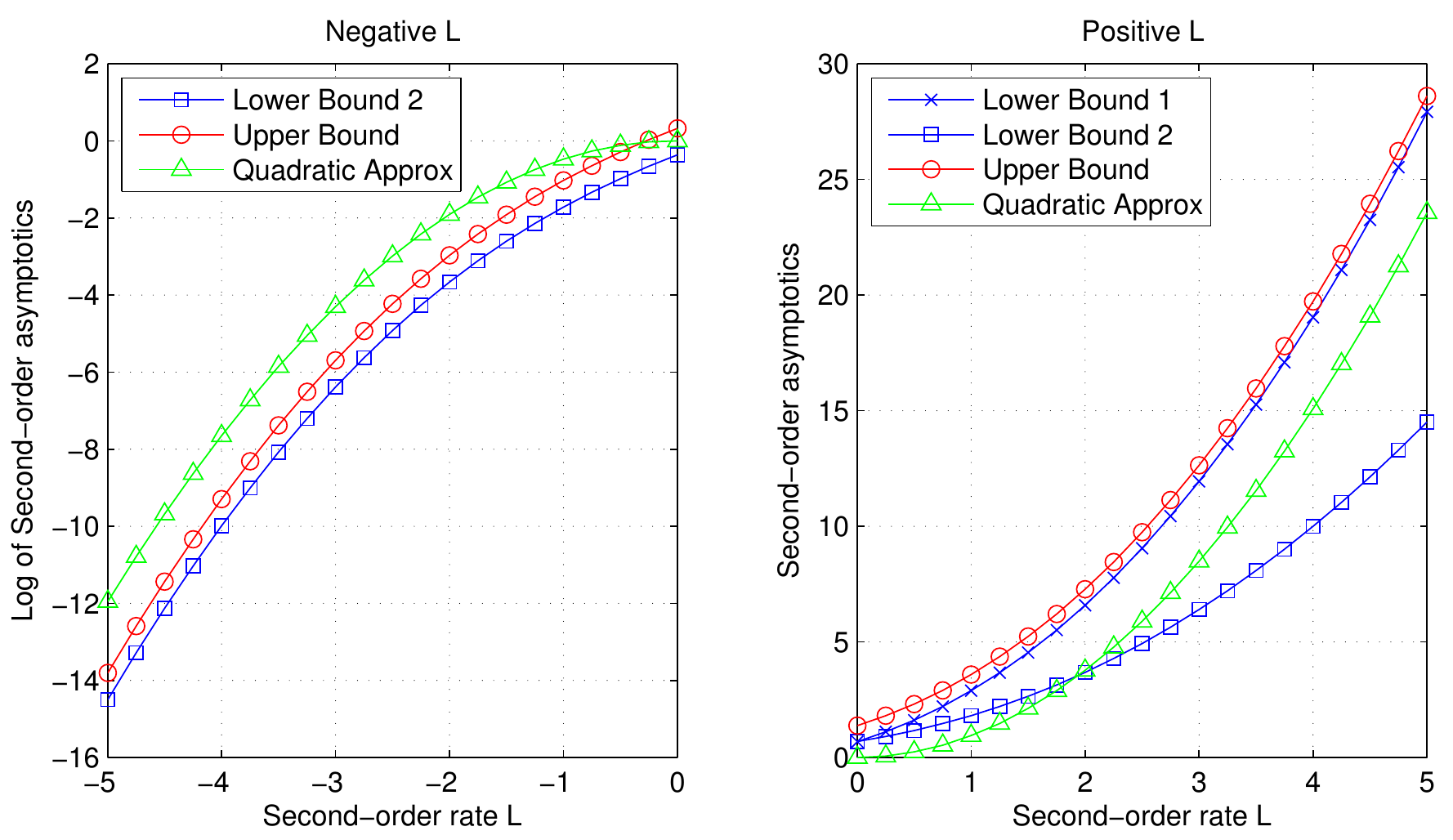}
\caption{Illustration of the bounds on the second-order asymptotics of $C_{1/2}^\uparrow(f_{X_n}(A^n)|E^n X_n|P_{AE}^n \times P_{X_n})$ (i.e., $s=1/2$) in Case (C) given by the left- and right-hand-sides of~\eqref{eqn:caseC2}   for the same source. For $L\le 0$,  lower bound 1 in \eqref{eqn:caseC2} is negative (lower bound 2 dominates) so is not shown in the left plot. Observe the quadratic behaviors--this is corroborated by \eqref{eqn:Linf_g} and \eqref{eqn:L-inf_g} in Theorem~\ref{thm:second_large}. These plots are obtained by using numerical integration to calculate the integral in $\Psi_1(s,L)$ in~\eqref{eqn:caseC2}.}
\label{fig:sec_caseC2}
\end{figure} 

\subsection{Approximations for Large Second-Order Coding Rates}

Since for Case (C) we only have bounds, we now examine the behavior of the bounds in the limit of  large $|L|$ for which we can show  tight results up to the quadratic terms. 
\begin{theorem}[Large Second-Order Rates]\label{thm:second_large}
  Assume that  $f_{X_n}:\calA^n\to \calM_n=\{1,\ldots, M_n\}$ is a 
random hash function. For Case (C) in Theorem \ref{thm:second} (R\'enyi parameter $\alpha=1-s$ where $s\in (0,1]$), we have the following asymptotic results as $L\to\infty$:
\begin{align}
\lim_{n\to\infty}  \inf_{f_{X_n}}    C_{1-s}(f_{X_n}(A^n)|E^n X_n|P_{AE}^n \times P_{X_n})  &=\frac{L^2}{2s V (A|E|P_{AE}) } + O(\log L) , \label{eqn:Linf} \\*
\lim_{n\to\infty}  \inf_{f_{X_n}}    C_{1-s}^\uparrow(f_{X_n}(A^n)|E^n X_n|P_{AE}^n \times P_{X_n})  &= \frac{1-s}{2s}\cdot \frac{L^2}{ V_1(A|E|P_{AE}) +  V_2(A|E|P_{AE})(1-s)} +O(\log L).\label{eqn:Linf_g}
\end{align}
Furthermore, we have the following asymptotic results as $L\to - \infty$:
\begin{align}
\lim_{n\to\infty} \log \inf_{f_{X_n}}    C_{1-s}(f_{X_n}(A^n)|E^n X_n|P_{AE}^n \times P_{X_n})  &= -\frac{L^2}{2V(A|E|P_{AE}) }+O(\log |L|) , \label{eqn:L-inf} \\*
\lim_{n\to\infty} \log \inf_{f_{X_n}}    C_{1-s}^\uparrow(f_{X_n}(A^n)|E^n X_n|P_{AE}^n \times P_{X_n})  &= -\frac{L^2}{2V(A|E|P_{AE}) }+O(\log |L|)  \label{eqn:L-inf_g}  .
\end{align}
The infima in \eqref{eqn:Linf}--\eqref{eqn:L-inf_g} are achieved by any sequence of $\epsilon$-almost universal$_2$ hash functions $f_{X_n}$.
\end{theorem}
The proof of Theorem~\ref{thm:second_large} can be found in Section~\ref{sec:prf_thm_large}.

The results in Theorem~\ref{thm:second_large} are somewhat analogous and similar those in  the study of the {\em moderate-deviations} asymptotics in information theory~\cite{altug14b, pol10e,AWK13, Tan12,TWH14}.
Note the difference between the results in \eqref{eqn:Linf}--\eqref{eqn:Linf_g} ($L\to\infty$) versus \eqref{eqn:L-inf}--\eqref{eqn:L-inf_g} ($L\to-\infty$). The former pair of results resembles the equivocation results presented in Section~\ref{sec:equiv} since the effective rate is {\em $(L/\sqrt{n})$-higher} than the conditional R\'enyi entropy and there is no logarithm preceding $C_{1-s}$ and $C_{1-s}^\uparrow$. The latter pair of results resembles the exponent results of Section~\ref{sec:exponent} since the effective rate is  {\em $(|L|/\sqrt{n})$-lower} than the conditional R\'enyi entropy  and there is a logarithm preceding $C_{1-s}$ and $C_{1-s}^\uparrow$. So the results presented in Theorem~\ref{thm:second_large} are  natural in view of the equivocation result in Theorem~\ref{thm:equiv} and the exponent result in Theorem~\ref{thm:exponents}. 
 
 \subsection{Proof of Theorem~\ref{thm:second_large}} \label{sec:prf_thm_large}
 In this section, we present the proof of Theorem~\ref{thm:second_large}. Since there are four statements in \eqref{eqn:Linf}--\eqref{eqn:L-inf_g}, we partition the proof into four  distinct subsections.
 
\paragraph{Proof of \eqref{eqn:Linf}}
When $L\to\infty$, the term $\Gamma_1(s,L)$ in \eqref{eqn:defA}  behaves as $-\frac{1}{s} \log \Phi\big(-\frac{L}{ \sqrt{V(A|E|P_{AE}) }}\big)+O(1)$, and 
attains the maximum in the lower bound in  \eqref{eqn:caseC1} 
because  $\Gamma_2(s,L)$ has the additional factor $1 - s$, which is smaller than $1$. 
   Also see  the right plot of  Fig.~\ref{fig:sec_caseC}.  Thus, in this limiting regime, the lower bound matches the upper bound in \eqref{eqn:caseC1}, namely $-\frac{1}{s} \log \Phi\big(-\frac{L}{ \sqrt{V(A|E|P_{AE}) }}\big)$, up to a constant term, i.e., 
\begin{equation}
\lim_{n\to\infty}  \inf_{f_{X_n}}    C_{1-s}(f_{X_n}(A^n)|E^n X_n|P_{AE}^n \times P_{X_n})  = -\frac{1}{s} \log \Phi\bigg( -\frac{L}{\sqrt{V(A|E|P_{AE}) } }\bigg) +O(1) 
\label{eqn:case1_calc}
\end{equation}
where $O(1)$ denotes a term bounded in $L$ (but dependent on $s$). 
Now by employing the asymptotic equality
\begin{equation}
\Phi(-t) =1-\Phi(t)\sim \frac{\rme^{-t^2/2}}{\sqrt{2\pi} t} ,\qquad \mbox{as  } t\to\infty, \label{eqn:Q_func}
\end{equation}
  we obtain from \eqref{eqn:case1_calc} that 
\begin{equation}
\lim_{n\to\infty}  \inf_{f_{X_n}}    C_{1-s}(f_{X_n}(A^n)|E^n X_n|P_{AE}^n \times P_{X_n})  =  \frac{L^2}{2s {V(A|E|P_{AE})   } } + O(\log L),
\end{equation}
which proves~\eqref{eqn:Linf}. 

\paragraph{Proof of \eqref{eqn:Linf_g}}
 When $L\to\infty$, the term $\Psi_1(s,L)$ dominates  the maximum in the lower bound in~\eqref{eqn:caseC2}  because $V_2\le V$ and thus the integrands in $\Psi_1(s,L)$, namely $ \Phi\big(- \frac{L+x}{\sqrt{ V_2(A|E|P_{AE})} }\big)^{\frac{1}{1-s}} \frac{ \rme^{-x^2/(2 V_1(A|E|P_{AE}))}}{\sqrt{2\pi V_1(A|E|P_{AE})}}$,  are not smaller than $\Phi\big( -\frac{L}{\sqrt{V(A|E|P_{AE}) } } \big)$ in $\Psi_2(s,L)$. See right plot of Fig.~\ref{fig:sec_caseC2}. We can then find the $x$ that dominates the integral in $\Psi_1(s,L)$. We denote this by $x^*$. Since $L$ is large, by \eqref{eqn:Q_func}, 
\begin{align}
\log\bigg[\Phi\bigg(- \frac{L+x}{\sqrt{ V_2 } }\bigg)^{\frac{1}{1-s}} \rme^{-x^2/(2 V_1 )} \bigg]-\bigg[ - \frac{1}{2(1-s)} \left( \frac{L+x}{\sqrt{V_2}}\right)^2  -\frac{x^2}{2V_1} \bigg]\to 0 ,\quad \mbox{as  } L\to\infty. 
\end{align}
Differentiating the quadratic, we obtain
\begin{equation}
x^*=-\frac{L V_1(A|E|P_{AE})}{V_1(A|E|P_{AE})+(1-s) V_2(A|E|P_{AE})}. \label{eqn:x_star}
\end{equation}
The exponential term $\rme^{-(x^*)^2/(2V_1(A|E|P_{AE}))}$ controls the behavior of the integral in $\Psi_1(s,L)$ and substituting  \eqref{eqn:x_star} into this exponential term  yields \eqref{eqn:Linf_g}.

\paragraph{Proof of \eqref{eqn:L-inf}}
Now we assume that $L\to -\infty$. In this case, we find that the term $\Gamma_2(s,L)$ attains the maximum in the lower bound in~\eqref{eqn:caseC1} because $\Gamma_1(s,L)$ is negative due to the constant negative term.   
  Also see  the left plot of  Fig.~\ref{fig:sec_caseC}.    In this case, taking the logarithm, we have  
\begin{align}
&\lim_{n\to\infty}  \log \inf_{f_{X_n}}    C_{1-s}(f_{X_n}(A^n)|E^n X_n|P_{AE}^n \times P_{X_n}) \nn\\*
&= \log\left[ -\log \left(1-\Phi\bigg(  \frac{L}{\sqrt{V(A|E|P_{AE}) } }\bigg)\right) \right]+O(1 )  \label{eqn:mod1}\\
&= \log\left[ \Phi\bigg(  \frac{L}{\sqrt{V(A|E|P_{AE}) } }\bigg) \right]+O(\log |L| ) \label{eqn:logx_approx} \\
&=-\frac{L^2}{2  {V(A|E|P_{AE})   } } + O(\log |L| )\label{eqn:mod3}
\end{align}
where in~\eqref{eqn:mod1}, $\log(\frac{1-s}{s})$ and $\log(\frac{1}{s})$ can be regarded as $O(1)$ when $ L\to-\infty$, in \eqref{eqn:logx_approx}, we used the fact that $\log(1-t) = - t + O(t^2)$ when $t\downarrow 0$, and finally in \eqref{eqn:mod3}, we used \eqref{eqn:Q_func}. This proves~\eqref{eqn:L-inf}.


\paragraph{Proof of \eqref{eqn:L-inf_g}}
In the other direction, when $L\to-\infty$, we claim that the term $\Psi_2(s,L)$ attains the maximum. This is shown as follows: First, we find that 
\begin{equation}
\left[1-\Phi\bigg( \frac{L+x}{\sqrt{V_2(A|E|P_{AE}) } } \bigg)\right]^{\frac{1}{1-s}}\ge 1-\frac{1}{1-s}\Phi\bigg( \frac{L+x}{\sqrt{V_2(A|E|P_{AE}) } } \bigg) . \label{eqn:use_Phi_prop} 
\end{equation}
This is because $a\mapsto a^{\frac{1}{1-s}}$ is convex and so the linear approximation underestimates the function. 
This means that 
\begin{align}
&\log\int_{-\infty}^{\infty} \Phi\bigg(- \frac{L+x}{\sqrt{ V_2(A|E|P_{AE})} }\bigg)^{\frac{1}{1-s}} \frac{\rme^{-x^2/(2 V_1(A|E|P_{AE}))}}{\sqrt{2\pi V_1(A|E|P_{AE})}} \,\rmd x \nn\\*
&\quad\ge \log\int_{-\infty}^{\infty} \left[1-\frac{1}{1-s}\Phi\bigg( \frac{L+x}{\sqrt{V_2(A|E|P_{AE}) } } \bigg) \right] \frac{\rme^{-x^2/(2 V_1(A|E|P_{AE}))}}{\sqrt{2\pi V_1(A|E|P_{AE})}} \,\rmd x \label{eqn:plug_Phi} \\*
&\quad= \log\left[ 1-\frac{1}{1-s}\Phi\bigg(  \frac{L }{\sqrt{V (A|E|P_{AE}) } } \bigg) \right] \label{eqn:convolution}\\
&\quad =- \frac{1}{1-s}\Phi\bigg(   \frac{L }{\sqrt{V (A|E|P_{AE}) } } \bigg) + O\left(\rme^{-{L^4 }/ {(4V (A|E|P_{AE})^2 )} }  \right)  \label{eqn:log_prop} \\
&\quad =-\frac{1}{1-s}\log \left[ 1-\Phi\bigg(   \frac{L }{\sqrt{V (A|E|P_{AE}) } } \bigg) \right]+ O\left(\rme^{-{L^4 }/ {(4V (A|E|P_{AE})^2 )} }  \right)\label{eqn:s_prop}\\
&\quad =-\frac{1}{1-s}\log  \Phi\bigg(   -\frac{L }{\sqrt{V (A|E|P_{AE}) } } \bigg) + O\left(\rme^{-{L^4 }/ {(4V (A|E|P_{AE})^2 )} }  \right), \label{eqn:end_chain}
\end{align}
where \eqref{eqn:convolution}  follows because the convolution of two independent zero-mean Gaussians is a Gaussian where the variances add and we also note that $V=V_1+V_2$ per~\eqref{eqn:sum_vars}.  
 This argument was also used in the second-order analysis of channels with state   \cite[Lemma~18]{TomTan13a}. Inequalities~\eqref{eqn:log_prop} and \eqref{eqn:s_prop}  follow  from the fact that $\log(1-x) = -x + O(x^2)$ as $x\downarrow 0$ (note that $L\to-\infty$ so the term $\Phi( L/\sqrt{V})$ tends to zero).  
Hence, \eqref{eqn:end_chain}  and the definitions of $\Psi_2(s,L)$ and $\Psi_1(s,L)$  (in \eqref{eqn:defD} and \eqref{eqn:defE} resp.) imply that $\Psi_2(s,L)$ asymptotically dominates $\Psi_1(s,L)$ as $L\to-\infty$. Also see left plot of Fig.~\ref{fig:sec_caseC2}. By a similar calculation as in \eqref{eqn:mod1}--\eqref{eqn:mod3}, we obtain the lower bound to \eqref{eqn:L-inf_g} as follows:
\begin{align}
&\liminf_{n\to\infty} \log \inf_{f_{X_n}}    C_{1-s}^\uparrow(f_{X_n}(A^n)|E^n X_n|P_{AE}^n \times P_{X_n})  \nn\\*
&\quad\ge \log\left[-\frac{1-s}{s}\log \Phi\bigg( -\frac{L}{\sqrt{V(A|E|P_{AE}) } } \bigg)\right] +O(1) \label{eqn:constant_term1}\\
&\quad=\log\left[-\frac{1-s}{s} \left(-\Phi\bigg(  \frac{L}{\sqrt{V(A|E|P_{AE}) } } \bigg) \right)\right] + O(\log |L| ) ) \label{eqn:constant_term0}\\
&\quad = \log\left[ \Phi\bigg(  \frac{L}{\sqrt{V(A|E|P_{AE}) } } \bigg)  \right] + O(\log |L| ) \label{eqn:constant_term} \\
&\quad= -\frac{L^2}{2V(A|E|P_{AE}) }+O(\log |L| ) , \label{eqn:lb_L}
\end{align}
where in  \eqref{eqn:constant_term1}, we used the above observation that $\Psi_2(s,L)=\Omega(\Psi_1(s,L))$ as $L\to-\infty$, in \eqref{eqn:constant_term0}, we wrote $\Phi(-L/\sqrt{V})=1-\Phi( L/\sqrt{V})$ and used the fact that $\log(1-x) = -x + O(x^2)$ as $x\downarrow 0$,  in \eqref{eqn:constant_term} we used the fact that $\log[ (1-s)/s]=O(1)$, and finally in \eqref{eqn:lb_L} we used the  approximation  in \eqref{eqn:Q_func}.
To show that the upper bound in \eqref{eqn:caseC2} matches the lower bound given by~\eqref{eqn:lb_L} (when $L\to -\infty$), we  use \eqref{eqn:use_Phi_prop} and steps similar to those in \eqref{eqn:plug_Phi}--\eqref{eqn:end_chain} to assert that 
\begin{align}
&-\frac{1-s}{s}\log\int_{-\infty}^{\infty} \Phi\bigg(- \frac{L+x}{\sqrt{ V_2(A|E|P_{AE})} }\bigg)^{\frac{1}{1-s}} \frac{\rme^{-x^2/(2 V_1(A|E|P_{AE}))}}{\sqrt{2\pi V_1(A|E|P_{AE})}} \,\rmd x \nn\\*
& \quad\le -\frac{1-s}{s}\log \left[ 1- \frac{1}{1-s}\Phi\bigg( \frac{L }{\sqrt{V (A|E|P_{AE}) } } \bigg) \right] \label{eqn:conv}\\
&\quad=  \frac{1}{s} \Phi\bigg( \frac{L }{\sqrt{V (A|E|P_{AE}) } } \bigg) + O\left(\rme^{-{L^4 }/ {(4V (A|E|P_{AE})^2 )} }  \right) ,\label{eqn:conv_2}
\end{align}
where in \eqref{eqn:conv_2}, we again used the fact that $\log(1-x)=-x+O(x^2)$ for $x\downarrow 0$.
 Now taking the logarithm and the limit  as $L\to-\infty$, we match the lower bound in \eqref{eqn:lb_L} completing the proof of~\eqref{eqn:L-inf_g}.

\section{One-Shot Bounds}\label{sec:one-shot}
To prove  Theorems~\ref{thm:equiv},  \ref{thm:exponents} and \ref{thm:second}, we leverage the following one-shot (i.e., blocklength $n$ equal to $1$) bounds. The proofs of these one-shot bounds are rather technical and hence we provide them in the appendices.

\subsection{One-Shot Bounds for the Direct Parts}
For the direct parts of the equivocation results, we evaluate the following one-shot bounds. The first two bounds in \eqref{eqn:os_direct_pluss} and \eqref{eqn:os_direct_G_pluss} can be considered as  generalizations of the bounds by Hayashi in \cite{Hayashi11} where $\epsilon=1$.
\begin{lemma} \label{lem:os_d_1}
For an ensemble of $\epsilon$-almost universal$_2$ hash functions $f_X : \calA\to\calM=\{1,\ldots, M\}$, we have for $s\in [0,1]$, 
\begin{align}
 \rme^{sC_{1+s} (f_X(A) | E X | P_{AE}\times P_X)}   &\le\epsilon^s + M^s \rme^{-s H_{1+s} (A|E|P_{AE})}  ,\label{eqn:os_direct_pluss}\\
 \rme^{\frac{s}{1+s} C_{1+s}^\uparrow (f_X(A) | E X | P_{AE}\times P_X) }&\le\epsilon^{\frac{s}{1+s}} + M^{\frac{s}{1+s}} \rme^{-{\frac{s}{1+s}} H_{1+s}^\uparrow (A|E|P_{AE})}. \label{eqn:os_direct_G_pluss}  
 \end{align}
 In the other direction with $s\in [0,1]$,
 \begin{align}
&  \rme^{-sC_{1-s} (f_X(A) | EX | P_{AE} \times P_X))}  \nn\\*
  &\ge 2^{-s}\sum_{(a,e) : P_{A|E}(a|e)\ge \frac{\epsilon }{M}  } P_{AE}(a,e) P_{A|E}(a|e)^{-s} M^{-s} 
+ 2^{-s}\sum_{(a,e) : P_{A|E}(a|e)< \frac{\epsilon }{M} } P_{AE}(a,e) \epsilon^{-s} , \label{eqn:os_direct_minuss} \\
&  \rme^{- \frac{s}{1-s}C_{1-s}^\uparrow (f_X(A) | E X | P_{AE}\times P_X))}  \nn\\*
    &\ge \frac{1}{2M^{\frac{s}{1-s}}}\sum_e P_E(e) \Big( \sum_{a: P_{A|E}(a|e)\ge \frac{\epsilon}{M}}P_{A|E}(a|e)^{1-s}\Big)^{\frac{1}{1-s}}   + (2\epsilon)^{-\frac{s}{1-s}} \sum_e P_E(e)\Big(\sum_{a:P_{A|E}(a|e) <\frac{\epsilon}{M}}P_{A|E}(a|e) \Big)^{\frac{1}{1-s}}. \label{eqn:os_direct_G_minuss} 
\end{align}
\end{lemma}
For the direct parts of the exponents results, we evaluate the following   one-shot bound. 
\begin{lemma}\label{lem:os_d_2}
For an ensemble of universal$_2$ hash functions $f_X : \calA\to\calM=\{1,\ldots, M\}$, we have  for any  $s\in [0,1]$, 
\begin{align}
 \rme^{\frac{s}{1+s} C_{1+s}^\uparrow (f_X(A)  | E X | P_{AE}\times P_X )}  \le  1 + \frac{1}{1  +   s}M^{s} \rme^{-sH_{1+s}(A|E|P_{AE})}.\label{eqn:os_direct_exp} 
\end{align}
\end{lemma} 

For the direct parts of the second-order results, we evaluate the following   one-shot bound. 

\begin{lemma}\label{lem:os_d_3}
For an ensemble of an $\epsilon$-almost universal$_2$ hash functions $f_X : \calA\to\calM=\{1,\ldots, M\}$, we have  for  any $s\in [0,1]$ and $c>0$, 
\begin{align}
 \rme^{-sC_{1-s} (f_X(A) | EX | P_{AE} \times P_X )} & \ge P_{AE}\Big\{ (a,e) : P_{A|E}(a|e) \le\frac{c}{M}\Big\} \Big( \frac{1}{c+\epsilon}\Big)^{s}, \label{eqn:one-shot-dir-second-1}\\
  \rme^{-\frac{s}{1-s} C_{1-s}^\uparrow (f_X(A) | EX | P_{AE} \times P_X )}  &\ge  \Big( \frac{1}{c+\epsilon}\Big)^{\frac{s}{1-s}}\sum_e P_E(e) \bigg(\sum_{a:P_{A|E}(a|e)\le\frac{c}{M}}P_{A|E}(a|e) \bigg)^{\frac{1}{1-s}}\label{eqn:one-shot-dir-second-2}  .
\end{align}
\end{lemma}
\subsection{One-Shot Bounds for the Converse Parts}
For the converse parts of the equivocation results, we evaluate the following   one-shot bounds.
\begin{lemma} \label{lem:os_c_1}
Fix $c>1$ and $s\ge 0$. Any hash function $f:\calA\to\calM=\{1,\ldots, M\}$ satisfies
\begin{align}
\rme^{-sC_{1-s}(f(A) | E | P_{AE}  )}&\le c^{-s}\sum_e P_E(e) \sum_{a:P_{A|E}(a|e) \ge\frac{c}{M}}P_{A|E}(a|e)^{1-s} M^{-s}   \nn\\*
 &\qquad\qquad+ 2^{\frac{s}{1-s}}s^{\frac{s}{1-s} }(1-s) P_{AE} \Big\{(a,e) : P_{A|E} (a|e)\le\frac{c}{M} \Big\} .\label{eqn:os_conv_eq} 
\end{align}
For the Gallager-type counterpart, 
\begin{align}
\rme^{-\frac{s}{1-s}C_{1-s}^\uparrow(f(A) | E | P_{AE} )}&\le 2^{\frac{s}{1-s}} \sum_e P_E(e) \bigg[ \Big(c^{-s}  \sum_{a:P_{A|E}(a|e) \ge\frac{c}{M}} P_{A|E}(a|e)^{1-s} M^{-s}  \Big)^{\frac{1}{1-s}}\nn\\*
&\qquad \qquad+ \Big(2^{\frac{s}{1-s}}s^{\frac{s}{1-s} }(1-s) P_{A|E=e}\Big\{a : P_{A|E}(a|e) <\frac{c}{M} \Big\}\Big)^{\frac{1}{1-s}}  \bigg].\label{eqn:os_conv_eqG}
\end{align}
\end{lemma}

For the converse parts of the exponents results and the   second-order results for the R\'enyi parameter being $1+s$ (with $s$ nonnegative), we evaluate the following    one-shot bounds.
\begin{lemma}\label{lem:os_c_2}
Fix $c> 1$ and $s\in [0,1]$.   Any hash function $f:\calA\to\calM=\{1,\ldots, M\}$ satisfies
\begin{align}
\rme^{-s C_{1-s}(f(A) | E  | P_{AE}  )}
& \le 
\sum_{ (a,e):P_{A|E} (a|e) \ge \frac{c}{M}}
P_E(e)P_{A|E}(a|e)^{1-s} M^{-s}
+\sum_e P_E(e) P_{A|E=e}\Big\{ a : P_{A|E} (a|e) < \frac{c}{M}\Big\}^{1-s}
\label{10-20-1b} \\
& \le 
P_{A,E}\Big\{ (a,e):P_{A|E} (a|e) \ge \frac{c}{M}\Big\} c^{-s}
+
P_{A,E}\Big\{ (a,e):P_{A|E} (a|e) < \frac{c}{M}\Big\}^{1-s} .
\label{10-20-1} 
\end{align}
 For the Gallager-type counterpart, for $s\in [0,1)$,
 \begin{align}  
 \rme^{-\frac{ s}{1-s}C_{1-s}^\uparrow(f(A) | E | P_{AE} )}   &\le \sum_e P_E(e)  \bigg[ P_{A|E=e}\Big\{a: P_{A|E}(a|e)\ge\frac{c}{M}\Big\}c^{-s} \nn\\*
 &\qquad \qquad+ P_{A|E=e}\Big\{a: P_{A|E}(a|e) < \frac{c}{M}\Big\}^{1-s}  \bigg]^{\frac{1}{1-s}} \label{eqn:os_conv_expG}.
\end{align}
In the other direction,  for $s\in [0,1]$, we have
\begin{align}
\rme^{s C_{1+s}(f(A) | E  | P_{AE}  )}
& \ge 
\sum_{ (a,e):P_{A|E} (a|e) \ge \frac{c}{M}}
P_E(e)P_{A|E}(a|e)^{1+s} M^{s} \nn\\*
&\qquad\qquad+\sum_e P_E(e) P_{A|E=e}
\Big\{ (a,e):P_{A|E} (a|e) < \frac{c}{M}\Big\}^{1+s}
\label{10-20-2b} \\
& \ge 
P_{AE}\Big\{ (a,e):P_{A|E} (a|e) \ge \frac{c}{M}\Big\} c^{s}
+
P_{AE}\Big\{ (a,e):P_{A|E} (a|e) < \frac{c}{M}\Big\}^{1+s}
\label{10-20-2}  .
\end{align}
 For the Gallager-type counterpart, for $s\in [0,1]$, we have
 \begin{align}
 \rme^{\frac{s}{1+s} C_{1+s}^\uparrow(f(A) | E  | P_{AE}  )}& \ge 
\sum_{ e} P_E(e) \bigg[ P_{A|E=E} \Big\{ a:P_{A|E}(a|e) \ge\frac{c}{M}\Big\}c^s   \nn\\*
&\qquad\qquad+ P_{A|E=E} \Big\{ a:P_{A|E}(a|e) <\frac{c}{M}\Big\}^{1+s}\bigg]^{\frac{1}{1+s}}
\label{10-20-2_205}  . 
 \end{align}
\end{lemma}
For the converse parts of the second-order results for the R\'enyi parameter being $1 - s$ (with $s$ nonnegative), we need the following one-shot bound  as well as  \eqref{eqn:os_conv_eqG} and \eqref{eqn:os_conv_expG} although  the converse parts of the second-order results with R\'enyi parameter being $1+s$ require  \eqref{10-20-2b} and \eqref{10-20-2_205}.

\begin{lemma}\label{lem:os_c_3}
Fix $c>1$ and $s\in [0,1]$. Any hash function
$f:{\cal A} \to {\cal M}=\{1, \ldots, M\}$
satisfies 
\begin{equation}
\rme^{-sC_{1-s}(f(A) | E | P_{AE}  )} \le 
c^{-s}P_{AE}\Big\{ (a,e) : P_{A|E}(a|e) \ge\frac{c}{M}\Big\}  
+ 2^{\frac{s}{1-s}}s^{\frac{s}{1-s} } P_{AE} \Big\{(a,e) : P_{A|E} (a|e)\le\frac{c}{M}\Big\} \label{eqn:os_conv_sec}  .
\end{equation}
\end{lemma}

\section{Proofs of the Asymptotic Results} \label{sec:prfs_asymp}
In this section, we prove the asymptotic results in Theorems~\ref{thm:equiv}, \ref{thm:exponents},  and~\ref{thm:second}.

{\em Notation:} Throughout, we let  $\ba=(a_1,a_2,\ldots,a_n)\in\calA^n$ and $\be=(e_1,e_2,\ldots,e_n)\in\calE^n$ denote deterministic  length-$n$ strings. We also let $A^n=(A_1,A_2,\ldots, A_n)$ and $E^n=(E_1,E_2,\ldots, E_n)$ denote random vectors of length $n$. We adopt the exponential equality notation: $a_n\doteq b_n$   if and only if  $\lim_{n\to\infty}\frac{1}{n}\log\frac{a_n}{b_n}= 0$.

Given a random variable $X$ with distribution (probability mass function) $P$,  we denote the expectation of a function of the random variable $g(X)$ by $\bbE [g(X)] =\sum_x P(x)g(x)$. If we want to make the dependence of the expectation on $X$ or $P$ explicit, we write $\bbE_X[g(X)]$ or $\bbE_P[g(X)]$. The same comment applies to the variance operator which we denote interchangeably  as $\var  [g(X)]$, $\var_X [g(X)]$ or $\var_P [ g(X)]$. 
\subsection{Proof of Theorem \ref{thm:equiv}}\label{sec:prf_equiv}

\subsubsection{Direct Parts}
We first prove the direct parts (upper bounds). 
\paragraph{Proof of the upper bound of~\eqref{eqn:C1s}}
The bound in \eqref{eqn:os_direct_pluss}  implies that 
\begin{align}
& C_{1+s} (f_{X_n}(A^n)|E^nX_n|P_{AE}^n\times P_{X_n}))  \nn\\*
&\le \frac{1}{s} \log\big(\epsilon^s+M_n^s \rme^{-sH_{1+s}( A^n  | E^n|P_{AE}^n)}\big)\label{eqn:equiv_direct_prf4}\\
&= \frac{1}{s} \log\big(\epsilon^s+M_n^s \rme^{-nsH_{1+s}( A   | E |P_{AE} )}\big) . \label{eqn:equiv_direct_prf3}
\end{align} 
For $\epsilon$ being a constant,  this achieves the upper bound of~\eqref{eqn:C1s} upon normalizing by $n$ and taking the $\limsup$.   

\paragraph{Proof of the upper bound of~\eqref{eqn:C1sG}}
The bound in~\eqref{eqn:os_direct_G_pluss} implies that  
\begin{align}
&  C_{1+s}^\uparrow (f_{X_n}(A^n)|E^nX_n|P_{AE}^n\times P_{X_n}))  \nn\\*
&\le \frac{1+s}{s} \log\big(\epsilon^{\frac{s}{1+s} }+M_n^{\frac{s}{1+s} } \rme^{-\frac{s}{1+s}  H_{1+s}(  A^n | E^n|P_{AE}^n)}\big)\label{eqn:equiv_direct_prf_g_4}\\
&= \frac{1+s}{s} \log\big(\epsilon^{\frac{s}{1+s} }+M_n^{\frac{s}{1+s} } \rme^{-n\frac{s}{1+s}  H_{1+s}(  A| E |P_{AE} )}\big).\label{eqn:equiv_direct_prf_g_3}
\end{align}
This   leads to the upper bound of \eqref{eqn:C1sG} for constant $\epsilon$ upon normalizing by $n$ and taking the $\limsup$.   

\paragraph{Proof of the upper bound of~\eqref{eqn:C1s_minus}}
To obtain \eqref{eqn:C1s_minus}, we  employ Cram\'er's theorem~\cite{Dembo} on the sequence of random variables   $-\log P_{A|E}^n(A^n|E^n)=\sum_{i=1}^n -\log P_{A|E}(A_i|E_i)$ under the product joint distribution $P_{AE}^n$. It is easy to see by using exponential tail bounds that
\begin{align}
&\lim_{n\to\infty}-\frac{1}{n}\log\sum_{(\ba,\be): P_{A|E}^n(\ba|\be) <   \rme^{-nR}}P_{AE}^n(\ba,\be)  \nn\\*
 &\quad=\lim_{n\to\infty}-\frac{1}{n}\log P_{AE}^n \left\{ (\ba,\be): \frac{1}{n}\sum_{i=1}^n\log P_{A|E} (a_i|e_i)\le -R\right\}\label{eqn:cramer0}\\
&\quad=\max_{ t\ge 0} t(R-H_{1-t} (A|E|P_{AE})).  \label{eqn:cramer1}
\end{align}
Note that the cumulant generating function of the random variable $-\log P_{A|E}(A|E)$ under the joint distribution $P_{AE}$ can be expressed in terms of the conditional R\'enyi entropy as
\begin{equation}
 \log\bbE_{  P_{AE}}\left[ \rme^{  t (-\log  P_{A|E}(A|E))} \right] = t H_{1-t}(A|E|P_{AE}),
\end{equation}
explaining the presence of this term in \eqref{eqn:cramer1}.
We again apply (a generalized version of) Cramer's theorem\footnote{The standard Cram\'er's theorem~\cite[Section 2.2]{Dembo} (or Sanov's theorem~\cite[Section 2.1]{Dembo}) is a  large-deviations result concerning the exponent  of $P^n(\calB)$ where $P$ is a {\em probability} measure and $\calB$ is an event in the sample space $\Omega$.  If $P$ is not necessarily a probability measure but a finite non-negative measure (as it is in our applications), say $\mu$,    Cram\'er's theorem  clearly also applies by defining the  new {\em probability} measure $\calB\mapsto \widetilde{P}(\calB) := \mu(\calB)/\mu(\Omega)$.  \label{fn:cramer} }  to  the sequence of random variables   $\log P_{A|E}^n(A^n|E^n)$  under the sub-distribution (non-negative product measure) $P_{AE}^n(\ba,\be)(P_{A|E}^{-s})^n(\ba|\be)$ and event $\{(\ba,\be):\log P_{A|E}^n(\ba|\be) \ge  -nR\}$. Note that   the cumulant generating function in this case  is 
\begin{align}
 \tau_s(t)& := \log\sum_{a,e} P_{AE}(a,e) P_{A|E}^{-s}(a|e) \exp \left(t\log  {P_{A|E}(a|e)}\right)\\
&  =  (s-t) H_{1-(s-t)}(A|E|P_{AE}),
\end{align}
and by direct differentiation, we also have that 
\begin{equation}
\tau_s'(0)=-\hatR_{-s}
\end{equation}
where $\hatR_s$ is defined in \eqref{eqn:crit_rate1} (cf.\ $\hatR_{-s}$ is presented in a different form in \eqref{eqn:minus_equivalence}).  Thus, by Cram\'er's theorem,
\begin{align}
& \sum_{(\ba,\be): P_{A|E}^n(\ba|\be) \ge \rme^{-nR}}  P_{AE}^n(\ba,\be)P_{A|E}^n(\ba|\be)^{-s}  \doteq \frac{ \bbE[ \rme^{t \log P_{A|E}(A|E)} ] }{\rme^{-tnR} } \label{eqn:expectation}\\
& \qquad = \exp\left[ -n \left(- tR - \log \sum_{a,e} P_{AE}(a,e) P_{A|E}(a|e)^s \rme^{t \log P_{A|E}(A|E)}  \right)\right]\\
&\qquad = \exp\left[ -n \left(- tR  +(t-s) H_{1+(t-s) }(A|E|P_{AE}) \right)\right],\label{eqn:expectation3}
\end{align}
where in \eqref{eqn:expectation}, the ``expectation'' $\bbE$ is taken with respect to the non-negative measure $(a,e)\mapsto P_{AE}(a,e)P_{A|E}(a|e)$. 
Since $t\ge 0$ is arbitrary,
\begin{align}
&-\frac{1}{n}\log\sum_{(\ba,\be): P_{A|E}^n(\ba|\be) \ge \rme^{-nR}}  P_{AE}^n(\ba,\be)P_{A|E}^n(\ba|\be)^{-s}   \nn\\*
&\qquad = \max_{t\ge 0 }   \big\{ - tR + (t-s) H_{1+(t-s)}(A|E|P_{AE})\big\} \label{eqn:app_cramer}
\end{align}
For the case where $R\ge\hatR_{-s}$, the constraint in the optimization above is active, i.e., $t^*=0$ because the function $s\mapsto \hatR_{-s}$ is monotonically non-decreasing as described in Section~\ref{sec:info_measures}. Conversely, when $R\le\hatR_{-s}$, the constraint is inactive, i.e., the maximum is realized with $R=\hatR_{-(s-t)}$.  Thus, we obtain  
\begin{align}
&\lim_{n\to\infty}-\frac{1}{n}\log \sum_{(\ba,\be): P_{A|E}^n(\ba|\be) \ge \rme^{-nR}}  P_{AE}^n(\ba,\be)P_{A|E}^n(\ba|\be)^{-s}\rme^{-snR} \nn\\*
&\qquad =\left\{ \begin{array}{cc}
s(R-H_{1-s}(A|E|P_{AE})) &\mbox{if }  R\ge \hatR_{-s}\\
\max_{t'\le s}t'(R-H_{1-t'}(A|E|P_{AE}) ) & \mbox{if }R\le \hatR_{-s}
\end{array} \right.  \label{eqn:cramer2} ,
\end{align} 
where the second clause follows by the substitution $t'=s-t$.  
Now with these preparations, we can employ the one-shot bound in~\eqref{eqn:os_direct_minuss} with $\epsilon=1$ to prove the direct part of \eqref{eqn:C1s_minus} as follows: Since \eqref{eqn:cramer2} is not greater than \eqref{eqn:cramer1}, the former dominates in the exponent  and we obtain
\begin{align}
&\limsup_{ n\to\infty} \frac{1}{n}C_{1-s} (f_{X_n}(A^n) | E^n X_n|P_{AE}^n\times P_{X_n}) \nn\\*
&=-\frac{1}{s}\liminf_{ n\to\infty} \frac{1}{n}\log\left[ \rme^{-s C_{1-s} (f_{X_n}(A^n) | E^nX_n|P_{AE}^n \times P_{X_n})} \right] \\
&\le -\frac{1}{s}\liminf_{ n\to\infty} \frac{1}{n}\log \bigg[ 2^{-s}  \sum_{(\ba,\be) : P_{A|E}^n(\ba|\be)  <  \rme^{-nR}}   P_{AE}^n (\ba,\be) \nn\\*
&\qquad\qquad+ 2^{-s}  \sum_{(\ba,\be) : P_{A|E}^n(\ba|\be)  \ge \rme^{-nR}}  P_{AE}^n(\ba,\be) P_{A|E}^n(\ba|\be)^{-s} \rme^{-snR} \bigg], \label{eqn:use_one_shot}
\end{align}
where  \eqref{eqn:use_one_shot} follows from~\eqref{eqn:os_direct_minuss}. Now we combine the asymptotic results in \eqref{eqn:cramer1} and \eqref{eqn:cramer2} to evaluate  the asymptotic behavior of~\eqref{eqn:use_one_shot}. In particular, we take into consideration the scaling factor $\frac{1}{s}$. We also note that the domain of maximization of $t$ in \eqref{eqn:cramer1}  and $t'$ in the second clause of \eqref{eqn:cramer2} are $[0,\infty)$ and $(\infty, s]$ respectively. So the intersection of these domains is $[0,s]$ and the eventual max should thus be taken over $[0,s]$. Uniting these observations, we obtain that  the upper bound 
\begin{align}
&\limsup_{ n\to\infty} \frac{1}{n}C_{1-s} (f_{X_n}(A^n) | E^n X_n|P_{AE}^n\times P_{X_n}) \nn\\*
&\qquad \le\left\{ \begin{array}{cc}
 R-H_{1-s}(A|E|P_{AE} ) &\mbox{if }  R\ge \hatR_{-s}\\
\max_{t \in [0, s]} \frac{t}{s} (R-H_{1-t}(A|E|P_{AE}) ) & \mbox{if }R\le \hatR_{-s}
\end{array} \right. .  \label{eqn:intersection_domains}
\end{align}

\paragraph{Proof of the upper bound of~\eqref{eqn:C1sG_minus}}
 The upper bound of \eqref{eqn:C1sG_minus} proceeds in an analogous manner. It proceeds in five distinct steps, each detailed in the following five paragraphs. 
 
In Step 1, we manipulate  the one-shot bound in \eqref{eqn:os_direct_G_minuss} with $\epsilon=1$ as follows: 
\begin{align}
&\limsup_{ n\to\infty} \frac{1}{n}C_{1-s}^\uparrow (f_{X_n}(A^n) | E^n X_n|P_{AE}^n\times P_{X_n}) \nn\\* 
&=-\frac{1-s}{s}\liminf_{ n\to\infty} \frac{1}{n}\log \left[\rme^{- \frac{s}{1-s} C_{1-s}^\uparrow (f_{X_n}(A^n) | E^nX_n|P_{AE}^n \times P_{X_n})} \right]\\* 
&\le-\frac{1-s}{s}\liminf_{n\to\infty} \log\bigg[ \frac{1}{2}\rme^{-\frac{s}{1-s} nR} \sum_{\be}P_E^n(\be) \bigg(\sum_{\ba: P_{A|E}^n(\ba|\be)\ge\epsilon \rme^{-nR}} P_{A|E}^n(\ba|\be)^{1-s} \bigg)^{\frac{1}{1-s}}\nn  \\*
&\qquad\qquad 
 + \frac{1}{2\epsilon^{\frac{s}{1-s}}}\sum_{\be}P_E^n(\be) \bigg(\sum_{\ba: P_{A|E}^n(\ba|\be) < \epsilon \rme^{-nR}} P_{A|E}^n(\ba|\be)^{1-s} \bigg)^{\frac{1}{1-s}}  \bigg]   \label{eqn:1shot_split} .
\end{align} 
In the following two steps, we evaluate the first and second terms in the $\liminf$ in \eqref{eqn:1shot_split}. 

In Step 2, we evaluate the second term in the $\liminf$ in \eqref{eqn:1shot_split} as it is simpler and provides the intuition and techniques for evaluating the first term. For this, we need to employ   the   G\"artner-Ellis theorem~\cite{Dembo} (instead of Cramer's theorem). Doing so to the sequence of  random variables $-\log P_{A|E}^n(A^n|\be) = \sum_{i=1}^n-\log P_{A|E}(A_i|e_i)$  with $\be$ of fixed type~\cite{Csi97} and $A^n$ with the  memoryless distribution $P_{A|E}^n(\cdot|\be)$, as will be shown in the following, we obtain
 \begin{align}
& \lim_{n\to\infty}-\frac{1}{n}\log\sum_{\be} P_E^n(\be) \bigg( \sum_{\ba: P_{A|E}^n(\ba|\be) < \epsilon \rme^{-nR}} P_{A|E}^n(\ba|\be)\bigg)^{\frac{1}{1-s}} \nn\\*
& \qquad=   \max_{t\ge 0} \,\, \frac{t}{1-s} \left( R - H_{1-t| 1-s} (A|E|P_{AE})  \right) , \label{eqn:use_cramer}
 \end{align}
 where $H_{1-t| 1-s} (A|E|P_{AE})$ is the two-parameter conditional R\'enyi entropy defined in \eqref{eqn:two_param}.   To show \eqref{eqn:use_cramer}, consider $\be\in \calT_{Q} = \{ \be \in\calE^n: \mathrm{type}(\be)=Q\}$. Then the G\"artner-Ellis theorem~\cite{Dembo} yields that
\begin{equation}
  P_{A^n|E^n=\be} \Big\{\ba: P_{A|E}^n(\ba|\be) < \epsilon \rme^{-nR} \Big\}  \doteq \exp \left(- n\max_{t\ge 0}\bigg[tR  - \bbE_Q  \log \sum_{a} P_{A|E}^{1-t} (a|E) \bigg] \right),
 \end{equation} 
 where 
 \begin{equation}
 \bbE_Q  \log \sum_{a} P_{A|E}^{1-t} (a|E) =\sum_e Q(e)  \log \sum_{a} P_{A|E}^{1-t} (a|e). \label{eqn:bbEQ}
 \end{equation}
Let $\calP_n(\calE)$ be the set of $n$-types with alphabet $\calE$. Splitting the sum  on the left-hand-side in \eqref{eqn:use_cramer}   into the polynomially many $n$-types on $\calE$, we obtain
 \begin{align}
&\sum_{\be \in\calE^n} P_E^n(\be) \bigg( \sum_{\ba: P_{A|E}^n(\ba|\be) < \epsilon \rme^{-nR}} P_{A|E}^n(\ba|\be)\bigg)^{\frac{1}{1-s}}  \nn\\*
&\doteq\sum_{Q\in\calP_n(\calE) } P_E^n(\calT_Q) \exp \left(- \frac{n}{1-s}\max_{t\ge 0}\bigg[tR  - \bbE_Q  \log \sum_{a} P_{A|E}^{1-t} (a|E) \bigg]\right) \label{eqn:split_type}\\
&\doteq  \max_{Q \in\calP_n(\calE)}\exp\big(-n D( Q\| P_E)\big)\exp \left(- \frac{n}{1-s}\max_{t\ge 0}\bigg[tR  -\bbE_Q  \log \sum_{a} P_{A|E}^{1-t} (a|E) \bigg] \right)\label{eqn:type_cl}\\
&\doteq \exp\left(-n \min_{Q} \max_{t\ge 0} \bigg[ \frac{tR}{1-s}  - \frac{\bbE_Q  \log \sum_{a} P_{A|E}^{1-t} (a|E) }{1-s} +D(Q\| P_E) \bigg]\right)\\
& =\exp\left(-n \max_{t\ge 0}  \min_{Q}\bigg[ \frac{tR}{1-s}  - \frac{\bbE_Q  \log \sum_{a} P_{A|E}^{1-t} (a|E) }{1-s} +D(Q\| P_E) \bigg]\right)  \label{eqn:split_ty} 
 \end{align}
 where \eqref{eqn:type_cl} follows from the fact that $P_E^n(\calT_Q)\doteq \exp(-nD( Q\| P_E))$ \cite[Ch.~2]{Csi97},  the swapping of $\min$ and $\max$ in~\eqref{eqn:split_ty}  follows from the  fact that the objective function is convex  and concave in $Q$ and  $t$ respectively, $Q$ resides in a compact, convex set (the probability simplex) and $t$ resides in a convex  set $[0,\infty)$ (Sion's minimax theorem~\cite{sion56}). Now 
by straightforward calculus, the optimizing distribution  for fixed $t$ is 
 \begin{equation}
 Q^*  (e) = \frac{P_E(e) \big( \sum_a P_{A|E}^{1-t}(a|e) \big)^{\frac{1}{1-s} }}{Z_{s,t} }\end{equation}
 where the normalizing constant   (partition function)
 \begin{equation}
  Z_{s,t}:= \sum_{e } P_E(e ) \Big( \sum_a P_{A|E}^{1-t}(a|e )  \Big)^{\frac{1}{1-s}}.
 \end{equation}
Plugging this into \eqref{eqn:split_ty} we obtain
\begin{align}
\sum_{\be \in\calE^n} P_E^n(\be) \bigg( \sum_{\ba: P_{A|E}^n(\ba|\be) < \epsilon \rme^{-nR}} P_{A|E}^n(\ba|\be)\bigg)^{\frac{1}{1-s}}  \doteq \exp\left(  -n\max_{t\ge 0}\frac{t}{1-s} \bigg[ R-\frac{1-s}{t}Z_{s,t} \bigg]\right) \label{eqn:after_plugging}
\end{align}
which then yields \eqref{eqn:use_cramer}. Note that we have to use the G\"artner-Ellis theorem  (and not Cramer's theorem) because the collection of random variables $\{-\log P_{A|E}(A_i| e_i) : i=1,\ldots, n\}$ is independent but {\em not identically distributed}. 

In Step 3, we evaluate the first term in the $\liminf$ in \eqref{eqn:1shot_split} again by applying the G\"artner-Ellis theorem~\cite{Dembo} to the sequence of random variables $\log P_{A|E}^n(A^n|\be) =\sum_{i=1}^n \log P_{A|E}(A_i| e_i)$ with non-negative measure $P_{A|E}^n(\cdot|\be)^{1-s}$, we have
\begin{align}
\sum_{\ba: P_{A|E}^n(\ba|\be) \ge \epsilon \rme^{-nR} } 
P_{A|E}^n(\ba|\be)^{1-s} \doteq \exp\left( -n \max_{t\ge 0}\bigg[-tR -  \bbE_Q \log\sum_a P_{A|E}^{1-(s-t)}(a|E)\bigg]\right)
\end{align}
where $Q$ is the type of $\be$ and $\bbE_Q \log\sum_a P_{A|E}^{1-(s-t)}(a|E)$ is defined in  \eqref{eqn:bbEQ}. 
So by using a type partitioning argument of sequences $\be$ similarly to~\eqref{eqn:split_type}--\eqref{eqn:after_plugging}, we obtain
\begin{align}
&\sum_{\be} P_E^n(\be)  
\bigg(\sum_{\ba: P_{A|E}^n(\ba|\be) \ge \epsilon \rme^{-nR} } 
P_{A|E}^n(\ba |\be)^{1-s}\bigg)^{\frac{1}{1-s}}    \nn\\*
&\qquad\doteq   \exp\left(  -n \max_{t\ge 0} \bigg[ -\frac{tR}{1-s} +\log\sum_e P_E(e)\Big( \sum_a P_{A|E}^{1-(s-t)}(a|e)\Big)^{ \frac{1}{1-s}}\bigg]\right).
\end{align}
Consequently, considering the two different cases similarly to~\eqref{eqn:cramer2}, we obtain
\begin{align}
& \lim_{n \to \infty}-\frac{1}{n}\log  \left[
\rme^{-\frac{s}{1-s}nR}\sum_{\be} P_E^n(\be)  
\bigg(\sum_{\ba: P_{A|E}^n(\ba|\be) \ge \epsilon \rme^{-nR} } 
P_{A|E}^n(\ba|\be)^{1-s}\bigg)^{\frac{1}{1-s}} \right] \nn\\*
&\qquad = 
\left\{
\begin{array}{ll}
\frac{s}{1-s}(R- H_{1-s}^{\uparrow}(A|E|P_{A E}))
& \hbox{ if }
R \ge  \frac{\rmd}{\rmd t} \,  tH_{1-t|1-s}(A|E|P_{A E}) \big|_{t=s} \\
\max_{t  \in [0,s]} \frac{t}{1-s} (R- H_{1-t|1-s}(A|E|P_{A E}))
& \hbox{ if }
R \le \frac{\rmd}{\rmd t} \,  tH_{1-t|1-s}(A|E|P_{A E}) \big|_{t=s} 
\end{array} 
\right.  . \label{eqn:use_ge}
\end{align}

In Step 4, we put together the asymptotic  evaluations in~\eqref{eqn:use_cramer}  and~\eqref{eqn:use_ge} in the bound in \eqref{eqn:1shot_split}. We observe that \eqref{eqn:use_cramer} is not smaller than   \eqref{eqn:use_ge}. Thus, the former dominates the exponential behavior, and  plugging \eqref{eqn:use_cramer} into \eqref{eqn:1shot_split}, we obtain
\begin{align}
&\limsup_{ n\to\infty} \frac{1}{n}C_{1-s}^\uparrow (f_{X_n}(A^n) | E^n X_n|P_{AE}^n\times P_{X_n}) \nn\\* 
&\le\left\{
\begin{array}{ll}
 R- H_{1-s}^{\uparrow}(A|E|P_{A E})
& \hbox{ if }
R \ge  \frac{\rmd}{\rmd t} \,  tH_{1-t|1-s}(A|E|P_{A E}) \big|_{t=s} \\
\max_{t\in [0,s]} \frac{t}{s} (R- H_{1-t|1-s}(A|E|P_{A E}))
& \hbox{ if }
R \le \frac{\rmd}{\rmd t} \,  tH_{1-t|1-s}(A|E|P_{A E}) \big|_{t=s} 
\end{array}
\right. \label{eqn:split_two_case_G}  .
\end{align}

Finally in Step 5, we   show  that the transition rate in \eqref{eqn:split_two_case_G} 
\begin{equation}
\frac{\rmd}{\rmd t} \,  tH_{1-t|1-s}(A|E|P_{A E}) \big|_{t=s}  = \hatR_{-s}^\uparrow \label{eqn:hatR_der}
\end{equation}
as follows.
Since $\max_{t}H_{1-s|1-t}(A|E|P_{A E})=
H_{1-s|1-s}(A|E|P_{A E})$ (i.e., the maximum is attained at $t=s$),
\begin{align}
\frac{\rmd}{\rmd t} \, s H_{1-s|1-t}(A|E|P_{A E})\big|_{t=s} 
=s \frac{\rmd}{\rmd t} H_{1-s|1-t}(A|E|P_{A E})\big|_{t=s} =0. \label{eqn:zero_der}
\end{align}
Hence, choosing $t_1=t_2=t$,
we have 
\begin{align}
  \hatR_{-s}^\uparrow&=\frac{\rmd}{\rmd t} \, t H_{1-t}^\uparrow(A|E|P_{A E}) \big|_{t=s} \\
&=\frac{\rmd}{\rmd t} \, t H_{1-t|1-t}(A|E|P_{A E}) \big|_{t=s}  \label{eqn:use_equivalence}\\
&=
\frac{\rmd t_1}{\rmd t}
\frac{\partial}{\partial t_1} 
\, t_1 H_{1-t_1|1-t_2}(A|E|P_{A E}) \big|_{t=s}  
+
\frac{\rmd t_2}{\rmd t}
\frac{\partial}{\partial t_2} 
\, t_1 H_{1-t_1|1-t_2}(A|E|P_{A E}) \big|_{t=s}  
 \label{eqn:chain_rule} \\
&=\frac{\rmd}{\rmd t} \, t H_{1-t|1-s}(A|E|P_{A E}) \big|_{t=s}
+\frac{\rmd}{\rmd t} \, s H_{1-s|1-t}(A|E|P_{A E}) \big|_{t=s}  \label{eqn:split_derivatives}\\
&=\frac{\rmd}{\rmd t} \, t H_{1-t|1-s}(A|E|P_{A E}) \big|_{t=s} \label{eqn:use_zero_der} ,
\end{align}
where \eqref{eqn:use_equivalence} follows from the relation in~\eqref{eqn:minus_equivalence2}, \eqref{eqn:chain_rule} follows from the chain rule,  and \eqref{eqn:use_zero_der} follows from~\eqref{eqn:zero_der}. 
The relations in \eqref{eqn:split_two_case_G}  and \eqref{eqn:hatR_der} complete  the justification of the upper bound of~\eqref{eqn:C1sG_minus}. 
\subsubsection{Converse Parts}

For the converse, we do not consider the common randomness $X_n$ (i.e., $X_n=\emptyset$) since the bound must hold for {\em all} (not just $\epsilon$-almost  universal$_2$) hash functions $f_{X_n}$. This statement applies to the proofs of all converse bounds in the sequel. 

\paragraph{Proofs of the lower bounds of~\eqref{eqn:C1s} and \eqref{eqn:C1sG}}
The lower bounds to \eqref{eqn:C1s} and \eqref{eqn:C1sG} can be easily obtained by using the data processing inequalities for R\'enyi conditional  entropies and their Gallager-type counterparts in \eqref{eqn:dpi1}--\eqref{eqn:dpi3}.

\paragraph{Proof of the lower bound  of~\eqref{eqn:C1s_minus}}
Now for \eqref{eqn:C1s_minus}, we note that when $R\ge \hatR_{-s}$, we have 
\begin{align}
   s(R -  H_{1-s}(A|E|P_{AE})) \le \max_{t\in [0,s]} t(R -  H_{1-t}( A|E|P_{AE}))   \label{eqn:R_one_dir}
\end{align}
and when $R\le \hatR_{-s}$, equality holds since $t=s$ attains the maximum. Fix $t\in [0,s]$. From~\eqref{eqn:os_conv_eq},  we obtain the bound
\begin{align}
\rme^{-s C_{1-s}(f(A^n)|E^n|P_{A E}^n)}
& \le 
c^{-s}
\sum_{\be} P_{E^n}(\be) \sum_{\ba: P_{A|E}^n(\ba|\be) \ge \frac{c}{M_n} }
 P_{A|E}^n(\ba|\be)^{1-s} M_n^{-s}
\nn\\*
&\qquad\qquad+ 
2^{\frac{s}{1-s}} s^{\frac{s}{1-s}} (1-s)
P_{A E}^n \Big\{(\ba,\be) :  P_{A|E}^n(\ba|\be) < \frac{c}{M_n} \Big\}  \\
& \le 
c^{-s}
\rme^{s H_{1-s}(A^n|E^n|P_{A E}^n)} M_n^{-s}
+ 
2^{\frac{s}{1-s}} s^{\frac{1}{1-s}-1} (1-s)
\rme^{t H_{1-t}(A^n|E^n|P_{A E}^n)} \Big(\frac{M_n}{c}\Big)^{-t}, \label{eqn:apply_markov}
\end{align}
where in   \eqref{eqn:apply_markov}, we upper bounded the probability in the second term using Markov's inequality, i.e.,  for any $t\in [0,s]$, 
\begin{align}
P_{A E}^n \Big\{(\ba,\be) :  P_{A|E}^n(\ba|\be) < \frac{c}{M_n} \Big\} &=P_{A E}^n \Big\{(\ba,\be) : \rme^{-t\log P_{A|E}^n(\ba|\be)}  > \rme^{-t\log \frac{c}{M_n}} \Big\} \\
&\le \frac{ \bbE_{P_{A E}^n}\left[ \rme^{-t\log  P_{A|E}^n(A^n|E^n)}  \right]}{\rme^{-t\log \frac{c}{M_n}} }\\
&= \rme^{t H_{1-t}(A^n|E^n|P_{A E}^n)} \Big(\frac{M_n}{c}\Big)^{-t}. \label{eqn:apply_markov2}
\end{align}
 Put $c=1$ in \eqref{eqn:apply_markov}. We then obtain the lower bound to~\eqref{eqn:C1s_minus} by applying \eqref{eqn:R_one_dir} and its equality version for $R\le \hatR_{- s}$.

\paragraph{Proof of the lower bound  of~\eqref{eqn:C1sG_minus}}
Finally, \eqref{eqn:C1sG_minus} can be obtained by evaluating~\eqref{eqn:os_conv_eqG} as follows:
\begin{align}
&\rme^{-\frac{s}{1-s} C_{1-s}^\uparrow(f(A^n)|E^n|P_{A E}^n)} \nn\\
 &\le  2^{\frac{s}{1-s}}\sum_{\be} \bigg[P_{E^n}(\be)
\Big(c^{-s}\sum_{\ba: P_{A|E}^n(\ba|\be)\ge \frac{c}{M_n} }
 P_{A|E}^n(\ba|\be)^{1-s} M_n^{-s} \Big)^{\frac{1}{1-s}}\nn\\*
& \qquad\qquad+ 
\left(
2^{\frac{s}{1-s}} s^{\frac{s}{1-s}}
 (1-s)
P_{A^n|E^n=\be} \Big\{\ba : P_{A|E}^n(\ba|\be) < \frac{c}{M_n} \Big\} \right)^{\frac{1}{1-s}}  \bigg]\\
 &\le 2^{\frac{s}{1-s}} \bigg[
c^{-s}
\rme^{\frac{s}{1-s} H_{1-s}^{\uparrow}(A^n|E^n|P_{A E}^n)} M_n^{-\frac{s}{1-s}} \nn\\*
&\qquad\qquad+ 
2^{\frac{s}{(1-s)^2}} s^{\frac{s}{(1-s)^2}}
 (1-s)^{\frac{1}{1-s}}
\sum_{\be}P_E^n(\be)
\rme^{\frac{t}{1-s} H_{1-t}(A^n|P_{A^n|E^n=\be} )}
\Big(\frac{M_n}{c} \Big)^{-\frac{t}{1-s}}  \bigg]\\
&= 2^{\frac{s}{1-s}}\bigg[
c^{-s}
\rme^{\frac{s}{1-s} H_{1-s}^{\uparrow}(A^n|E^n|P_{A E}^n)} M_n^{-\frac{s}{1-s}}
+ 
 2^{\frac{s}{(1-s)^2}} s^{\frac{s}{(1-s)^2}} (1-s)^{\frac{1}{1-s}}
\rme^{\frac{t}{1-s} H_{1-t|1-s}(A^n|E^n|P_{A E}^n)}
\Big(\frac{M_n}{c}\Big)^{-\frac{t}{1-s}} \bigg] \label{eqn:equiv_G_lb}
\end{align}
with $s\ge t\ge 0$. For brevity, let $\beta_s:=2^{\frac{s}{(1-s)^2}} s^{\frac{s}{(1-s)^2}} (1-s)^{\frac{1}{1-s}}$ be a function that only depends on $s$. By taking the logarithm of \eqref{eqn:equiv_G_lb}, normalizing by $n$, and using $M_n =\rme^{nR}$, we obtain
\begin{align}
&\frac{1}{n}C_{1-s}^\uparrow(f(A^n)|E^n|P_{A E}^n)  \nn\\*
& \ge - \frac{1-s}{sn }\log \Bigg[ 2^{\frac{s}{1-s}}\bigg[
c^{-s}
\rme^{\frac{s}{1-s} H_{1-s}^{\uparrow}(A^n|E^n|P_{A E}^n)}M_n^{-\frac{s}{1-s}}
+ 
 \beta_s
\rme^{\frac{t}{1-s} H_{1-t|1-s}(A^n|E^n|P_{A E}^n)}
\Big(\frac{M_n}{c}\Big)^{-\frac{t}{1-s}} \bigg] \Bigg]\\
& = - \frac{1-s}{sn }\log \Bigg[ 2^{\frac{s}{1-s}}\bigg[
c^{-s}
\rme^{-n \frac{s}{1-s}(R-  H_{1-s}^{\uparrow}(A |E |P_{A E} ))}
+ 
 \beta_s  c^{\frac{t}{1-s}} 
\rme^{-n \frac{t}{1-s}(R- H_{1-t|1-s}(A |E |P_{A E} ))} \bigg] \Bigg] \label{eqn:equiv_G_lb2}  .
\end{align}
Now similarly to \eqref{eqn:R_one_dir}, we have 
\begin{equation}
   s(R -  H_{1-s}^\uparrow(A|E|P_{AE})) \le \max_{t\in [0,s]} t(R -  H_{1-t|1-s}( A|E|P_{AE}))  .  \label{eqn:R_one_dir2}
\end{equation}
Thus when $R\ge\hatR_{-s}^\uparrow$,  the first term in \eqref{eqn:equiv_G_lb2}  dominates (exponent is not larger),  resulting in the   first  clause in~\eqref{eqn:C1sG_minus}, namely $R-H_{1-s}^\uparrow(A|E|P_{AE})$. On the other hand, when $R\le\hatR_{-s}^\uparrow$, the second term  in \eqref{eqn:equiv_G_lb2}   dominates. To complete the argument, we optimize over $t\in [0,s]$ to obtain the second clause in \eqref{eqn:C1sG_minus}, namely $\max_{t\in [0,s]} \frac{t}{s}(R- H_{1-t|1-s } (A|E|P_{AE}) )$. This completes the proof of the lower bound of~\eqref{eqn:C1sG_minus}.


\subsection{Proof of Theorem \ref{thm:exponents}}\label{sec:prf_exp}
\subsubsection{Direct Parts}
\paragraph{Proof of the lower bound  of~\eqref{eqn:exp1}}
We note, per the discussion following Theorem \ref{thm:exponents}, that 
\begin{align}
 C_{1+s}(f_{X_n}(A^n) | E^nX_n|P_{AE}^n\times P_{X_n})) & = O(n) ,\quad\mbox{ and }\\
 C_{1+s}^\uparrow(f_{X_n}(A^n) | E^nX_n|P_{AE}^n\times P_{X_n}))&=O(n).
\end{align}
Thus, the exponents are lower bounded by zero, explaining the $|\fndot|^+$ in \eqref{eqn:exp1} and \eqref{eqn:exp2}. 

Now for the non-trivial (non-zero) lower bound on the exponents, we employ \eqref{eqn:os_direct_pluss} with $\epsilon=1$ and $t\in [s,1]$. We recall that $M_n=\rme^{nR}$. Now   we have 
\begin{align}
&  C_{1+s}(f_{X_n}(A^n) | E^nX_n|P_{AE}^n\times P_{X_n}))\nn\\* 
&\le C_{1+t}(f_{X_n}(A^n) | E^nX_{n}|P_{AE}^n\times P_{X_n}))  \label{eqn:exp_prf_1} \\
& \le \frac{1}{t}\log \Big(1 +  M_n^t \rme^{-t H_{1+t} (A^n|E^n|P_{AE}^n)}   \Big)\\
& \le  \frac{1}{t}M_n^t\rme^{-n t H_{1+t} (A|E|P_{AE})}. \label{eqn:exp_prf_3}
\end{align}
Taking the logarithm and optimizing over $t\in [s,1]$, we obtain the lower bound to \eqref{eqn:exp1}. 

\paragraph{Proof of the lower bound  of~\eqref{eqn:exp2}}
Similarly, applying \eqref{eqn:os_direct_exp}  to the case $t\in [s,1]$, we obtain
\begin{align}
&  C_{1+s}^\uparrow(f_{X_n}(A^n) | E^nX_n|P_{AE}^n\times P_{X_n}))\nn\\* 
&\le    C_{1+t}^{\uparrow}(f_{X_n}(A^n) | E^n X_n|P_{AE}^n\times P_{X_n})) \label{eqn:exp_prf_g1}  \\
&\le  \frac{1+t}{t}\log\Big( 1+\frac{1}{1+t} M_n^t \rme^{-t H_{1+t}(A^n|E^n|P_{AE}^n)}\Big)\label{eqn:exp_prf_g2}\\
& \le  \frac{1}{t} M_n^t \rme^{-nt H_{1+t}(A|E|P_{AE})} \label{eqn:exp_prf_g3}
\end{align}
which implies the lower bound to \eqref{eqn:exp2} upon optimizing over $t\in [s,1]$. 

\paragraph{Proofs of the lower bounds of~\eqref{eqn:gal_exp1} and~\eqref{eqn:gal_exp2}}
For the $-s$ versions in \eqref{eqn:gal_exp1} and \eqref{eqn:gal_exp2}, we simply note that 
\begin{align}
C_{1+s}&\ge C_{1-s'}\\
C_{1+s}^\uparrow&\ge C_{1-s'}^\uparrow
\end{align}
 for any $s,s'\in [0,1]$ because as mentioned in Section~\ref{sec:info_measures} (after \eqref{eqn:renyi_ent_Q} and \eqref{eqn:Q_tilt} respectively), $H_{1+s}$ and $H_{1+s}^\dagger$ are monotonically 
 decreasing in $s$. Thus,   we have
\begin{align}
\liminf_{n\to\infty}-\frac{1}{n} \log C_{1-s'} (f_{X_n}(A^n) | E^nX_n|P_{AE}^n\times P_{X_n}))&\!\ge\!\liminf_{n\to\infty}-\frac{1}{n} \log C_{1+s} (f_{X_n}(A^n) | E^nX_n|P_{AE}^n\times P_{X_n})) \\
\liminf_{n\to\infty}-\frac{1}{n} \log C_{1-s'}^\uparrow (f_{X_n}(A^n) | E^nX_n|P_{AE}^n\times P_{X_n}))&\!\ge\!\liminf_{n\to\infty}-\frac{1}{n} \log C_{1+s}^\uparrow (f_{X_n}(A^n) | E^nX_n|P_{AE}^n\times P_{X_n}))  .
\end{align} 
Combining these statements with the bounds derived in \eqref{eqn:exp_prf_1}--\eqref{eqn:exp_prf_g3} completes the proof of \eqref{eqn:gal_exp1} and \eqref{eqn:gal_exp2}.

\subsubsection{Converse Parts}
\paragraph{Proof of the upper bound  of~\eqref{eqn:exp1}}
For the converse, we first show the upper bound to  \eqref{eqn:exp1}. Choose a constant $c_0$ satisfying $c_0^s > 1+s$. Recall the definition of $\hatR_s$ in \eqref{eqn:crit_rate1}. 
Now assume that $R\ge \hatR_s$. We claim that 
\begin{align}
\lim_{n\to\infty}  -\frac{1}{n}\log  \left[P_{AE}^n \left\{ (\ba,\be): P_{A|E}^n(\ba|\be )\ge c_0 \rme^{-n  \hatR_s} \right\} c_0^s \rme^{-sn \hatR_s} \rme^{snR}\right] = sH_{1+s}(A|E|P_{AE}) - sR . \label{eqn:crma}
\end{align}
This is justified as follows. We know from  Cram\'er's theorem~\cite{Dembo} that
\begin{equation}
P_{AE}^n \left\{ (\ba,\be): P_{A|E}^n(\ba|\be)\ge c_0 \rme^{-n  \hatR_s}\right\}\doteq \exp\left(-n \max_{t\ge 0} \left\{ tH_{1+t}(A|E|P_{AE}) -t   \hatR_s \right\}\right). \label{eqn:cramer_3}
\end{equation}
The maximum appeared in right-hand-side of \eqref{eqn:cramer_3} is attained when the derivative of $tH_{1+t}(A|E|P_{AE}) - t \hat{R}_s$ is zero because $tH_{1+t}(A|E|P_{AE})$ is concave in $t$.
Hence, the real number $t$ satisfies 
 \begin{equation}
\frac{\rmd }{\rmd t}tH_{1+t}(A|E|P_{AE})=   \hatR_s ,\label{eqn:differentiated}
\end{equation}
which implies $t=s \ge 0$ due to the definition of $ \hatR_s$ in \eqref{eqn:crit_rate1} and the strictly decreasing nature of $\hatR_s$.
As a result, 
\begin{equation}
P_{AE}^n \left\{ (\ba,\be):P_{A|E}^n(\ba|\be)\ge c_0 \rme^{-n  \hatR_s}\right\}\doteq\exp\left( -n \left\{sH_{1+s} (A|E|P_{AE})-s \hatR_s \right\}\right).
\end{equation}
Plugging this into the left-hand-side of \eqref{eqn:crma} yields the claim.
The one-shot bound in \eqref{10-20-2}  with $c=c_0 \rme^{-n \hatR_s+nR}>1$ implies that
\begin{align}
 & \rme^{sC_{1+s} (f(A^n) | E^n   | P_{AE}^n  ) } \nn\\
&\ge P_{AE}^n \left\{ (\ba,\be) : P_{A|E}^n(\ba|\be) \ge c_0 \rme^{-n  \hatR_s}  \right\}c_0^s \rme^{-n  \hatR_s}\rme^{snR}  \nn\\*
&\qquad\qquad + 1-(1+s) P_{AE}^n \left\{ (\ba,\be) : P_{A|E}^n(\ba|\be) \ge c_0 \rme^{-n  \hatR_s} \right\} \\
&=1+ \left(c_0^s \rme^{-sn \hatR_s}\rme^{snR}-1-s \right)  P_{AE}^n \left\{ (\ba,\be) : P_{A|E}^n(\ba|\be) \ge c_0 \rme^{-n  \hatR_s} \right\}. \label{eqn:esC}
\end{align} 
Hence, taking the logarithm of \eqref{eqn:esC}, employing the lower bound $\log(1+b)\ge b- \frac{b^2}{2}$, the large-deviations result \eqref{eqn:cramer_3}, and the fact that   $\lim_{n\to\infty}\frac{1}{n}\log (c_0^s \rme^{-sn \hatR_s+snR}-1-s)
=s( -\hatR_s+R)$, we obtain
\begin{align}
sC_{1+s} (f(A^n) | E^n   | P_{AE}^n  )\dotgeq \exp\left[ -ns (H_{1+s}(A|E|P_{AE}) -R) \right]. \label{eqn:sC}
\end{align}
Finally,  we obtain the upper bound to \eqref{eqn:exp1} by taking another logarithm and normalizing by $n$. 

 For the other case $R\le  \hatR_s$, we claim that 
\begin{align}
\max_{t\ge s}   \,\,\{tH_{1+t}(A|E|P_{AE})-tR\} = \max_{t\ge 0}\,\,\{ tH_{1+t}(A|E|P_{AE})-tR \} \label{eqn:eq_H} .
\end{align}
This is because by the  strict concavity of $t\mapsto tH_{1+t}$, the map $s\mapsto \hatR_s$ is strictly decreasing. So for $R\le  \hatR_s$ the maximum on the right-hand-side of \eqref{eqn:eq_H} is attained at some $t\ge s$.   
This is also reflected in Fig.~\ref{fig:exp}. 
Thus, \eqref{10-20-1}   implies that
\begin{align}
 & \rme^{sC_{1+s} (f(A^n) | E^n  | P_{AE}^n ) } \nn\\
& \ge  P_{AE}^n \left\{ (\ba,\be) : P_{A|E}^n(\ba|\be) \ge c_0 \rme^{-n R} \right\}c_0^s   \nn\\*
&\qquad\qquad+ \left(1-  P_{AE}^n \left\{ (\ba,\be) : P_{A|E}^n(\ba|\be) \ge c_0 \rme^{-n sR} \right\} \right)^{1+s} \\
&\ge 1+ (c_0^s-1-s)  P_{AE}^n \left\{ (\ba,\be) : P_{A|E}^n(\ba|\be) \ge c_0 \rme^{-n sR } \right\}  .
\end{align} 
Hence, 
\begin{align}
&C_{1+s} (f(A^n) | E^n | P_{AE}^n  )\nn\\*
& =\frac{1}{s}\log(\rme^{sC_{1+s} (f(A^n) | E^n  | P_{AE}^n  ) }) \\
&\ge\frac{1}{s}\log\left( 1+  (c_0^s-1-s)  P_{AE}^n \left\{ (\ba,\be) : P_{A|E}^n(\ba|\be) \ge c_0 \rme^{-n sR } \right\} \right) \\
&\doteq \frac{c_0^s-1-s}{s}P_{AE}^n \left\{ (\ba,\be) : P_{A|E}^n(\ba|\be)\ge c_0 \rme^{-nsR}  \right\}.  \label{eqn:final_combine_cramer}
\end{align}
By combining the asymptotic evaluation using Cram\'er's theorem in \eqref{eqn:cramer_3} and the equality in~\eqref{eqn:eq_H}, we see that for  $R\le  \hatR_s$, we also obtain the upper bound to \eqref{eqn:exp1}.  This completes the proof.

\paragraph{Proof of the upper bound  of~\eqref{eqn:exp2}}
The proof of the upper bound to \eqref{eqn:exp2} is similar and we present the details here.  Similarly to the above proof,  choose $c = c_0 \rme^{-n \hatR_s^\uparrow + nR }>1$ and the constant $c_0^s>1+s$.  Assume that $R\ge \hatR_s^\uparrow$, where  $\hatR_s^\uparrow$ is defined in \eqref{eqn:crit_rate2}. 
 Then  the one-shot bound in~\eqref{10-20-2_205}  implies that 
 \begin{align}
& \rme^{\frac{s}{1+s} C_{1+s}^{\uparrow} (f(A^n) | E^n  | P_{AE}^n  )} \nn\\
&\ge  
\sum_{\be} P_E^n(\be)
\Big(
P_{A^n|E^n=\be}\left\{ \ba: P_{A|E}^n (\ba|\be) \ge c_0 \rme^{-n  \hatR_s^\uparrow}
\right\} c_0^{s}\rme^{-s n \hatR_s^\uparrow} \rme^{s nR} \nn\\
&\qquad +
\left(1-
P_{A^n|E^n=\be}\left\{ \ba: P_{A|E}^n (\ba|\be) \ge c_0 \rme^{-n \hatR_s^\uparrow}
\right\}\right)^{1+s}
\Big)^{\frac{1}{1+s}} \\
&\ge 
\sum_{\be} P_E^n(\be)
\Big(
P_{A^n|E^n=\be}\left\{ \ba: P_{A|E}^n (\ba|\be) \ge c_0 \rme^{-n  \hatR_s^\uparrow}
\right\} c_0^{s}\rme^{-s n  \hatR_s^\uparrow} \rme^{s nR} \nn\\*
&\qquad+
1- (1+s)
P_{A^n|E^n=\be}\left\{ \ba: P_{A|E}^n (\ba|\be) \ge c_0 \rme^{-n \hatR_s^\uparrow}
\right\}
\Big)^{\frac{1}{1+s}} \\
&= 
\sum_{\be} P_E^n(\be)
\left(
1 +
\big(c_0^{s} \rme^{-s n  \hatR_s^\uparrow} \rme^{s nR}- (1+s)  \big)
P_{A^n|E^n=\be} \left\{ \ba: P_{A|E}^n (\ba|\be) \ge c_0 \rme^{-n  \hatR_s^\uparrow}
\right\}
\right)^{\frac{1}{1+s}} \\
&\doteq
\sum_{\be} P_E^n(\be)
\Big(
1 +
\frac{ c_0^{s}\rme^{-s n   \hatR_s^\uparrow} \rme^{s nR}- (1+s)}{1+s}
P_{A^n|E^n=\be}\left\{ \ba: P_{A|E}^n (\ba|\be)  \ge c_0 \rme^{-n  \hatR_s^\uparrow}
\right\}
\Big) \label{eqn:doteq_use} \\
&= 
1+
\frac{c_0^{s}\rme^{-s n \hatR_s^\uparrow} \rme^{s nR}- (1+s)}{1+s}
\sum_{\be} P_E^n(\be)
P_{A^n|E^n=\be}\left\{\ba: P_{A|E}^n (\ba|\be)  \ge c_0 \rme^{-n \hatR_s^\uparrow}
\right\},
\end{align}
where~\eqref{eqn:doteq_use} follows from the fact that $\log [(1+a)^t] =t[a+O(a^2)]$ for $a\downarrow 0$.  
 Hence,
\begin{align}
&C_{1+s}^{\uparrow}(f(A^n) | E^n  | P_{AE}^n  )\nn\\
&= \frac{1+s}{s}\log\left( \rme^{\frac{s}{1+s} C_{1+s}^{\uparrow}(f(A^n) | E^n X_n | P_{AE}^n\times P_{X_n} )} \right)\\
&\ge  
\frac{1+s}{s}\log \bigg(
1+
\frac{c_0^{s}\rme^{-s n \hatR_s^\uparrow} \rme^{s nR}- (1+s)}{1+s} 
\sum_{\be} P_E^n(\be)
P_{A^n|E^n=\be} \left\{ \ba: P_{A|E}^n (\ba|\be)  \ge c_0 \rme^{-n \hatR_s^\uparrow}
\right\}
\bigg) \label{11-12-10e_prev} \\
&\doteq
\frac{c_0^{s}\rme^{-s n \hatR_s^\uparrow} \rme^{s nR}- (1+s)}{s}
\sum_{\be} P_E^n(\be)
 P_{A^n|E^n=\be}\left\{ \ba :P_{A|E}^n (\ba|\be)  \ge c_0 \rme^{-n \hatR_s^\uparrow}
\right\} 
 \label{11-12-10e1},\\
 &=
\frac{c_0^{s}\rme^{-s n \hatR_s^\uparrow} \rme^{s nR}- (1+s)}{s}
P_{AE}^n\left\{ (\ba,\be) :P_{A|E}^n (\ba|\be)  \ge c_0 \rme^{-n \hatR_s^\uparrow}
\right\} 
 \label{11-12-10e},
\end{align}
where \eqref{11-12-10e} follows from $\log(1+a)=a+O(a^2)$ as $a\downarrow 0$ and the fact that the summation in~\eqref{11-12-10e_prev} vanishes as $n$ grows.
Combining \eqref{eqn:crma} and \eqref{11-12-10e} yields the upper bound to \eqref{eqn:exp2} for $R\ge \hatR_s^\uparrow$. A similar calculation for the case $R\le \hatR_s^\uparrow$  also yields the the same upper bound to \eqref{eqn:exp2}.

\paragraph{Proof of the upper bound  of~\eqref{eqn:gal_exp1}}
We choose the constant $c$ such that $(1-s)>c^{-s} $.
We apply Cramer's Theorem \cite{Dembo} to the sequence of random variables $\log P_{A|E}^n (A^n|E^n)$. 
Then,
\begin{align}
\lim_{n \to \infty}
-\frac{1}{n}\log P_{A E}^n \bigg\{ (\ba,\be):P_{A|E}^n (\ba|\be)  \ge \frac{c}{\rme^{nR}}\bigg\} 
=\max_{t \ge 0}\left\{ t H_{1+t}(A|E|P_{A E}) -t R  \right\}.
\label{11-12-6}
\end{align}
The one-shot bound in \eqref{10-20-1}  implies that
\begin{align}
&\rme^{-s C_{1-s}(f(A^n) | E^n  | P_{AE}^n  )}\nn\\*
& \le 
P_{A E}^n \Big\{ (\ba,\be):P_{A|E}^n (\ba|\be)  \ge \frac{c}{M }\Big\} c^{-s}
+
P_{A E}^n\Big\{ (\ba,\be):P_{A|E}^n (\ba|\be)  < \frac{c}{M }\Big\}^{1-s} \\
& =
P_{A E}^n\Big\{ (\ba,\be):P_{A|E}^n (\ba|\be)  \ge \frac{c}{M }\Big\} c^{-s}
+
\Big(1-P_{A E}^n\Big\{ (\ba,\be):P_{A|E}^n (\ba|\be) \ge \frac{c}{M }\Big\} \Big)^{1-s} \\
& \le
1- ( (1-s)-c^{-s} )P_{A E}^n\Big\{ (\ba,\be):P_{A|E}^n (\ba|\be)  \ge \frac{c}{M }\Big\}  . \label{eqn:basic_ineq}
\end{align}

Thus,
\begin{align}
& C_{1-s}(f(A^n) | E^n  | P_{AE}^n  )\nn\\*
& = - \frac{1}{s}\log \left[\rme^{-s C_{1-s}(f(A^n) | E^n  | P_{AE}^n  )}\right] \\
&\ge 
 - \frac{1}{s}\log 
\left[1- ( (1-s)-c^{-s} )P_{A E}^n\Big\{ (\ba,\be):P_{A|E}^n (\ba|\be)  \ge \frac{c}{M }\Big\} \right] \\
&\ge  
\frac{1}{s}
( (1-s)-c^{-s} )P_{A E}^n\Big\{ (\ba,\be):P_{A|E}^n (\ba|\be)   \ge \frac{c}{M }\Big\},
\label{11-12-5}
\end{align}
where the final step uses the inequality $\log(1-t)\le -t$.  Combining the limiting statement in~\eqref{11-12-6} and the bound in~\eqref{11-12-5}, we have the upper bound to~\eqref{eqn:gal_exp1}.

\paragraph{Proof of the upper bound  of~\eqref{eqn:gal_exp2}}
When 
$
\sum_{\be} P_E^n(\be)P_{A^n|E^n=\be}\big\{ \ba : P_{A|E}^n (\ba|\be)  \ge \frac{c}{M } \big\}
$  
is exponentially small, \eqref{eqn:os_conv_expG} implies that
\begin{align}
&\rme^{-\frac{s}{1-s} C_{1-s}^{\uparrow}(f(A^n) | E^n  | P_{AE}^n  )}\nn\\
& \le 
\sum_{\be} P_E^n(\be)
\bigg(
 P_{A^n|E^n=\be}\Big\{ \ba:P_{A|E}^n (\ba|\be)  \ge \frac{c}{M }\Big\} c^{-s} \nn\\*
&\qquad\qquad+
 P_{A^n|E^n=\be}\Big\{ \ba:P_{A|E}^n (\ba|\be)  < \frac{c}{M }\Big\}^{1-s}
\bigg)^{\frac{1}{1-s}} \\
&=
\sum_{\be} P_E^n(\be)
\bigg(
c^{-s} P_{A^n|E^n=\be}\Big\{ \ba:P_{A|E}^n (\ba|\be)  \ge \frac{c}{M }\Big\}  \nn\\*
&\qquad\qquad+
\Big(1- P_{A^n|E^n=\be}\Big\{ \ba:P_{A|E}^n (\ba|\be)  \ge \frac{c}{M }\Big\} \Big)^{1-s}
\bigg)^{\frac{1}{1-s}} \\
&\le 
\sum_{\be} P_E^n(\be)
\left(
1- (1-s-c^{-s})  P_{A^n|E^n=\be}\Big\{ \ba:P_{A|E}^n (\ba|\be)  \ge \frac{c}{M }\Big\} 
\right)^{\frac{1}{1-s}} \\
&\doteq
\sum_{\be} P_E^n(\be)
\left(
1- \frac{1-s-c^{-s}}{1-s}  P_{A^n|E^n=\be}\Big\{ \ba:P_{A|E}^n (\ba|\be)  \ge \frac{c}{M } \Big\}
\right) \label{eqn:same_doteq}\\
&=
1-\frac{1-s-c^{-s}}{1-s}
\sum_{\be}P_E^n(\be)
 P_{A^n|E^n=\be}\Big\{ \ba:P_{A|E}^n (\ba|\be)  \ge \frac{c}{M } \Big\},
\end{align}
where \eqref{eqn:same_doteq} follows from the same reasoning as~\eqref{eqn:doteq_use}.
Thus,
\begin{align}
& C_{1-s}^{\uparrow}(f(A^n) | E^n  | P_{AE}^n  )
= - \frac{1-s}{s}\log  \left[ \rme^{-\frac{s}{1-s} C_{1-s}^{\uparrow}(f(A^n) | E^n  | P_{AE}^n  )}  \right]\\
&\ge 
 - \frac{1-s}{s}\log 
\left[
1-
\frac{1-s-c^{-s}}{1-s}
\sum_{\be}P_E^n(\be)
 P_{A^n|E^n=\be}\Big\{ \ba:P_{A|E}^n (\ba|\be) \ge \frac{c}{M } \Big\} \right] \\
&\ge 
\frac{1-s}{s} \cdot \frac{1-s-c^{-s}}{1-s}
\sum_{\be}P_E^n(\be)
 P_{A^n|E^n=\be}\Big\{ \ba:P_{A|E}^n (\ba|\be)  \ge \frac{c}{M } \Big\} \\
&= 
\frac{1-s-c^{-s}}{s}
\sum_{\be} P_E^n(\be)
 P_{A^n|E^n=\be}\Big\{ \ba:P_{A|E}^n (\ba|\be)  \ge \frac{c}{M } \Big\}
\label{11-12-5c} .
\end{align}
Combining \eqref{11-12-6} and \eqref{11-12-5c}, we have the upper bound to \eqref{eqn:gal_exp2}.

\subsection{Proof of Theorem \ref{thm:second}}\label{sec:prf_sec}
\subsubsection{Direct Parts}
\paragraph{Proof of upper bounds for Case (A)}
First, we prove the upper bounds for Case (A) where the R\'enyi parameter $\alpha = 1+s$ for $s\in (0,1]$. Substituting $\rme^{ n H_{1+s}(A|E|P_{AE}) + \sqrt{n} L}$ into $M_n$ in the chain of inequalities in \eqref{eqn:equiv_direct_prf4}--\eqref{eqn:equiv_direct_prf3}, we obtain, for the class of   $\epsilon$-almost  universal$_2$ hash functions $f_{X_n}$, that
\begin{align}
C_{1+s} (f_{X_n}(A^n)|E^nX_n|P_{AE}^n\times P_{X_n}) \le \frac{1}{s}\log \big(\epsilon^s + \rme^{s\sqrt{n}L} \big) .
\end{align}
Set $\epsilon$ to be a constant (not varying with $n$). Normalizing by $\sqrt{n}$ and taking the $\limsup$ as $n\to\infty$ yields the upper bound to  \eqref{eqn:2order1}.

In an exactly analogous way, the upper bound to~\eqref{eqn:2order1_g} can be shown by substituting  $\rme^{ n H_{1+s}^\uparrow(A|E|P_{AE}) + \sqrt{n} L}$ into $M_n$ in the chain   of inequalities in~\eqref{eqn:exp_prf_g1}--\eqref{eqn:exp_prf_g3}.

Substituting  $\rme^{ n H_{1+s} (A|E|P_{AE}) + \sqrt{n} L}$ into $M_n$ in the   chain of inequalities in \eqref{eqn:exp_prf_1}--\eqref{eqn:exp_prf_3} with $t=s$, we obtain
\begin{equation}
C_{1+s} (f_{X_n}(A^n)|E^nX_n|P_{AE}^n\times P_{X_n})\le M_n^s \rme^{-snH_{1+s}(A|E|P_{AE})}=\rme^{s\sqrt{n}L}
\end{equation}
which implies the upper bound to \eqref{eqn:2order1_Lneg} after we take the logarithm,  normalize both sides by $\sqrt{n}$ and take the $\limsup$ as $n\to\infty$.  

In an exactly analogous way,  the upper bound to~\eqref{eqn:2order1_g_Lneg} can be shown by substituting  $\rme^{ n H_{1+s}^\uparrow(A|E|P_{AE}) + \sqrt{n} L}$ into $M_n$ in the chain of inequalities in~\eqref{eqn:equiv_direct_prf_g_4}--\eqref{eqn:equiv_direct_prf_g_3}. This completes the proof for the direct part of Case (A) of Theorem~\ref{thm:second}. 

\paragraph{Proof of upper bound  for Case (B)}
Case (B) follows from four distinct steps, detailed in each of the following paragraphs. 

In Step 1, we fix any function $f:\calA^n\to\{1,\ldots, \|f\|\}$. We  partition the space $\calA^n\times\calE^n$ into pairs of sequences of the same joint type~\cite{Csi97}. Let $Q_{AE}$ denote a generic joint  type on $\calA\times\calE$. Let $U^{( Q_{AE } )}$ be the uniform distribution over the type class $\calT_{ Q_{AE} } \subset\calA^n\times\calE^n$. 
Let 
\begin{equation}
 U^{( Q_{AE } )}_{f(A^n),E^n }(i,\be) := \sum_{\ba: f(\ba) =i } U^{( Q_{AE } )}(\ba,\be)
\end{equation}
be the distribution on $\{1,\ldots, \|f\|\}\times\calE^n$ when the hash function $f$ is applied to the variable $A^n$ and denote
 \begin{equation}
 U^{( Q_{AE } )}_{ E^n }( \be)
:=\sum_{i=1}^{\|f\|} U^{( Q_{AE } )}_{f(A^n),E^n }(i,\be)
   \end{equation} 
as its $\calE^n$-marginal. Because the probability of pairs of sequences of the same joint type have the same $P_{AE}^n$-probability, we can write
\begin{equation}
P_{f(A^n),E^n}(i,\be) = \sum_{Q_{AE}\in\calP_n(\calA\times\calE)} P_{AE}^n(\calT_{Q_{AE}})  U^{( Q_{AE } )}_{f(A^n),E^n }(i,\be) . \label{eqn:mixture}
\end{equation}
   By using \eqref{eqn:mixture}, we have
\begin{align}
C_1(f(A^n) | E^n | P_{AE}^n) 
&=D ( P_{f(A^n),E^n} \| P_{\mix,f(\calA^n)}\times P_{E^n}) \\
&\le\sum_{Q_{AE} \in\calP_n(\calA\times\calE)} P_{AE}^n( \calT_{ Q_{AE} } ) D \left( U^{( Q_{AE } )}_{f(A^n),E^n } \Big\| P_{\mix,f(\calA^n)} \times U^{( Q_{AE } )}_{ E^n }  \right) \label{eqn:cvx_div1} \\
&=\sum_{Q_{AE} \in\calP_n(\calA\times\calE)} 
P_{AE}^n( \calT_{ Q_{AE} } ) 
C_1 \left(f(A^n) | E^n \, \Big|\, U^{( Q_{AE } )}_{A^n E^n }  \right) \label{eqn:c1_def} \\
&\le \sum_{Q_{AE} \in\calP_n(\calA\times\calE)} 
P_{AE}^n( \calT_{ Q_{AE} } ) 
C_2 \left(f(A^n) | E^n \, \Big|\, U^{( Q_{AE } )}_{A^n E^n } \right)  .
\label{eqn:cvx_div2}
\end{align}
where \eqref{eqn:cvx_div1} follows from the fact that relative entropy is convex, \eqref{eqn:c1_def} follows from the definition of $C_1$, and \eqref{eqn:cvx_div2} follows from the fact that $s\mapsto C_{1+s}$ is monotonically non-decreasing.


In Step 2, we regard $f$ as a  universal$_2$ hash function $f_{X_n}$.
Thus, \eqref{eqn:cvx_div2} implies that
\begin{align}
C_1(f_{X_n}(A^n) | E^n X_n | P_{AE}^n\times P_{X_n}) 
&=
\bbE_{X_n} \left[ C_1(f_{X_n}(A^n) | E^n     | P_{AE}^n)  \right]\\
&\le \sum_{Q_{AE} \in\calP_n(\calA\times\calE)} 
P_{AE}^n( \calT_{ Q_{AE} } ) 
\bbE_{X_n} \left[ C_2 \left(f_{X_n}(A^n) | E^n \,\Big|\, U^{( Q_{AE } )}_{A^n E^n } \right) \right].
\label{eqn:cvx_div}
\end{align}
Let  $\calT_{Q_{A|E}}(\be):=\{\ba: (\ba,\be) \in\calT_{Q_{AE}} \}$ be the conditional type class of $Q_{A|E}$  given $\be$, also known as the $Q_{A|E}$-shell. By the method of types~\cite[Ch.~2]{Csi97}, we know that  for $\be$ of type $Q_E$, 
\begin{equation}
\log \left| \calT_{Q_{A|E}}(\be) \right| =  n H(A|E | Q_{AE}) + O(\log n).
\end{equation}
By using the fact that $\rme^{-H_2(A|E)}$ is the  conditional collision probability (i.e., $\rme^{-H_2(A|E)}=\sum_e P_E(e) P_{AA'|E=e}\{(a,a') :a=a'   \}$ where $A,A'$ are conditionally independent and identically distributed given $E$), 
\begin{align}
\rme^{-H_2(A|E|U^{( Q_{AE } )})} 
&=
\sum_{\be} U^{( Q_{AE } )}_{ E^n }( \be)
\sum_{\ba \in \calT_{Q_{A|E}}(\be)} \frac{1}{|\calT_{Q_{A|E}}(\be)|^2}\\
&=
\sum_{\be} U^{( Q_{AE } )}_{ E^n }( \be)
\frac{1}{|\calT_{Q_{A|E}}(\be)|}\\
&=
\rme^{-n H(A|E | Q_{AE}) + O(\log n)}.\label{H5-16}
\end{align}
Furthermore, by a Taylor expansion of $H(A|E | Q_{AE})$ around $P_{AE}$ as in the rate redundancy lemma~\cite{TK12c, ingber11}, we have 
\begin{equation}
H(A|E | Q_{AE})  = H(A|E|P_{AE}) + \sum_{a,e} (  Q_{AE}(a,e)-P_{AE}(a,e)  ) h_{A|E}(a|e) + O\big(  \|Q_{AE}-P_{AE}\|^2 \big)   \label{eqn:taylor} 
\end{equation}
where the {\em conditional entropy density} $h_{A|E}(a|e) $ is defined as
\begin{equation}
h_{A|E}(a|e) := \log \frac{1}{P_{A|E}(a|e)} 
\end{equation}
and $\|Q-P\|=\sum_{z \in\calZ}|Q(z)-P(z)|$ is the variational distance between $Q$ and $P$. 
For brevity, we denote the $\sqrt{n}$-scaled version of the second term in \eqref{eqn:taylor} as 
\begin{equation}
b_n(Q_{AE}):= \sqrt{n} \left( \sum_{a,e} (  Q_{AE}(a,e)-P_{AE}(a,e)  ) h_{A|E}(a,e) \right).
\end{equation}
If $Q_{AE}$ is a random type formed from $n$ independent copies of $P_{AE}$,
\begin{equation} \label{eqn:clt_types}
  b_n(Q_{AE}) = \sqrt{n} \left( \frac{1}{n}\sum_{i=1}^n h_{A|E}(A_i|E_i) - H(A|E|P_{AE}) \right) \stackrel{\rmd}{\longrightarrow} \calN\left( 0, V(A|E|P_{AE})\right)
\end{equation}
by the central limit theorem. That is, $b_n(Q_{AE})$ converges in distribution to the Gaussian $ \calN\left( 0, V(A|E|P_{AE})\right)$. 

In Step 3,  we first fix $\delta>0$. 
Applying the universal$_2$ property of 
the universal$_2$ hash function $f_{X_n}$ to 
the collision relative entropy (see~\eqref{eqn:equiv_direct_prf4}--\eqref{eqn:equiv_direct_prf3} with $\epsilon=s=1$), and combining the above notations and bounds, we obtain for all $\be\in\calT_{Q_E}$ and all $n$ large enough (depending on $\delta$) that
\begin{align}
& 
\bbE_{X_n} \left[ C_2 \left(f_{X_n}(A^n) | E^n \,\Big|\, U^{( Q_{AE } )}_{A^n E^n } \right) \right]
\nn\\*
&\qquad =\bbE_{X_n}\left[ \log M - H_2\left( f_{X_n}(A^n) | E^n \,\Big|\, U^{( Q_{AE } )}_{A^n E^n } \right)  \right]\label{eqn:C2_def}\\
&\qquad\le
\log \left( 1+  M_n \rme^{-n H(A|E | Q_{AE}) + O(\log n)}
\right)  \label{eqn:C2_prop}\\
 &\qquad \le 
\log \left(1+\exp\left[ \sqrt{n} ( L - b_n(Q_{AE})  + o(b_n(Q_{AE}) ) + O(\log n) \right] \right) \label{eqn:use_size_f}\\
&\qquad\le
\left\{  \begin{array}{cc}
 \sqrt{n} \big(L-  b_n(Q_{AE})+  o(b_n(Q_{AE}) ) \big)  + O(\log n) &  b_n(Q_{AE})  \le   L+\delta  \\
  \rme^{-\delta \sqrt{n}   /2} &   b_n(Q_{AE})   >    L+\delta 
\end{array}  \right. ,\label{eqn:D2_cases}
\end{align}
where \eqref{eqn:C2_def} follows from the definition of $C_2$ and \eqref{eqn:C2_prop} uses the bound in \eqref{H5-16}.  Also
note that we used the fact that $\|f\|=M_n=\rme^{n H(A|E|P_{AE} )+ \sqrt{n} L }$ in \eqref{eqn:use_size_f}.    

Finally in Step 4,   by plugging \eqref{eqn:D2_cases} back into \eqref{eqn:cvx_div}, we obtain that for all $n$ large enough (depending on $\delta$),
\begin{align}
&C_1(f_{X_n}(A^n) | E^n X_n | P_{AE}^n\times P_{X_n}) \nn\\*
 &\quad\le \sum_{Q_{AE}\in\calP_n(\calA\times\calE):    b_n(Q_{AE}) \le L +\delta } P_{AE}^n (\calT_{Q_{AE}})\big( L-   (1-\delta)b_n(Q_{AE}) \big) +O\left(\frac{\log n}{\sqrt{n}}\right).
\end{align}
Let $V:=V(A|E|P_{AE})$.  By the central limit-type convergence in \eqref{eqn:clt_types}, we obtain
\begin{equation}
\limsup_{n\to\infty}\frac{1}{\sqrt{n}}C_1(f_{X_n}(A^n) | E^n X_n | P_{AE}^n\times P_{X_n}) \le \int_{-\infty}^{L+\delta} \frac{L-(1-\delta)b}{\sqrt{2\pi V}}\, \rme^{-b^2/(2V) } \, \rmd b.
\end{equation}
By a change of variables to $x:=b/\sqrt{V}$ and taking $\delta\downarrow 0$, we immediately obtain the direct part (upper bound) of Case (B) in~\eqref{eqn:caseB}.

\paragraph{Proof of upper bounds for Case (C)}
For Case (C), the upper bound to \eqref{eqn:caseC1} can be obtained by specializing the one-shot bound in \eqref{eqn:one-shot-dir-second-1} with $\epsilon=1$, $M_n=\rme^{ n H  (A|E|P_{AE}) + \sqrt{n} L}$ and $c=\rme^{-n^{1/4}}$. With these choices, we have 
\begin{align}
& \rme^{-s  C_{1-s} (f_{X_n}(A^n)|E^nX_n|P_{AE}^n\times P_{X_n}) } \nn\\*
 & \ge P_{AE}^n\bigg\{ (\ba,\be):P_{A|E}^n(\ba |\be)\le \frac{\rme^{-n^{1/4}}}{\rme^{ n H  (A|E|P_{AE}) + \sqrt{n} L}}\bigg\}\Big(\frac{1}{1+\rme^{-n^{1/4}} }\Big)^s \label{eqn:clt} \\
&=P_{AE}^n\bigg\{ (\ba,\be):  \frac{1}{\sqrt{n}}\sum_{i=1}^n\big[-\log P_{A|E} (a_i |e_i) - H(A|E|P_{AE})  \big]\ge  L + \frac{1}{\sqrt{n}}\bigg\}  \Big(\frac{1}{1+\rme^{-n^{1/4}} }\Big)^s \label{eqn:clt2} . 
\end{align}
The probability is an information spectrum~\cite{Han10} term with $n$ independent and identically distributed random variables and  since $P_{A_i E_i}=P_{AE}$   for each $1\le i \le n$,
\begin{align}
\bbE_{P_{A_i E_i}}\big[ -\log P_{A|E}(A_i|E_i)\big]&=H  (A|E|P_{AE}), \label{eqn:stats0} \\
\var_{P_{A_i E_i}}\big[ -\log P_{A|E}(A_i|E_i)\big]&=V  (A|E|P_{AE})  \label{eqn:stats} . 
\end{align}
So by the central limit theorem, the right-hand-side  of \eqref{eqn:clt2} converges uniformly as follows:
\begin{equation}
\lim_{n\to\infty}P_{AE}^n\bigg\{ (\ba,\be):  \frac{1}{\sqrt{n}}\sum_{i=1}^n\big[-\log P_{A|E} (a_i |e_i) - H(A|E|P_{AE})  \big]\ge  L + \frac{1}{\sqrt{n}}\bigg\} =  \Phi\Big(- \frac{L}{ \sqrt{V(A|E|P_{AE})} }\Big). \label{eqn:clt_res}
\end{equation}
 Plugging \eqref{eqn:clt_res} into \eqref{eqn:clt2}, taking the logarithm, and normalizing by $-s$ yields  the upper bound to \eqref{eqn:caseC1}.  

In a similar way, the upper bound to \eqref{eqn:caseC2} can be obtained by specializing the one-shot bound in \eqref{eqn:one-shot-dir-second-2} with $\epsilon=1$, $M_n=\rme^{ n H  (A|E|P_{AE}) + \sqrt{n} L}$ and $c=\rme^{-n^{1/4}}$.  The calculation for the specialization is similar to the converse part which is detailed in full in \eqref{eqn:obtain_D}--\eqref{eqn:convolu} in the next section. This completes the proof for the direct part of Case (C) of Theorem~\ref{thm:second}.  
\subsubsection{Converse Parts}

\paragraph{Proof of lower bounds for Case (A)}
We now prove the lower bounds for Case (A).  The first two  bounds can be shown using the data processing inequalities in \eqref{eqn:dpi1}--\eqref{eqn:dpi3}. In particular, the  lower bound    to~\eqref{eqn:2order1}  can be evaluated as follows:
\begin{align}
C_{1+s}(f(A^n) | E^n|P_{AE}^n) &= n H_{1+s}(A|E|P_{AE})+\sqrt{n}L- H_{1+s}(f(A^n) | E^n|P_{AE}^n) \\
&\ge n H_{1+s}(A|E|P_{AE})+\sqrt{n}L- H_{1+s}( A^n  | E^n|P_{AE}^n) \label{eqn:use_dpi2} \\
&= n H_{1+s}(A|E|P_{AE})+\sqrt{n}L- n H_{1+s}(A|E|P_{AE}) \\
&=\sqrt{n}L,
\end{align}
where \eqref{eqn:use_dpi2} follows from \eqref{eqn:dpi2}. 
The lower bound to \eqref{eqn:2order1_g} follows completely analogously using \eqref{eqn:dpi3}.

The lower bound to \eqref{eqn:2order1_Lneg} can be shown by  first
 relaxing~\eqref{10-20-2b}  as follows:
\begin{align}
&\rme^{s C_{1+s}(f(A) | E  | P_{AE}  )} \nn\\*
& \ge 
\sum_{ (a,e):P_{A|E} (a|e) \ge \frac{c}{M}}
P_E(e)P_{A|E}(a|e)^{1+s} M^{s}
+\sum_e P_E(e) P_{A|E=e}
\Big\{ (a,e):P_{A|E} (a|e) < \frac{c}{M}\Big\}^{1+s}\\
& \ge 
  \sum_e P_E(e) P_{A|E=e}
\Big\{ (a,e):P_{A|E} (a|e) < \frac{c}{M}\Big\}^{1+s}\\
& \ge 
P_{A E} \Big\{ (a,e):P_{A|E} (a|e) < \frac{c}{M}\Big\}^{1+s} \label{eqn:use_jens}\\
&=\left[1-P_{A E} \Big\{ (a,e):P_{A|E} (a|e) \ge \frac{c}{M}\Big\} \right]^{1+s} \label{eqn:one_minus}\\
&\ge 1-(1 +s)P_{A E} \Big\{ (a,e):P_{A|E} (a|e) \ge \frac{c}{M}\Big\} \label{eqn:use_ineq} 
\end{align} 
where \eqref{eqn:use_jens} uses Jensen's inequality (for the convex function $t\mapsto t^{1+s}$) and \eqref{eqn:use_ineq} uses the inequality $(1-x)^{1+s }\ge 1-(1+s)x$ (also due to the convexity of $t\mapsto t^{1+s}$).  Hence we have for the $n$-shot setting
\begin{equation}
sC_{1+s}(f(A^n) | E^n  | P_{AE}^n  )\ge\log\left( 1-(1 +s)P_{A E}^n \Big\{ (\ba,\be):P_{A|E}^n (\ba|\be) \ge \frac{c}{M_n}\Big\}  \right) \label{eqn:take_logs} .
\end{equation}
Applying the modified G\"artner-Ellis theorem derived in Hayashi-Tan~\cite[Appendix~A]{HayashiTan2015a} to the sequence of random variables $-\log P_{A|E}^n( A^n|E^n)$ with $M_n = \rme^{n H_{1+s}(A|E|P_{AE}) + \sqrt{n} L}$ and $c=1$, we have 
\begin{equation}
\lim_{n\to\infty} \frac{1}{\sqrt{n}}\log P_{A E}^n \Big\{ (\ba,\be):P_{A|E}^n (\ba|\be) \ge \frac{c}{M_n}\Big\}  =-sL \label{eqn:ge} .
\end{equation}
The modification here is due to the different normalization of $\sqrt{n}$ as opposed to the normalization by $n$  in the  usual G\"artner-Ellis theorem in~\cite{Dembo}. Also see Remark (a) to Theorem 2.3.6 in \cite{Dembo}. Combining \eqref{eqn:ge} with \eqref{eqn:take_logs} yields the lower bound to \eqref{eqn:2order1_Lneg}.  The lower bound to \eqref{eqn:2order1_g_Lneg} can be proved in a completely analogous way by relaxing the one-shot bound in~\eqref{10-20-2_205}.

\paragraph{Proof of lower bound  for Case (B)}
For the converse part of Case (B), we use Theorem 8 of~\cite{Hayashi08}, which analyzes the second-order asymptotics of intrinsic randomness~\cite[Ch.~2]{Han10} \cite{vembu}.  Define the second-order  coding rate at length $n$ as 
\begin{equation}
L_n := \frac{1}{\sqrt{n}} \big(\log M_n- n H (A|E|P_{AE})\big) \label{eqn:defLn}
\end{equation}
and the distribution function $F_n^{(\be)}$ which is  dependent on $\be$ as
\begin{equation}
F_n^{(\be)}( x):= P_{A^n|E^n=\be}\left\{ \ba:-\frac{1}{n}\log P_{A^n|E^n=\be}(\ba) \le H(A|E|P_{AE}) + \frac{x}{\sqrt{n}}\right\}.
\end{equation}
Now, from the proof of Theorem 8 of~\cite{Hayashi08} (second column page 4634), we deduce that for each $\be\in\calE^n$, 
\begin{align}
&H(f(A^n) | P_{A^n|E^n=\be})\nn\\*
& \le \sqrt{n}\int_{-\infty}^{L_n} a \, \rmd F_n^{(\be)}( a) + nH (A|E|P_{AE}) \nn\\*
&\quad + P_{A^n|E^n=\be}\left\{\ba: P_{A^n|E^n=\be}(\ba) \le\frac{1}{M_n}  \right\} \left( \sqrt{n}L_n  - \log  P_{A^n|E^n=\be}\left\{\ba: P_{A^n|E^n=\be}(\ba) \le\frac{1}{M_n}  \right\} \right)  . \label{eqn:th8_hayashi}
\end{align} 
Now note that 
 $F_n^{(\be)}( x)$ depends only on $\be$ through its  type. Our next step is to take the expectation of \eqref{eqn:th8_hayashi} over $\be$ with distribution $P_{E}^n$. Let 
\begin{equation}
 g(\be) := P_{A^n|E^n=\be}\left\{\ba: P_{A^n|E^n=\be}(\ba) \le\frac{1}{M_n}  \right\}  .
 \end{equation} Since $t\mapsto -t\log t$ is concave, by Jensen's inequality, we have
\begin{equation}
 \bbE_{P_{E^n}} [ g(E^n) (\gamma-\log g(E^n) ) ] \le \bbE_{P_{E^n}} [ \gamma g(E^n) ] - \bbE_{P_{E^n}}[ g(E^n) ]\log\bbE_{P_{E^n}}[ g(E^n) ].\label{eqn:jens}
 \end{equation} 
 Now define the averaged distribution function as
\begin{equation}
F_n ( x):=\sum_{\be}P_E^n(\be) F_n^{(\be)}( x)= P_{AE}^n\left\{ (\ba,\be):-\frac{1}{n}\log P_{A|E}^n(\ba|\be) \le H(A|E|P_{AE}) + \frac{x}{\sqrt{n}}\right\}. \label{eqn:dist_fu}
\end{equation}
Let $\gamma:=\sqrt{n}L_n$.  From \eqref{eqn:jens} and the definition of $F_n ( x)$,   
\begin{align}
&H(f(A^n) | E^n | P_{AE}^n)\nn\\*
 &\le  \sqrt{n}  \int_{-\infty}^{L_n} a \, \rmd F_n(a ) + n H(A|E|P_{AE})   \nn\\*
&\qquad + P_{AE}^n \left\{ (\ba,\be) : P_{A|E}^n (\ba|\be) \le\frac{1}{M_n}\right\} \left( \sqrt{n}L_n - \log P_{AE}^n \left\{(\ba,\be) : P_{A|E}^n (\ba|\be)  \le\frac{1}{M_n}\right\} \right).
\end{align}
Thus, by invoking the definition of $L_n$ in \eqref{eqn:defLn} and $F_n$ in \eqref{eqn:dist_fu}, we obtain the inequality
\begin{align}
&\frac{1}{\sqrt{n}} \left( H(f(A^n) | E^n | P_{AE}^n)- n H(A|E|P_{AE}) \right)\nn\\*
&\le\int_{-\infty}^{L_n} a\, \rmd F_n(a) + P_{AE}^n \left\{ (\ba,\be) : P_{A|E}^n (\ba|\be) \le\frac{1}{M_n}\right\}\left(  L_n - \frac{\log P_{AE}^n \left\{ (\ba,\be) : P_{A|E}^n (\ba|\be) \le\frac{1}{M_n}\right\}}{\sqrt{n}}\right) \\
&=\int_{-\infty}^{L_n} a\, \rmd F_n(a)  + (1-F_n(L_n)  ) \left(L_n - \frac{\log (1-F_n(L_n)) }{\sqrt{n}}\right) \label{eqn:limit_F} .
\end{align}
By the central limit theorem
\begin{equation}
F_n(x)\to F(x) = \int_{-\infty}^x \frac{1}{\sqrt{2\pi V}} \rme^{-y^2/(2V)}\,\rmd y,\quad\forall\, x\in\bbR.  \label{eqn:clt_F}
\end{equation}
Taking the $\limsup$ of \eqref{eqn:limit_F}, and using the central limit result in~\eqref{eqn:clt_F}, we obtain
\begin{equation}
\limsup_{n\to\infty}\frac{1}{\sqrt{n}} \left( H(f(A^n) | E^n | P_{AE}^n)- n H(A|E|P_{AE}) \right)\le\int_{-\infty}^L a\, \rmd F(a)+ L(1-F(L)).
\end{equation}
 Since, we have the simple relation
\begin{align}
 C_1(f(A^n) | E^n | P_{AE}^n) &=D( P_{f(A^n),E^n} \| P_{\mix,f(\calA^n)}\times P_{E^n}) \\
 &=-H(f(A^n) | E^n | P_{AE}^n) + \log M_n \\
  &=-H(f(A^n) | E^n | P_{AE}^n) + nH (A|E|P_{AE}) + \sqrt{n}L ,
\end{align}
we immediately obtain the desired lower bound for the second-order asymptotics of $C_1$:
\begin{equation}
\liminf_{n\to\infty}\frac{1}{\sqrt{n}} C_1(f(A^n) | E^n | P_{AE}^n)\ge \int_{-\infty}^L (L-a) \, \rmd F(a) = \int_{-\infty}^{L/\sqrt{V}} \frac{L-\sqrt{V}x}{\sqrt{2\pi}}\rme^{-x^2/2}\, \rmd x .
\end{equation}

\paragraph{Proof of lower bounds for Case (C)}
For Case (C), the first part of the maximum  in the lower bound in \eqref{eqn:caseC1}, namely $\Gamma_1(s,L)$ in \eqref{eqn:defA}, follows from~\eqref{eqn:os_conv_sec} and the second part of the maximum, namely $\Gamma_2(s,L)$ in \eqref{eqn:defB}, follows from~\eqref{10-20-1} with the common choice of  $c=\rme^{n^{1/4}}$. In particular, specializing the bound in one-shot bound in~\eqref{10-20-1} with this choice of $c$, we obtain   
\begin{equation}
\rme^{-sC_{1-s}(f(A^n)|E^n|P_{AE}^n)}\le (\rme^{n^{1/4}})^{-s} + P_{AE}^n\bigg\{ (\ba,\be): P_{A|E}^n(\ba|\be)\le\frac{\rme^{n^{1/4}} }{\rme^{ n H(A|E|P_{AE}) + \sqrt{n}L}}\bigg\}^{1-s} \label{eqn:clt3}
\end{equation}
where we trivially upper bounded the first probability in the one-shot bound by $1$.  The first term in  \eqref{eqn:clt3} goes to zero (since $s>0$) while the second term is an information spectrum term that asymptotically behaves as 
\begin{equation}
\lim_{n\to\infty} P_{AE}^n\bigg\{ (\ba,\be): P_{A|E}^n(\ba|\be)\le\frac{\rme^{n^{1/4}} }{\rme^{ n H(A|E|P_{AE}) + \sqrt{n}L}}\bigg\}= \Phi\bigg(-\frac{L}{ \sqrt{V(A|E|P_{AE}) }}  \bigg) 
\end{equation}
 by the central limit theorem and the statistics computed in \eqref{eqn:stats0}--\eqref{eqn:stats}. Hence, taking the logarithm in~\eqref{eqn:clt3}, and normalizing by $-s$, we obtain the second term in the maximum in the lower bound in~\eqref{eqn:caseC1}, namely $\Gamma_2(s,L)$. In exactly the same way, specializing the bound in \eqref{eqn:os_conv_sec}, we obtain 
\begin{equation}
\rme^{-sC_{1-s}(f(A^n)|E^n|P_{AE}^n)}\le 
(\rme^{n^{1/4}})^{-s} 
+ 2^{\frac{s}{1-s}}s^{\frac{s}{1-s} }P_{AE}^n\bigg\{ (\ba,\be):  P_{A|E}^n(\ba|\be)\le\frac{\rme^{n^{1/4}} }{\rme^{ n H(A|E|P_{AE}) + \sqrt{n}L}}\bigg\}\label{eqn:clt4} .
\end{equation}
Applying the central limit theorem to the probability in the second term recovers  $\Gamma_1(s,L)$ in  the lower bound in~\eqref{eqn:caseC1}.  

The method to obtain the two terms in the maximum in the lower bound in \eqref{eqn:caseC2} is more complicated than that for~\eqref{eqn:caseC1} because we need to condition on various sequences $\be\in\calE^n$. In particular, to obtain the lower bound $\Psi_1(s,L)$ in \eqref{eqn:defD}, we evaluate~\eqref{eqn:os_conv_eqG}  with $c=\rme^{n^{1/4}}$. We obtain
\begin{align}
& \rme^{-\frac{s}{1-s}C_{1-s}^\uparrow(f(A^n)|E^n|P_{AE}^n)  } \nn\\*
&\le 2^{\frac{s}{1-s}}\bigg[  (\rme^{n^{1/4}})^{-\frac{s}{1-s} } + \sum_{\be}P_E^n(\be)    \Big(2^{\frac{s}{1-s}}s^{\frac{s}{1-s}}(1-s)P_{A^n|E^n=\be}\Big\{ \ba: P_{A|E}^n(\ba|\be) < \frac{c}{M_n} \Big\}\Big)^{\frac{1}{1-s}}  \bigg].\label{eqn:obtain_D}
\end{align}
As usual, the first term goes to zero.  To compute the probability in the second term, let us denote the type (empirical distribution)~\cite{Csi97} of $\be$ by $Q_{\be} \in\calP_n(\calE)$ for the moment. Then we have 
\begin{equation}
\left| P_{A^n|E^n=\be}\Big\{ \ba: P_{A|E}^n(\ba|\be) < \frac{c}{M_n} \Big\}  -  \Phi\bigg(- \frac{L+H(A|E|P_{AE}) - H(A|E|P_{AE} \| Q_{\be})}{\sqrt{ V_2(A|E|P_{AE}\|Q_{\be})} }\bigg)  \right|\le O\bigg(\frac{1}{\sqrt{n}}\bigg) \label{eqn:use_be}
\end{equation}
by the   Berry-Esseen theorem~\cite[Sec.\ XVI.7]{feller}, where the conditional entropy given  another distribution $Q_{\be}$, denoted as $H(A|E|P_{AE} \| Q_{\be})$, was defined in \eqref{eqn:cond_entr_given}, and conditional varentropy given  another distribution $Q_{\be}$ is defined as
\begin{equation}
V_2(A|E|P_{AE}\|Q_{\be}):= \sum_e Q_{\be}(e) \sum_a P_{A|E}(a|e)\Big[ \log \frac{1}{P_{A|E}(a|e)} -  H(A|P_{A|E=e})\Big]^2.
\end{equation}
Note that $V_2(A|E|P_{AE}\|P_E)=V_2(A|E|P_{AE})$ defined in \eqref{eqn:V2_def}.  
In \eqref{eqn:use_be}, the remainder term $O(\frac{1}{\sqrt{n}})$  is uniform in $L$ and $Q_{\be}$. We now plug this into \eqref{eqn:obtain_D} and notice that we are  then averaging over all types $Q_{\be}$ (where $E^n$ has  distribution $P_E^n$). Now, employing a weak (expectation) form of the Berry-Esseen theorem  \cite[Thm.~2.2.14]{tao12} with $x=H(A|E|P_{AE}) -H(A|E|P_{AE} \| Q_{\be})$  yields 
\begin{align}
& \rme^{-\frac{s}{1-s}C_{1-s}^\uparrow(f(A^n)|E^n|P_{AE}^n)  } \nn\\*
&\qquad\le 2^{\frac{s}{1-s}} \bigg[ O\Big( n^{ -\frac{s}{4(1-s)}}\Big) +\Big(2^{\frac{s}{1-s}}s^{\frac{s}{1-s}}(1-s)\Big)^{\frac{1}{1-s}} \nn\\*
&\qquad\qquad\times \int_{-\infty}^{\infty}   \Big[\Phi\bigg(- \frac{L+x}{\sqrt{ V_2(A|E|P_{AE})} }\bigg) +  O \Big( \frac{1}{\sqrt{n}}\Big) \Big]^{\frac{1}{1-s}} \frac{ \rme^{-x^2/(2 V_1(A|E|P_{AE}))}}{\sqrt{2\pi V_1(A|E|P_{AE})}}\,\rmd x\bigg]  + O \Big( \frac{1}{\sqrt{n}}\Big). \label{eqn:convolu}
\end{align}
Now we take the logarithm, divide both sides by $-\frac{s}{1-s}$,  and take the limit as $n\to\infty$. This yields the lower bound $\Psi_1(s,L)$ in \eqref{eqn:defD}. Note that here unlike in the steps leading to~\eqref{eqn:conv}, we cannot add $V_1$ and $V_2$ due the exponentiation of the first term by  $\frac{1}{1-s} $  in the integral.

Using   similar techniques, we can  obtain the lower bound $\Psi_2(s,L)$ defined in \eqref{eqn:defE} from \eqref{eqn:os_conv_expG}. In particular,   evaluate~\eqref{eqn:os_conv_expG} with the same choice of $c$.  Here, in fact, no averaging over $E^n$ is needed because the first term in~\eqref{eqn:os_conv_expG} vanishes by our choice of $c=\rme^{n^{1/4}}$. Thus, we obtain the lower bound in \eqref{eqn:caseC2}.

This completes the proof of the converse parts of Theorem~\ref{thm:second}.


\section{Conclusion} \label{sec:concl}

\subsection{Summary}
We have derived the fundamental limits of the asymptotic behavior of the equivocation when a hash function $f$ is applied to the source (Theorem~\ref{thm:equiv}). We have also showed that   optimal key generation rates change when we use   alternative R\'enyi information measures (Corollary~\ref{cor:key}). Under these R\'enyi quantities, we have   evaluated  the  corresponding exponential rates of decay of the security measures (Theorem~\ref{thm:exponents}) as well as their second-order coding rates (Theorems~\ref{thm:second} and~\ref{thm:second_large}). The  R\'enyi information measures generalize  the ubiquitous  Shannon information measures and may be useful in many settings as described in the Introduction. To establish our asymptotic theorems, we have introduced new families  of non-asymptotic achievability and converse bounds on the R\'enyi information measures and their Gallager counterparts and used various probabilistic limit theorems (such as large deviation theorems and the central limit theorem) to  evaluate  these bounds when the number of realizations of the joint source tends to infinity. 
\subsection{Future Research Directions}
In the future, we plan to explore various extensions to the results contained herein. 
\begin{enumerate}
\item We would like to study security problems such as the {\em remaining or residual uncertainty} of a source $A^n$ when another party observes a compressed version $f(A^n)\in\calM:= \{1,\ldots, M_n\}$ and another correlated source $E^n$. Namely, we aim to study  the asymptotic behavior of the  conditional R\'enyi entropy $H_{1+s}(A^n| f(A^n), E^n | P_{AE}^n)$ and  its Gallager counterpart $H_{1+s}^\uparrow(A^n| f(A^n), E^n | P_{AE}^n)$. 
\item Another  set of related problems involve the  analyses of the asymptotic behavior of  $H_{1+s}(f(A^n)|   E^n | P_{AE}^n)$ and  $H_{1+s}^\uparrow(f(A^n)|   E^n | P_{AE}^n)$. These represent  the uncertainties of an eavesdropper with regard to the message index $f(A^n)\in\calM$. The eavesdropper, however, is equipped with correlated observations $E^n$.  We anticipate that some of the techniques developed in the current paper may be useful to perform various calculations. 
\item We focused primarily on analyzing $C_{1+s}$ and $C_{1+s}^\uparrow$ for    $s\in [-1,1]$. It may be of interest to study the various asymptotic behaviors of $C_{1+s}$ and $C_{1+s}^\uparrow$  for general $s\in\bbR$ since for example, $H_{\min}=\lim_{s\to\infty}H_{1+s}$ and $H_{\min}$ \cite{Impagliazzo, BBCM, hastad,dodis13} is a fundamental quantity in cryptography and information-theoretic security as mentioned in Section~\ref{sec:motivation}.  Indeed, $\rme^{-H_{\min}(A|E|P_{AE})} $ is the best (highest) probability of successfully guessing $A$ given $E$. As remarked after Theorems~\ref{thm:equiv} and~\ref{thm:second}, we already have the converse parts for all $s\ge 0$ for the results in~\eqref{eqn:C1s}, \eqref{eqn:C1sG}, \eqref{eqn:2order1} and~\eqref{eqn:2order1_g}. They follow immediately from various information processing inequalities.  It would be ideal, though challenging, to complete the story.   
\item Lastly, we aim to apply the results and techniques derived herein to information-theoretic security problems such as the wiretap channel~\cite{Wyn75} and secret key agreement~\cite{AC93} as was done  by various researchers in~\cite{Hayashi11,Hayashi13,chou12,chou15}.
\end{enumerate}

\appendices
\newcommand{\sM}{M}
\newcommand{\rE}{\bbE}
\section{Proof of Lemma~\ref{lem:os_d_1}}

\subsection{Proof of \eqref{eqn:os_direct_pluss}}

\begin{proof}
The derivation here is similar to that in \cite{Hayashi06, Hayashi11} for  universal$_2$ hash functions. Throughout, for any function $f:\calA \to \calM$, we let
\begin{equation}
f^{-1} (i):= \{a\in\calA: f(a) = i\},\qquad\forall\, i \in\calM. 
\end{equation}
   Now, for any $a$, due to the  $\epsilon$-almost universal$_2$  property of $f_X$,  we have
\begin{align}
\rE_X  \sum_{a' \in f_X^{-1}(f_X(a))} P_{A|E}(a'|e) 
&\le    
P_{A|E}(a|e)  + \frac{\epsilon}{\sM}\sum_{a' \neq a } P_{A|E}(a'|e) \\
&\le 
P_{A|E}(a|e)  + \frac{\epsilon}{\sM} .\label{10-14-7b}
\end{align}
Starting from the definition of the conditional R\'enyi divergence, we have
\begin{align}
&  \rme^{-s H_{1+s}(f_X(A)|EX|P_{A E}\times P_X)}\nn\\*
&= 
\rE_X \sum_e P_E(e)  
\sum_{i=1}^M \bigg(\sum_{a \in f_X^{-1}(i)} P_{A|E}(a|e)\bigg)^{1+s} \label{eqn:seq1}
\\
&= 
\rE_X \sum_e P_E(e)  
\sum_{a } P_{A|E}(a|e)
\bigg(\sum_{a' \in f_X^{-1}(f_X(a))} P_{A|E}(a'|e)\bigg)^{s}
\\
&\le 
\sum_e P_E(e)  
\sum_{a } P_{A|E}(a|e)
\bigg(\rE_X \sum_{a' \in f_X^{-1}(f_X(a))} P_{A|E}(a'|e)\bigg)^{s}  \label{eqn:convexity_s}\\
&\le 
\sum_e P_E(e)  
\sum_{a } P_{A|E}(a|e)
\bigg(P_{A|E}(a|e)  + \frac{\epsilon}{\sM} \bigg)^{s}
\label{eqn:use_hash} \\
&\le 
\sum_e P_E(e)  
\sum_{a } P_{A|E}(a|e)
\bigg( P_{A|E}(a|e)^s  + \Big(\frac{\epsilon}{\sM} \Big)^{s}\bigg) \label{eqn:s_inequ}
\\
&= 
\Big(\frac{\epsilon}{\sM} \Big)^{s}
+
\sum_e P_E(e)  
\sum_{a } P_{A|E}(a|e)^{1+s}\\
&= \frac{\epsilon^s}{\sM^s}+ \rme^{-s H_{1+s}(A|E|P_{AE})}, \label{eqn:direct_prf_non_asymp1}
\end{align}
where \eqref{eqn:convexity_s}, we used the concavity of $t\mapsto t^s$ for $s\in [0,1]$, in \eqref{eqn:use_hash} we used the fact that $f_X$ is a $\epsilon$-almost universal$_2$ hash function, and in \eqref{eqn:s_inequ} we used the inequality $(\sum_i a_i)^s\le \sum_i a_i^s$ for $s\in [0,1]$ \cite[Problem 4.15(f)]{gallagerIT}. 

We remark that the sequence of steps in \eqref{eqn:seq1} to \eqref{eqn:direct_prf_non_asymp1} is inspired by the work of Hayashi~\cite{Hayashi11} who derived a similar result but for Shannon-type quantities instead of R\'enyi-type quantities as we do here. 

By \eqref{eqn:relate_C_H}, we have 
\begin{equation}
C_{1+s}(f_X(A)|EX|P_{A E}\times P_X)
=\log M - H_{1+s}(f_X(A)|EX|P_{A E}\times P_X). \label{eqn:relate_C_H2}
\end{equation}
Uniting \eqref{eqn:direct_prf_non_asymp1} and \eqref{eqn:relate_C_H2} proves \eqref{eqn:os_direct_pluss} as desired.\end{proof}

\subsection{Proof of \eqref{eqn:os_direct_G_pluss}  }
\begin{proof}
Along exactly the same lines, we also have
\begin{align}
& \rme^{-\frac{s}{1+s} H_{1+s}^{\uparrow}(f_X(A)|EX|P_{AE}\times P_X)}\nn\\*
&= 
\rE_X \sum_e P_E(e)  
\bigg(\sum_{i=1}^M \Big(\sum_{a \in f_X^{-1}(i)} P_{A|E}(a|e)\Big)^{1+s}\bigg)^{\frac{1}{1+s}}
\\
&=
\rE_X \sum_e P_E(e)  
\bigg(\sum_{a } P_{A|E}(a|e)
\Big(\sum_{a' \in f_X^{-1}(f_X(a))} P_{A|E}(a'|e)\Big)^{s}\bigg)^{\frac{1}{1+s}}
\\
&\le 
\sum_e P_E(e)  
\bigg(\sum_{a } P_{A|E}(a|e)
\Big(\rE_X \sum_{a' \in f_X^{-1}(f_X(a))} P_{A|E}(a'|e)\Big)^{s} \bigg)^{\frac{1}{1+s}}\\
&\le 
\sum_e P_E(e)  
\bigg(\sum_{a } P_{A|E}(a|e)
\Big(P_{A|E}(a|e)  + \frac{\epsilon}{\sM} \Big)^{s}
\bigg)^{\frac{1}{1+s}}\\
&\le 
\sum_e P_E(e)  
\bigg(\sum_{a } P_{A|E}(a|e)
\Big( P_{A|E}(a|e)^s  + \Big(\frac{\epsilon}{\sM} \Big)^{s}\Big)
\bigg)^{\frac{1}{1+s}}\\
&= 
\sum_e P_E(e)  
\bigg(
\Big(\frac{\epsilon}{\sM} \Big)^{s} +
\sum_{a } P_{A|E}(a|e)^{1+s} 
\bigg)^{\frac{1}{1+s}}\\
&\le 
\sum_e P_E(e)  
\bigg(
\Big(\frac{\epsilon}{\sM} \Big)^{\frac{s}{1+s}} +
\Big(\sum_{a } P_{A|E}(a|e)^{1+s} \Big)^{\frac{1}{1+s}}
\bigg)\\
&=  
\Big(\frac{\epsilon}{\sM} \Big)^{\frac{s}{1+s}} +
\sum_e P_E(e)  
\Big(\sum_{a } P_{A|E}(a|e)^{1+s} \Big)^{\frac{1}{1+s}} \\
&= 
 \frac{\epsilon^{\frac{s}{1+s}}}{\sM^{\frac{s}{1+s}}}
+ \rme^{-\frac{s}{1+s} H_{1+s}^{\uparrow}(A|E|P_{AE})} .
\end{align}
Combining this with the relation between $H_{1+s}^\uparrow$ and $C_{1+s}^\uparrow$ in \eqref{eqn:relate_C_H_g}, we obtain \eqref{eqn:os_direct_G_pluss}.
\end{proof}

\subsection{Proof of \eqref{eqn:os_direct_minuss} }

\begin{proof}
For any $a$, we have
\begin{align}
&\rE_X \sum_{a' \in f_X^{-1}(f_X(a))} P_{A|E}(a'|e) \nn\\*
&\le 
P_{A|E}(a|e)  + \frac{\epsilon}{\sM}\sum_{a' \neq a } P_{A|E}(a'|e) \\
&\le 
P_{A|E}(a|e)  + \frac{\epsilon}{\sM} \\
&\le 
2 \max \bigg\{P_{A|E}(a|e) , \frac{\epsilon}{\sM} \bigg\}.\label{10-14-7}
\end{align}
First we observe that when $P_{A|E}(a|e) \le \frac{c}{\sM}$, we have
\begin{align}
\rE_X \sum_{a' \in f_X^{-1}(f_X(a))} P_{A|E}(a'|e) 
\le 
P_{A|E}(a|e)  + \frac{\epsilon}{\sM} 
\le \frac{c+\epsilon}{\sM}. \label{10-14-6}
\end{align}

Now  we have
\begin{align}
& \rme^{-s C_{1-s}(f_X(A)|EX|P_{AE} \times P_X)}\nn\\*
&= \frac{1}{\sM^s}
 \rme^{s H_{1-s}(f_X(A)|EX|P_{AE}\times P_X)} \label{eqn:use_convexity0} \\
&\ge 
\frac{1}{\sM^s} \sum_e P_E(e)  
\sum_{a } P_{A|E}(a|e)
\Big(\rE_X \sum_{a' \in f_X^{-1}(f_X(a))} P_{A|E}(a'|e)\Big)^{-s} \label{eqn:use_convexity}
\\
 &\ge
\sum_{a,e } P_{AE}(a,e)
\big(2 \max \{M P_{A|E}(a|e)  , \epsilon\} \big)^{-s} \label{eqn:use_pr}
\\
& \ge  2^{-s}
\sum_{a,e} P_{AE}(a,e) 
\min \{P_{A|E}(a|e)^{-s} \sM^{-s}, \epsilon^{-s} \}\\
&=
2^{-s}
\sum_{a,e: P_{A|E}(a|e) \ge \epsilon \sM^{-1}} P_{AE}(a,e) 
P_{A|E}(a|e)^{-s} \sM^{-s}\nn\\*
&\qquad+
2^{-s}
\sum_{a,e: P_{A|E}(a|e) < \epsilon \sM^{-1}} P_{AE}(a,e) 
\epsilon^{-s},  \label{eqn:use_pr1}
\end{align}
where in \eqref{eqn:use_convexity} we used the convexity of $x\mapsto x^{-s}$ where $s\in [0,1]$ and $x\ge 0$ and in \eqref{eqn:use_pr}, we used~\eqref{10-14-7}, Thus, we obtain \eqref{eqn:os_direct_minuss}. \end{proof}

\subsection{Proof of \eqref{eqn:os_direct_G_minuss} }
\begin{proof}
Using \eqref{10-14-7} and the convexity of $a\mapsto a^{\frac{1}{1-s}}$ we have
\begin{align}
&\rme^{\frac{s}{1-s} H_{1-s}^{\uparrow}(f_X(A)|EX|P_{AE} \times P_X)} \nn\\*
&= 
\rE_X \sum_e P_E(e)  
\bigg(\sum_{i=1}^M \Big(\sum_{a \in f_X^{-1}(i)} P_{A|E}(a|e)\Big)^{1-s} \bigg)^{\frac{1}{1-s}}
\\
&= 
\rE_X \sum_e P_E(e)  
\bigg(\sum_{a } P_{A|E}(a|e)
\Big(\sum_{a' \in f_X^{-1}(f_X(a))} P_{A|E}(a'|e)\Big)^{-s} \bigg)^{\frac{1}{1-s}}
\\
&\ge 
\sum_e P_E(e)  
\bigg(\sum_{a } P_{A|E}(a|e)
\Big(\rE_X \sum_{a' \in f_X^{-1}(f_X(a))} P_{A|E}(a'|e)\Big)^{-s} \bigg)^{\frac{1}{1-s}}\\
&\ge 
\sum_e P_E(e)  
\bigg(\sum_{a } P_{A|E}(a|e)
\Big(P_{A|E}(a|e)  + \frac{\epsilon}{\sM} \Big)^{-s}
\bigg)^{\frac{1}{1-s}}\\
&\ge 
\sum_e P_E(e)  
\bigg(\sum_{a } P_{A|E}(a|e)
\Big(2 \max \Big\{ P_{A|E}(a|e) , \frac{\epsilon}{\sM} \Big\} \Big)^{-s}
\bigg)^{\frac{1}{1-s}}\\
&= 
2^{-\frac{s}{1-s}}
\sum_e P_E(e)  
\bigg(\sum_{a } P_{A|E}(a|e)
\min \Big\{ P_{A|E}(a|e)^{-s} , \frac{\epsilon^{-s}}{\sM^{-s}} \Big\} 
\bigg)^{\frac{1}{1-s}}\\
 &=
2^{-\frac{s}{1-s}} \sum_e P_E(e)  
\bigg(
\sum_{a: P_{A|E}(a|e) \ge \frac{\epsilon}{\sM} } 
P_{A|E}(a|e)^{1-s}
+
\frac{\epsilon^{-s}}{\sM^{-s}} \sum_{a: P_{A|E}(a|e) < \frac{\epsilon}{\sM} } 
P_{A|E}(a|e)
\bigg)^{\frac{1}{1-s}}\\
&\ge 
2^{-\frac{s}{1-s}} \sum_e P_E(e)  
\bigg(
\Big(\sum_{a: P_{A|E}(a|e) \ge \frac{\epsilon}{\sM} } 
P_{A|E}(a|e)\Big)^{\frac{1}{1-s}}
+
\Big(\frac{\epsilon^{-s}}{\sM^{-s}} \sum_{a: P_{A|E}(a|e) < \frac{\epsilon}{\sM} } 
P_{A|E}(a|e)^{1-s}
\Big)^{\frac{1}{1-s}}\bigg)\\
&\ge 
2^{-\frac{s}{1-s}} \sum_e P_E(e)  
\bigg(
\Big(\sum_{a: P_{A|E}(a|e) \ge \frac{\epsilon}{\sM} } 
P_{A|E}(a|e)^{1-s}\Big)^{\frac{1}{1-s}}
+
\Big(\frac{\epsilon^{-s}}{\sM^{-s}} \sum_{a: P_{A|E}(a|e) < \frac{\epsilon}{\sM} } 
P_{A|E}(a|e)
\Big)^{\frac{1}{1-s}}\bigg)\\
&= 
2^{-\frac{s}{1-s}} \sum_e P_E(e)  
\bigg(\sum_{a: P_{A|E}(a|e) \ge \frac{\epsilon}{\sM} } 
P_{A|E}(a|e)^{1-s}\bigg)^{\frac{1}{1-s}}\nn\\*
&\qquad+
2^{-\frac{s}{1-s}}
\frac{\sM^{\frac{s}{1-s}}}{\epsilon^{\frac{s}{1-s}}}
\sum_e P_E(e)  
\bigg(\sum_{a: P_{A|E}(a|e) < \frac{\epsilon}{\sM} } 
P_{A|E}(a|e)
\bigg)^{\frac{1}{1-s}}.
\end{align}
Thus we obtain \eqref{eqn:os_direct_G_minuss}. \end{proof}
\section{Proof of Lemma~\ref{lem:os_d_2}}
\subsection{Proof of \eqref{eqn:os_direct_exp} }

\begin{proof}
Since 
$(1+x)^{\frac{1}{1+s}}\le 1+ \frac{1}{1+s}x$, and $x\mapsto x^{\frac{1}{1+s}}$ is concave for $s\in [0,1]$, 
we have
\begin{align}
& \sM^{\frac{s}{1+s}} \rme^{-\frac{s}{1+s} H_{1+s}^{\uparrow}(f_X(A)|EX|P_{AE}\times P_X)}\nn\\*
&= 
\sM^{\frac{s}{1+s}} 
\rE_X \sum_e P_E(e)  
\bigg(\sum_{i=1}^{\sM} \Big(\sum_{a \in f_X^{-1}(i)} P_{A|E}(a|e)\Big)^{1+s}\bigg)^{\frac{1}{1+s}}
\\
&= 
\sM^{\frac{s}{1+s}} 
\rE_X \sum_e P_E(e)  
\bigg(\sum_{a } P_{A|E}(a|e)
\Big(\sum_{a' \in f_X^{-1}(f_X(a))} P_{A|E}(a'|e)\Big)^{s} \bigg)^{\frac{1}{1+s}}
\\
&\le 
\sM^{\frac{s}{1+s}} 
\sum_e P_E(e)  
\bigg(\sum_{a } P_{A|E}(a|e)
\Big(\rE_X \sum_{a' \in f_X^{-1}(f_X(a))} P_{A|E}(a'|e)\Big)^{s} \bigg)^{\frac{1}{1+s}}\\
&\le 
\sM^{\frac{s}{1+s}} 
\sum_e P_E(e)  
\bigg(\sum_{a } P_{A|E}(a|e)
\Big(P_{A|E}(a|e)  + \frac{1}{\sM} \Big)^{s}
\bigg)^{\frac{1}{1+s}}\\
&\le 
\sM^{\frac{s}{1+s}} 
\sum_e P_E(e)  
\bigg(\sum_{a } P_{A|E}(a|e)
\Big( P_{A|E}(a|e)^s  + \frac{1}{\sM^s} \Big)
\bigg)^{\frac{1}{1+s}}\\
&= 
\sM^{\frac{s}{1+s}} 
\sum_e P_E(e)  
\bigg(
\frac{1}{\sM^s} +
\sum_{a } P_{A|E}(a|e)^{1+s} 
\bigg)^{\frac{1}{1+s}}\\
&= 
\sum_e P_E(e)  
\bigg(
1+
\sM^{s} 
\sum_{a } P_{A|E}(a|e)^{1+s} 
\bigg)^{\frac{1}{1+s}}\\
 &\le
\sum_e P_E(e)  
\bigg(
1+
\frac{1}{1+s} \sM^{s} 
\sum_{a } P_{A|E}(a|e)^{1+s} 
\bigg)\\
&= 
1+
\frac{1}{1+s} \sM^{s}
\sum_e P_E(e)  
\sum_{a } P_{A|E}(a|e)^{1+s} 
\\
&= 
1+\frac{1}{1+s} \sM^{s} \rme^{-s H_{1+s}(A|E|P_{AE})}.
\end{align}
Using \eqref{eqn:relate_C_H2}, we obtain \eqref{eqn:os_direct_exp}.
\end{proof}

\section{Proof of Lemma~\ref{lem:os_d_3}}
\subsection{Proof of \eqref{eqn:one-shot-dir-second-1}}
\begin{proof}
Using \eqref{10-14-6}, we have
\begin{align}
& \rme^{-s C_{1-s}(f_X(A)|EX|P_{AE} \times P_X)}\nn\\*
&= \frac{1}{\sM^s}
 \rme^{s H_{1-s}(f_X(A)|EX|P_{AE}\times P_X)} \\
 &=
\frac{1}{\sM^s} \rE_X \sum_e P_E(e)  
\sum_{i=1}^M \Big(\sum_{a \in f_X^{-1}(i)} P_{A|E}(a|e)\Big)^{1-s}
\\
&= 
\frac{1}{\sM^s} \rE_X \sum_e P_E(e)  
\sum_{a } P_{A|E}(a|e)
\Big(\sum_{a' \in f_X^{-1}(f_X(a))} P_{A|E}(a'|e)\Big)^{-s}
\\
&\ge 
\frac{1}{\sM^s} \sum_e P_E(e)  
\sum_{a } P_{A|E}(a|e)
\Big(\rE_X \sum_{a' \in f_X^{-1}(f_X(a))} P_{A|E}(a'|e)\Big)^{-s}
\\
&\ge
\frac{1}{\sM^s} \sum_e P_E(e)  
\sum_{a: P_{A|E}(a|e) \le \frac{c}{\sM}  } P_{A|E}(a|e)
\Big(\rE_X \sum_{a' \in f_X^{-1}(f_X(a))} P_{A|E}(a'|e)\Big)^{-s}
\\
&\ge
\frac{1}{\sM^s} \sum_e P_E(e)  
\sum_{a: P_{A|E}(a|e) \le \frac{c}{\sM}  } P_{A|E}(a|e)
\Big(\frac{c+\epsilon}{\sM}\Big)^{-s}
\\
&=
P_{AE} \Big\{(a,e): P_{A|E}(a|e) \le \frac{c}{\sM}  \Big\}
\Big(\frac{1}{c+\epsilon}\Big)^s .
\end{align}
We obtain \eqref{eqn:one-shot-dir-second-1} as desired.\end{proof}
\subsection{Proof of \eqref{eqn:one-shot-dir-second-2}}
\begin{proof}
Using \eqref{10-14-6}, we have
\begin{align}
& \rme^{\frac{s}{1-s} H_{1-s}^{\uparrow}(f_X(A)|EX|P_{AE}\times P_X)}\nn\\*
&= 
\rE_X \sum_e P_E(e)  
\bigg(\sum_{i=1}^M \Big(\sum_{a \in f_X^{-1}(i)} P_{A|E}(a|e) \Big)^{1-s}  \bigg)^{\frac{1}{1-s}}
\\
 &=
\rE_X \sum_e P_E(e)  
\bigg(\sum_{a } P_{A|E}(a|e)
\Big(\sum_{a' \in f_X^{-1}(f_X(a))} P_{A|E}(a'|e)\Big)^{-s}\bigg)^{\frac{1}{1-s}}
\\
 &\ge
\sum_e P_E(e)  
\bigg(\sum_{a } P_{A|E}(a|e)
\Big(\rE_X \sum_{a' \in f_X^{-1}(f_X(a))} P_{A|E}(a'|e)\Big)^{-s} \bigg)^{\frac{1}{1-s}}\\
 &\ge
\sum_e P_E(e)  
\bigg(\sum_{a: P_{A|E}(a|e)\le \frac{c}{\sM} } P_{A|E}(a|e)
\Big(\rE_X \sum_{a' \in f_X^{-1}(f_X(a))} P_{A|E}(a'|e)\Big)^{-s} \bigg)^{\frac{1}{1-s}}\\
 &\ge
\sum_e P_E(e)  
\bigg(\sum_{a: P_{A|E}(a|e)\le \frac{c}{\sM} } P_{A|E}(a|e)
\Big(\frac{c+\epsilon}{\sM}\Big)^{-s} \bigg)^{\frac{1}{1-s}}\\
&= 
\Big(\frac{c+\epsilon}{\sM}\Big)^{-\frac{s}{1-s}} 
\sum_e P_E(e)  
\bigg(\sum_{a: P_{A|E}(a|e)\le \frac{c}{\sM} } P_{A|E}(a|e)\bigg)^{\frac{1}{1-s}}.
\end{align}
By combining with \eqref{eqn:relate_C_H_g}, we  obtain \eqref{eqn:one-shot-dir-second-2}. \end{proof}

\section{Proof of Lemma~\ref{lem:os_c_1}}
\subsection{Proof of \eqref{eqn:os_conv_eq} } \label{sec:prf_eqn:os_conv_eq}
\begin{proof}
Define the functions
\begin{align}
g_1(x,y)&:=x+y- 2 x^{1-s} y^s \\ 
g_2(x,y)&:=x- 2 x^{1-s} y^s .
\end{align}
Then, we can show that
\begin{align}
\min_{y}g_1(x,y)= x
(1-2^{\frac{1}{1-s}} s^{\frac{1}{1-s}-1} (1-s) ),\label{28-2}
\end{align}
which is attained when $y= x(2 s)^{\frac{1}{1-s}}$. 
We also define
\begin{align}
h_{e,1}(m)
:= & \sum_{a\in f^{-1}(m): P_{A|E}(a|e) < \frac{c}{M }} P_{A|E}(a|e) \\
h_{e,2}(m)
:= & \sum_{a\in f^{-1}(m): P_{A|E}(a|e) \ge \frac{c}{M }} P_{A|E}(a|e) .
\end{align}
Hence, 
\begin{align}
P_{f(A)|E}(m|e)^{1-s} 
=
(h_{e,1}(m)+ h_{e,2}(m))^{1-s}
\le 
h_{e,1}(m)^{1-s}+ h_{e,2}(m)^{1-s},
\end{align}
which implies from the definitions of $g_1$ and $g_2$ that 
\begin{align}
g_1\Big(P_{f(A)|E}(m|e), \frac{1}{M } \Big)
\ge
g_2\Big( h_{e,2}(m), \frac{1}{M } \Big)+
g_1\Big( h_{e,1}(m), \frac{1}{M } \Big).\label{28-1}
\end{align}
Also, we have
\begin{align}
h_{e,2}(m)^{1-s} \le
\sum_{a \in f^{-1}(m)}
P_{A|E}(a|e)^{1-s}. \label{28-3}
\end{align}
Thus,
\begin{align}
&1- 2 \rme^{-s C_{1-s}(f(A)|E|P_{AE})} \nn\\*
&= \sum_e P_E(e) \sum_{m} g_1\Big(P_{f(A)|E}(m|e), \frac{1}{M } \Big) \\
&\ge
\sum_e P_E(e)  \bigg(
\sum_{m}
g_2 \Big( h_{e,2}(m), \frac{1}{M } \Big) 
+
\sum_{m} 
g_1\Big( h_{e,1}(m), \frac{1}{M } \Big)  \bigg)
  \label{eqn:g1g2}\\
&\ge
\sum_e P_E(e) 
\bigg(
\sum_{m}
(h_{e,2}(m) -2 h_{e,2}(m)^{1-s} M ^s ) + \sum_{m}
h_{e,1}(m)(1-2^{\frac{1}{1-s}} s^{\frac{s}{1-s} } (1-s) )\bigg)
\label{10-21-1} \\
 &=
1-
2c^{-s}  \sum_e P_E(e) 
\sum_{m } h_{e,2}(m)^{1-s} M^s 
- 2^{\frac{1}{1-s}} s^{\frac{s}{1-s} } (1-s) 
\sum_{m } h_{e,1}(m) \\
& \ge
1-
2c^{-s}  \sum_e P_E(e) 
\sum_{a: P_{A|E}(a|e) \ge \frac{c}{M} } P_{A|E}(a|e)^{1-s} M^s 
- 2^{\frac{1}{1-s}} s^{\frac{s}{1-s} } (1-s) 
\sum_{a: P_{A|E}(a|e) < \frac{c}{M} } P_{A|E}(a|e) \label{eqn:g1g22}\\
&= 
1-
2 c^{-s} \sum_e P_E(e) 
\sum_{a: P_{A|E}(a|e) \ge \frac{c}{M} } P_{A|E}(a|e)^{1-s} M^s \nn\\*
&\qquad - 2 \cdot 2^{\frac{s}{1-s}} s^{\frac{s}{1-s} } (1-s) 
P_{AE}  \Big\{(a,e) : P_{A|E}(a|e) < \frac{c}{M}  \Big\} ,
\label{10-21-2}
\end{align}
where \eqref{eqn:g1g2},  \eqref{10-21-1}, and  \eqref{eqn:g1g22} follow  from \eqref{28-1}, \eqref{28-2}, and
\eqref{28-3}  respectively.
Hence, we obtain \eqref{eqn:os_conv_eq}.
\end{proof}

\subsection{Proof of \eqref{eqn:os_conv_eqG}} \label{sec:prf_eqn:os_conv_eq_B}
We first state a useful  and easy lemma:
\begin{lemma} \label{lem:convex}
Let $x,y\ge 0 $ and $t\ge 1$. Then we have 
\begin{equation}
(x+y)^t\le 2^{t-1} (x^t +y^t).
\end{equation}
\end{lemma}
\begin{proof}
It is clear that $a\mapsto a^t$ is convex for $a\ge 0$. Thus,
\begin{align}
(x+y)^t = 2^t \Big(\frac{x}{2} + \frac{y}{2}\Big)^t \le 2^t\Big(\frac{x^t}{2} + \frac{y^t}{2}\Big)=2^{t-1} (x^t+y^t),
\end{align}
which proves the claim.
\end{proof}
\begin{proof}[Proof of \eqref{eqn:os_conv_eqG}]
Using the previously proved bound in \eqref{eqn:os_conv_eq} with $|\calE|=1$, we obtain
\begin{align}
\rme^{-sC_{1-s}^{\uparrow}(f(A)|P_{A|E=e})}
&\le  c^{-s} \sum_{a:P_{A|E}(a|e) \ge\frac{c}{M}} P_{A|E}(a|e)^{1-s}M^{-s} \nn\\*
&\qquad\qquad+ 2^{\frac{s}{1-s}}s^{\frac{s}{1-s}} (1-s) P_{AE} \Big\{ (a,e) : P_{A|E}(a|e)\le\frac{c}{M}\Big\} \label{eqn:Eequals1} .
\end{align}
Taking average over $P_E$ and using the bound in \eqref{eqn:Eequals1}, we have 
\begin{align}
& \rme^{-\frac{s}{1-s} C_{1-s}^{\uparrow}(f(A)|E|P_{AE})} \nn\\*
&=
\sum_{e}P_E(e)
\left(\rme^{-sC_{1-s}^\uparrow(f(A)|P_{A|E=e})} \right)^{\frac{1}{1-s}} \\
&\le \sum_e P_E(e) 
\bigg[ c^{-s} \sum_{a:P_{A|E}(a|e) \ge\frac{c}{M}} P_{A|E}(a|e)^{1-s}M^{-s} \nn\\*
&\qquad\qquad+ 2^{\frac{s}{1-s}}s^{\frac{s}{1-s}} (1-s) P_{AE} \Big\{ (a,e) : P_{A|E}(a|e)\le\frac{c}{M}\Big\} \bigg]^{\frac{1}{1-s}} \label{eqn:parent}\\
&\le 2^{\frac{s}{1-s}} \sum_e P_E(e) \bigg[ 
\Big( c^{-s} \sum_{a:P_{A|E}(a|e) \ge\frac{c}{M}} P_{A|E}(a|e)^{1-s}M^{-s}  \Big)^{\frac{1}{1-s}} \nn\\*
&\qquad\qquad+ \Big( 2^{\frac{s}{1-s}}s^{\frac{s}{1-s}} (1-s) P_{AE} \Big\{ (a,e) : P_{A|E}(a|e)\le\frac{c}{M}\Big\} \Big)^{\frac{1}{1-s}}\bigg]
\end{align}
where in the last step, we applied Lemma~\ref{lem:convex} with $t=\frac{1}{1-s}\ge 1$ to the term in parentheses in \eqref{eqn:parent}. Thus we obtain~\eqref{eqn:os_conv_eqG}.
\end{proof}

\section{Proof of Lemma~\ref{lem:os_c_2}} \label{app:prf_lem:os_c_2}

The  inequalities  in Lemma \ref{lem:os_c_2} can be shown by 
the information processing inequality  for R\'enyi divergence in~\eqref{eqn:dpi_rd}. 

\subsection{Proofs of   \eqref{10-20-1b} and \eqref{10-20-1}} \label{sec:prf_lemosc21}
\begin{proof}
For every $e \in {\cal E}$, 
define the  function $f_e:{\cal A} \to {\cal M}$ to be 
\begin{equation}
f_e:=\argmin_{f}   D_{1-s}(P_{f(A)|E=e}\| P_{\mix,\calM }) . \label{eqn:fe}
\end{equation}
 We start with a claim that will be proved at the end of this subsection. 
\begin{lemma} \label{lem:one}
For every $a$ such that  $P_{A|E=e}(a)\ge \frac{1}{\sM} $, 
we have
$| f_e^{-1}(f_e(a))|=1$. 
\end{lemma}

\begin{figure}
\centering
\begin{overpic}[width=10cm,height=7cm]{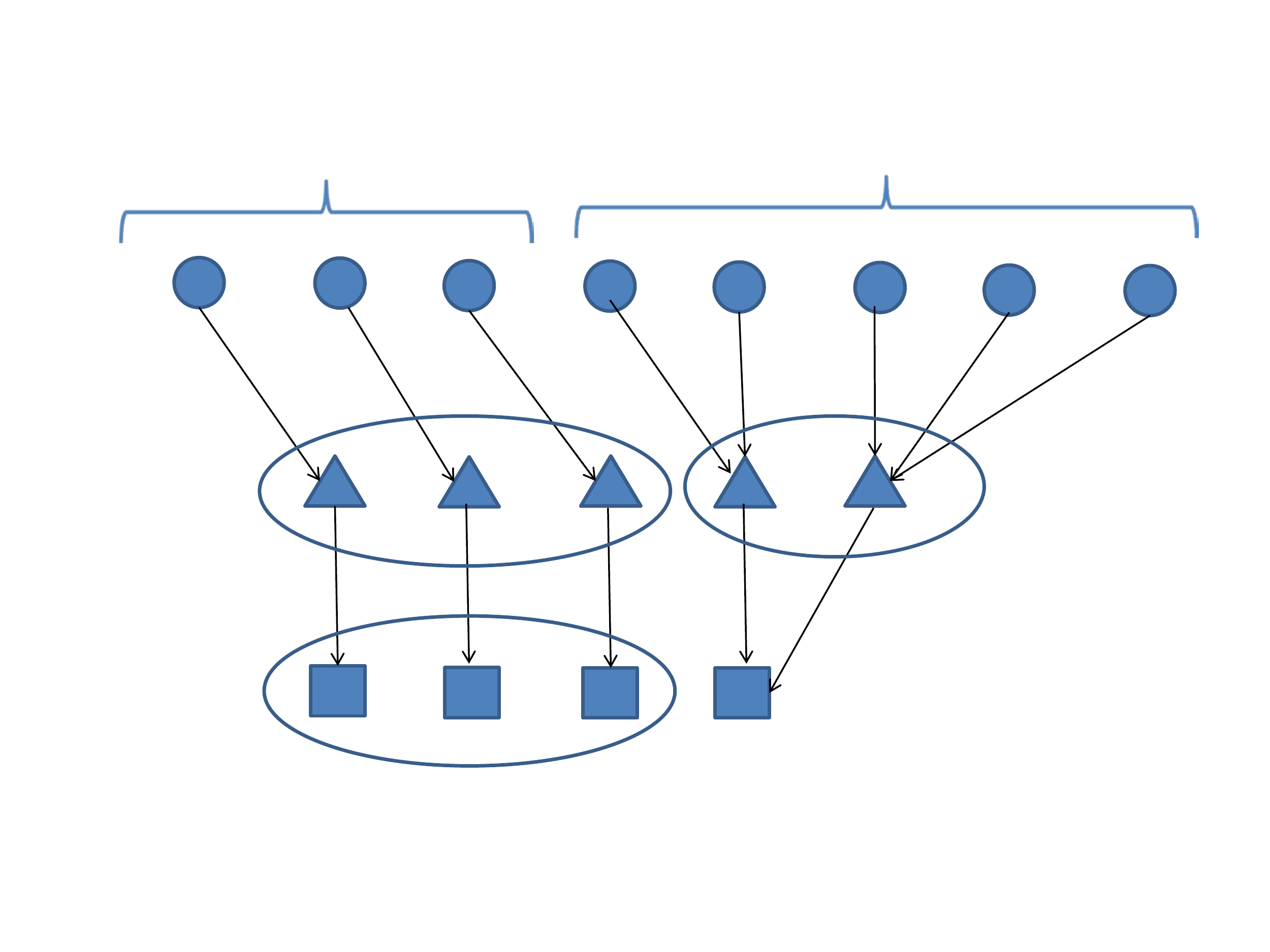}%
\put(15,60){$\{a:P_{A}(a)\ge\frac{c}{M}\}$} 
\put(60,60){$\{a:P_{A}(a)<\frac{c}{M}\}$} 
\put(2,47){$a\in\calA$}
\put(-5,33){$i\in\calM_1\cup\calM_2$}
\put(-3,18){$i\in\calM_1\cup\{0\}$}
\put(15,41){$f$}
\put(23,26){$g$}
\put(30,10){$\calM_1$}
\put(57,13){$0$}
\put(33,40.5){$\calM_1$}
\put(78,35){$\calM_2$}
\end{overpic}\vspace{-.4in}
\caption{Illustration of the steps in \eqref{eqn:dpi} to \eqref{eqn:gg2} where the dependences  on $e\in\calE$ are suppressed.}
\label{fig:one}
\end{figure}
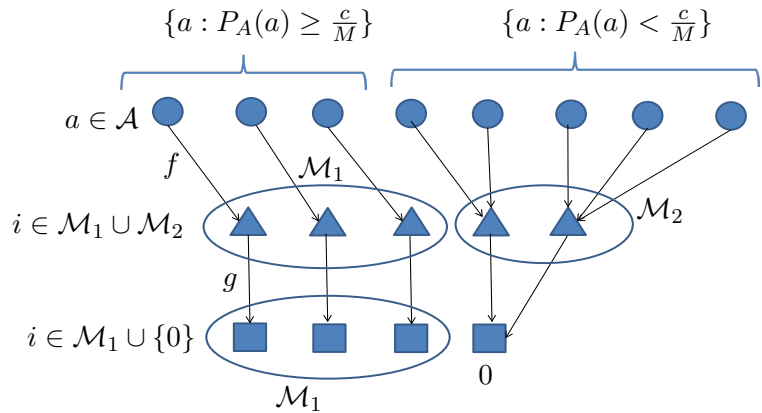

Now, we partition the set ${\cal M}$ into two subsets as follows:
${\cal M}_1:= f_e ( \{a:  P_{A|E=e}(a)\ge \frac{c}{\sM} \})$
and ${\cal M}_2:= {\cal M}_1^c$.  See Fig.~\ref{fig:one}. 
Next, we define the map $g:{\cal M}\to {\cal M}_1\cup \{0\}$ as 
\begin{align}
g(i):= \left\{
\begin{array}{ll}
i & \hbox{if } i \in {\cal M}_1 \\
0 & \hbox{if } i \in {\cal M}_2 
\end{array}
\right. . \label{eqn:def_g}
\end{align}
Note that the map $g$ depends on $e\in\calE$ (through $\calM_1$ and $\calM_2 $) but we suppress this dependence for brevity. 
Let $P_{f_e(A)|E=e}\circ g^{-1}$ be the ``output distribution'' taking values on $\calM_1\cup\{0\}$ induced by the ``input distribution'' $P_{f_e(A)|E=e}$  and the  ``deterministic channel'' $g$, i.e., for every $i\in\calM_1\cup\{0\}$, $(P_{f_e(A)|E=e }\circ g^{-1}) (i) := P_{f_e(A)|E=e } (g^{-1}(i)) =\sum_{j \in\calM : g(i) = j} P_{f_e(A)|E=e }(j) $.  We also use the notation  $P_{\mix,\calM}\circ g^{-1}$ analogously. Due to the information processing inequality for the R\'enyi divergence in~\eqref{eqn:dpi_rd}, we obtain
\begin{align}
& \rme^{-s D_{1-s}(P_{f_e(A)|E=e}\| P_{\mix, \calM })} \nn\\*
&\le 
\rme^{- sD_{1-s}(P_{f_e(A)|E=e}\circ g^{-1}  \| P_{\mix,\calM } \circ g^{-1})} \label{eqn:dpi} \\ 
&=\sum_{i \in \{0\}\cup\calM_1} \big( (P_{f_e(A)|E=e} \circ g^{-1} ) (i)\big)^{1-s} \big( (P_{\mix,\calM}\circ g^{-1})(i) \big)^{s}\label{eqn:gg3}\\
&=\big((  P_{f_e(A)|E=e}\circ g^{-1} )(0)\big)^{1-s}  \big( (P_{\mix,\calM}\circ g^{-1})(0) \big)^{s}+\sum_{i\in\calM_1}\big(( P_{f_e(A)|E=e}\circ g^{-1}  )(i)  \big)^{1-s}\Big(\frac{1}{M}\Big)^{ s} \label{eqn:gg0}\\\
&\le    P_{A|E=e} \Big\{ a : P_{A|E} (a|e) < \frac{c}{\sM} \Big\}^{1-s}+
\sum_{i\in\calM_1}\big(( P_{f_e(A)|E=e}\circ g^{-1}  )(i)  \big)^{1-s}\Big(\frac{1}{M}\Big)^{ s}  \label{eqn:gg1}\\
&= 
 P_{A|E=e} \Big\{ a : P_{A|E} (a|e) < \frac{c}{\sM} \Big\}^{1-s}+\sum_{a: P_{A|E=e}(a)\ge \frac{c}{\sM}}
 P_{A|E}(a|e)^{1-s} \sM^{-s}
, \label{eqn:gg2}
\end{align}
where in \eqref{eqn:gg3}  we used the definition of $D_{1-s}$, in \eqref{eqn:gg0}  we split the resulting sum into  $\{0\}$ and $\calM_1$, in \eqref{eqn:gg1} we upper bounded $(P_{\mix,\calM}\circ g^{-1})(0) $ by $1$ and we noted that  all the symbols $i\in \calM_2 = f_e(\{a : P_{A|E=e}(a)<\frac{c}{M}\})$  are merged into the  symbol $0$ under $g$ and  finally in \eqref{eqn:gg2}, we used the fact that $g$ and $f_e$ are one-to-one restricted to $\calM_1$ and $\{a:P_{A|E=e}(a)\ge \frac{c}{M}\}$ respectively (Lemma \ref{lem:one} and the fact that $c\ge 1$).  See Fig.~\ref{fig:one} for an illustration of these steps.
Taking the average  of \eqref{eqn:gg2} over $P_E(e)$, we obtain \eqref{10-20-1b}. 

Furthermore, we have
\begin{align}
& \sum_{a: P_{A|E=e}(a)\ge \frac{c}{\sM}}
 P_{A|E}(a|e)^{1-s} \sM^{-s} \nn\\*
 & \le
\sum_{a: P_{A|E=e}(a)\ge \frac{c}{\sM}}
 P_{A|E}(a|e) \Big(\frac{c}{\sM}\Big)^{-s} \sM^{-s}\\
& =
P_{A|E=e}   \Big\{a : P_{A|E=e}(a)\ge \frac{c}{\sM}  \Big\}  c^{-s} .   \label{eqn:first_term}
\end{align}
Substituting \eqref{eqn:first_term} into the second term in \eqref{eqn:gg2} and then taking the average over $P_E(e)$, we obtain \eqref{10-20-1}. 
\end{proof}

It remains to prove Lemma \ref{lem:one}.

\begin{proof}[Proof of Lemma \ref{lem:one}]
It suffices to consider the case $|\calE|=1$. Dropping the dependences on $e$, we denote $f_e$ as $f$ and $P_{A|E=e}$ as $P_A$ in the sequel.

We proceed by contradiction. The essential idea is that an optimal $f$ (given by \eqref{eqn:fe}) must induce a distribution $P_{f(A)}$ on $\calM$ that is ``as close to uniform as possible'' since we are minimizing $D_{1-s}(P_{f(A) } \| P_{\mix,\calM})$. 

Formally, assume, to the contrary, that $|f^{-1}(f(a_1))|\ge 2$ for some $a_1\in\calA$ with  $P_A(a_1)\ge \frac{1}{M}$. Because   $|f^{-1}(f(a_1))|\ge 2$, there exists $a_2\ne a_1$ such that $f(a_1)=f(a_2)=j$ for some $j\in\calM$. Because $P_{A}(a_2)>0$, we have $\sum_{a\in f^{-1}(j)}P_A(a)>\frac{1}{M}$. This in turn implies that there exists $i\in\calM$ such that $\sum_{a\in f^{-1}(i)}P_A(a)<\frac{1}{M}$. Recall that  $f$ was designed to minimize 
\begin{equation}
D_{1-s}(P_{f(A) } \| P_{\mix,\calM})= -\frac{1}{s}\log\sum_{i\in\calM} \bigg(\sum_{a \in f^{-1}(i)} P_A(a) \bigg)^{1-s} \Big(\frac{1}{M}\Big)^s,
\end{equation}
or equivalently, to maximize $\sum_i \big(\sum_{a \in f^{-1}(i)} P_A(a) \big)^{1-s}$.  Now we create a new hash function 
\begin{equation}
\tilf(a) := \left\{  \begin{array}{cc}
f(a) & a\ne a     _2\\
i & a = a_2
\end{array} \right. .  \label{eqn:tilf}
\end{equation}
Let $u,v,u',v'$  be any four non-negative numbers such that $u+v=u'+v'=t$ and $p=u/t$, $p'=u'/t$ and $|p-1/2|<|p'-1/2|$. This means that $(p,1-p)$ is closer to the uniform Bernoulli distribution compared to  $(p',1-p')$. Then it is easy to check that  
\begin{equation}
(u')^{1-s}+(v')^{1-s}<u^{1-s}+v^{1-s}.\label{eqn:alpha_ine}
\end{equation}
 Now denoting $\frac{1}{M}+\delta_1: =\sum_{a  \in f^{-1}(j)} P_A(a)$,   $\frac{1}{M}-\delta_2:=\sum_{a \in f^{-1}(i)} P_A(a)$ for positive numbers $\delta_1$ and $\delta_2$ (positive by the above construction of the sets $f^{-1}(i)$ and $f^{-1}(j)$), and letting $q:= P_{A}(a_2)\le\delta_1$, $t:= (\frac{1}{M}+\delta_1)+(\frac{1}{M}-\delta_2) $, we find that $\frac{1}{t}( \frac{1}{M}+\delta_1-q , \frac{1}{M}-\delta_2+q)$ is  closer to the uniform Bernoulli distribution compared to  $\frac{1}{t}( \frac{1}{M}+\delta_1 , \frac{1}{M}-\delta_2 )$. Using   inequality \eqref{eqn:alpha_ine}, we find that 
\begin{align} 
&\bigg(\sum_{a  \in f^{-1}(j)} P_A(a) \bigg)^{1-s}  + \bigg(\sum_{a \in f^{-1}(i)} P_A(a) \bigg)^{1-s}\nn\\*
& \qquad < \bigg(\sum_{a  \in f^{-1}(j) \setminus \{a_2\} } P_A(a) \bigg)^{1-s}  + \bigg(\Big(\sum_{a \in f^{-1}(i)} P_A(a)\Big) +P_A(a_2)  \bigg)^{1-s} . \label{eqn:transfer}
\end{align}
Since $f^{-1}(k) = \tilf^{-1}(k)$ for all $k \notin\{ i,j\}$,  \eqref{eqn:transfer} implies that
\begin{equation}
\sum_i  \bigg(\sum_{a  \in f^{-1}(i)} P_A(a) \bigg)^{1-s} < \sum_i  \bigg(\sum_{a  \in \tilf^{-1}(i)} P_A(a) \bigg)^{1-s} ,
\end{equation}
contradicting the optimality of $f$.
\end{proof}

\subsection{Proof of \eqref{eqn:os_conv_expG} }\label{app:prf_eqn:os_conv_expG}

\begin{proof}
Here, we employ the following expression for $C_{1-s}^{\uparrow}(A|E|P_{AE})$:
\begin{align}
\rme^{-\frac{s}{1-s}C_{1-s}^{\uparrow}(A|E|P_{AE})}
= \frac{1}{|{\cal A}|^{\frac{s}{1-s}}}
\sum_{e} P_{E}(e) \left( \sum_{a} P_{A|E}(a|e)^{1-s} \right)^{\frac{1}{1-s}}.
\end{align}
To minimize 
$\rme^{-\frac{s}{1-s}C_{1-s}^{\uparrow}(f(A)|E|P_{AE})}$,
it is enough to minimize 
$\sum_{i \in \calM} (\sum_{a \in f^{-1}(i)}P_{A|E=e}( a) )^{1-s}$
for each $e$.

Fortunately,  the discussion in the proof in Appendix~\ref{sec:prf_lemosc21} (and, in particular, the bound~\eqref{10-20-1b}) shows that
this value is upper bounded by 
\begin{equation}
\sum_{a: P_{A|E=e}(a)\ge \frac{c}{\sM}}
 P_{A|E}(a|e)^{1-s} 
+ P_{A|E=e}\Big\{ a : P_{A|E} (a|e) < \frac{c}{\sM}\Big\}^{1-s} \sM^s.
\end{equation}
Thus,
\begin{align}
& \rme^{-\frac{s}{1-s} C_{1-s}^{\uparrow}(f(A)|E|P_{AE})} \\
&\le 
\frac{1}{\sM^{\frac{s}{1-s}}}
\sum_{e} P_{E}(e) \bigg( 
P_{A|E=e}  \Big\{a: P_{A|E=e}(a)\ge \frac{c}{\sM} \Big\}   c^{-s} \sM^s\nn\\*
&\qquad+ P_{A|E=e}\Big\{ a  : P_{A|E} (a|e) < \frac{c}{\sM}\Big\}^{1-s} \sM^s
\bigg)^{\frac{1}{1-s}} \\
 &=
\sum_{e} P_{E}(e) \bigg( 
P_{A|E=e}  \Big\{a: P_{A|E=e}(a)\ge \frac{c}{\sM} \Big\}  c^{-s} \nn\\*
&\qquad+ P_{A|E=e}\Big\{ a : P_{A|E} (a|e) < \frac{c}{\sM}\Big\}^{1-s} 
\bigg)^{\frac{1}{1-s}} .
\end{align}
Hence, we obtain \eqref{eqn:os_conv_expG} as desired.  
\end{proof}
\subsection{Proofs of \eqref{10-20-2b} and \eqref{10-20-2}}
\begin{proof}
The proofs of these bounds are  similar to those of   \eqref{10-20-1b} and \eqref{10-20-1} in Appendix~\ref{sec:prf_lemosc21} and thus are omitted. \end{proof}

\subsection{Proof of \eqref{10-20-2_205}   }
\begin{proof}
The proof  of this bound is   similar to that of \eqref{eqn:os_conv_expG}  in Appendix~\ref{app:prf_eqn:os_conv_expG} and is thus   omitted. \end{proof}
\section{Proof of Lemma~\ref{lem:os_c_3}}
\subsection{Proof of \eqref{eqn:os_conv_sec} }
\begin{proof}
The proof of this bound is   similar to the proof of \eqref{eqn:os_conv_eq}  which is presented in Appendix~\ref{sec:prf_eqn:os_conv_eq}. We provide the details here.  

When $P_{A|E}(a|e) < \frac{c}{M} $, 
we have
$ P_{A|E}(a|e)^{1-s} M^s
\le 
P_{A|E}(a|e) c^{-s}$.
Thus, starting from \eqref{10-21-2}, 
we have
\begin{align}
&1- 2 \rme^{-s C_{1-s}(f(A)|E|P_{AE})} \nn\\*
&\ge 
1-
2 c^{-s} \sum_e P_E(e) 
\sum_{a: P_{A|E}(a|e) \ge \frac{c}{M} } P_{A|E}(a|e)^{1-s} M^s\nn\\* 
&\qquad - 2 \cdot 2^{\frac{s}{1-s}} s^{\frac{s}{1-s} } (1-s) 
P_{AE}  \Big\{(a,e) : P_{A|E}(a|e) < \frac{c}{M}  \Big\} \\
& \ge
1 
-2 c^{-s}
P_{AE}\Big\{(a,e): P_{A|E}(a|e) \ge \frac{c}{M} \Big\}\nn\\* 
&\qquad-2 c \cdot 2^{\frac{s}{1-s}} s^{\frac{s}{1-s} } (1-s) 
P_{AE} \Big\{(a,e) :  P_{A|E}(a|e) < \frac{c}{M} \Big\} .
\end{align}
This completes the proof of  \eqref{eqn:os_conv_sec}.
\end{proof}


\subsection*{Acknowledgements}   
The authors would like to acknowledge the Associate Editor (Prof.\ Aaron B.\ Wagner) and the anonymous reviewers for their extensive and useful comments during the revision process.

MH is partially supported by a MEXT Grant-in-Aid for Scientific Research (A) No.\ 23246071. 
MH is also partially supported by the National Institute of Information and Communication Technology (NICT), Japan.
The Centre for Quantum Technologies is funded by the Singapore Ministry of Education and the National Research Foundation as part of the Research Centres of Excellence programme.

VYFT is partially supported an NUS Young Investigator Award (R-263-000-B37-133) and a Singapore Ministry of Education Tier 2 grant ``Network Communication with Synchronization Errors: Fundamental Limits and Codes'' (R-263-000-B61-112).

%
%

\bibliographystyle{unsrt}
\bibliography{isitbib}

\end{document}